\PassOptionsToPackage{unicode}{hyperref}
\PassOptionsToPackage{hyphens}{url}
\PassOptionsToPackage{dvipsnames,svgnames,x11names}{xcolor}
\documentclass[
  12pt]{article}

\usepackage{amsmath,amssymb}
\usepackage{iftex}
\ifPDFTeX
  \usepackage[T1]{fontenc}
  \usepackage[utf8]{inputenc}
  \usepackage{textcomp} 
\else 
  \usepackage{unicode-math}
  \defaultfontfeatures{Scale=MatchLowercase}
  \defaultfontfeatures[\rmfamily]{Ligatures=TeX,Scale=1}
\fi
\usepackage{lmodern}
\ifPDFTeX\else  
\fi
\IfFileExists{upquote.sty}{\usepackage{upquote}}{}
\IfFileExists{microtype.sty}{
  \usepackage[]{microtype}
  \UseMicrotypeSet[protrusion]{basicmath} 
}{}
\makeatletter
\@ifundefined{KOMAClassName}{
  \IfFileExists{parskip.sty}{%
    \usepackage{parskip}
  }{
    \setlength{\parindent}{0pt}
    \setlength{\parskip}{6pt plus 2pt minus 1pt}}
}{
  \KOMAoptions{parskip=half}}
\makeatother
\usepackage{xcolor}
\makeatletter
\ifx\paragraph\undefined\else
  \let\oldparagraph\paragraph
  \renewcommand{\paragraph}{
    \@ifstar
      \xxxParagraphStar
      \xxxParagraphNoStar
  }
  \newcommand{\xxxParagraphStar}[1]{\oldparagraph*{#1}\mbox{}}
  \newcommand{\xxxParagraphNoStar}[1]{\oldparagraph{#1}\mbox{}}
\fi
\ifx\subparagraph\undefined\else
  \let\oldsubparagraph\subparagraph
  \renewcommand{\subparagraph}{
    \@ifstar
      \xxxSubParagraphStar
      \xxxSubParagraphNoStar
  }
  \newcommand{\xxxSubParagraphStar}[1]{\oldsubparagraph*{#1}\mbox{}}
  \newcommand{\xxxSubParagraphNoStar}[1]{\oldsubparagraph{#1}\mbox{}}
\fi
\makeatother

\usepackage{longtable,booktabs,array}
\usepackage{calc} 
\usepackage{etoolbox}
\makeatletter
\patchcmd\longtable{\par}{\if@noskipsec\mbox{}\fi\par}{}{}
\makeatother
\IfFileExists{footnotehyper.sty}{\usepackage{footnotehyper}}{\usepackage{footnote}}
\makesavenoteenv{longtable}
\usepackage{graphicx}
\makeatletter
\def\maxwidth{\ifdim\Gin@nat@width>\linewidth\linewidth\else\Gin@nat@width\fi}
\def\maxheight{\ifdim\Gin@nat@height>\textheight\textheight\else\Gin@nat@height\fi}
\makeatother
\setkeys{Gin}{width=\maxwidth,height=\maxheight,keepaspectratio}
\makeatletter
\def\fps@figure{htbp}
\makeatother

\makeatletter
\@ifpackageloaded{caption}{}{\usepackage{caption}}
\AtBeginDocument{%
\ifdefined\contentsname
  \renewcommand*\contentsname{Table of contents}
\else
  \newcommand\contentsname{Table of contents}
\fi
\ifdefined\listfigurename
  \renewcommand*\listfigurename{List of Figures}
\else
  \newcommand\listfigurename{List of Figures}
\fi
\ifdefined\listtablename
  \renewcommand*\listtablename{List of Tables}
\else
  \newcommand\listtablename{List of Tables}
\fi
\ifdefined\figurename
  \renewcommand*\figurename{Figure}
\else
  \newcommand\figurename{Figure}
\fi
\ifdefined\tablename
  \renewcommand*\tablename{Table}
\else
  \newcommand\tablename{Table}
\fi
}
\@ifpackageloaded{float}{}{\usepackage{float}}
\floatstyle{ruled}
\@ifundefined{c@chapter}{\newfloat{codelisting}{h}{lop}}{\newfloat{codelisting}{h}{lop}[chapter]}
\floatname{codelisting}{Listing}

\makeatother
\makeatletter
\@ifpackageloaded{caption}{}{\usepackage{caption}}
\@ifpackageloaded{subcaption}{}{\usepackage{subcaption}}
\makeatother

\ifLuaTeX
  \usepackage{selnolig}  
\fi
\usepackage[]{natbib}
\usepackage{bookmark}

\IfFileExists{xurl.sty}{\usepackage{xurl}}{} 
\hypersetup{
  pdftitle={Title},
  pdfauthor={Author 1; Author 2},
  pdfkeywords={3 to 6 keywords, that do not appear in the title},
  colorlinks=true,
  linkcolor={blue},
  filecolor={Maroon},
  citecolor={Blue},
  urlcolor={Blue},
  pdfcreator={LaTeX via pandoc}}

\newcommand{\anon}{1}

\usepackage{amsmath, mathrsfs, amssymb, amsthm, fullpage, graphicx, mathtools, bbm}
\usepackage[ruled,vlined]{algorithm2e}
\newtheorem{theorem}{Theorem}[section]
\newtheorem{lemma}{Lemma}[section]
\newtheorem{example}{Example}[section]
\newtheorem*{example*}{Example}

\newcommand\numberthis{\addtocounter{equation}{1}\tag{\theequation}}

\makeatletter
\def\BState{\State\hskip-\ALG@thistlm}
\makeatother

\newtheorem{proposition}{Proposition}[section]
\newtheorem{corollary}{Corollary}[section]
  {
      \theoremstyle{plain}
      \newtheorem{assumption}{Assumption}
  }

\numberwithin{equation}{section}

\usepackage{caption}
\usepackage{chngcntr}

\newcommand\independent{\protect\mathpalette{\protect\independenT}{\perp}}
\def\independenT#1#2{\mathrel{\rlap{$#1#2$}\mkern2mu{#1#2}}}

\newcommand{\mme}[0]{\mathbb{E}}
\newcommand{\mmp}[0]{\mathbb{P}}
\newcommand{\mmr}[0]{\mathbb{R}}
\newcommand{\mmn}[0]{\mathbb{N}}

\newcommand{\bone}[0]{\mathbbm{1}}

\DeclarePairedDelimiterX{\inp}[2]{\langle}{\rangle}{#1, #2}

\usepackage{multirow,float}

\usepackage{relsize}

\usepackage{accents}

\newcommand{\footremember}[2]{%
    \footnote{#2}
    \newcounter{#1}
    \setcounter{#1}{\value{footnote}}%
}
\newcommand{\footrecall}[1]{%
    \footnotemark[\value{#1}]%
}

\allowdisplaybreaks

\begin{document}

\def\spacingset#1{\renewcommand{\baselinestretch}%
{#1}\small\normalsize} \spacingset{1}


\if1\anon
{
  \title{\bf Correcting the Coverage Bias of Quantile Regression}
  \author{Isaac Gibbs\footremember{equalAuth}{The first two authors contributed equally to this work.}\footnote{Department of Statistics, University of California, Berkeley. Corresponding author: \href{mailto:igibbs@berkeley.edu}{igibbs@berkeley.edu}.} \and John J. Cherian\footrecall{equalAuth}
\footnote{Department of Statistics, Stanford University.} \and Emmanuel J. Cand\`{e}s\footnote{Departments of Mathematics and Statistics, Stanford University.}}
  \maketitle
} \fi


\if0\anon
{
  \bigskip
  \bigskip
  \bigskip
  \begin{center}
    {\LARGE\bf Correcting the Coverage Bias of Quantile Regression}
\end{center}
  \medskip
} \fi

\bigskip 
\begin{abstract}
We develop a collection of methods for adjusting the predictions of quantile regression to ensure coverage. Our methods are model agnostic and can be used to correct for high-dimensional overfitting bias with only minimal assumptions. Theoretical results show that the estimates we develop are consistent and facilitate accurate calibration in the proportional asymptotic regime where the ratio of the dimension of the data and the sample size converges to a constant. This is further confirmed by experiments on both simulated and real data. A key component of our work is a new connection between the leave-one-out coverage and the fitted values of variables appearing in a dual formulation of the quantile regression problem. This facilitates the use of cross-validation in a variety of settings at significantly reduced computational costs. 
\end{abstract}

\noindent%


\section{Introduction}

Quantile regression is a popular tool for bounding the tail of a target outcome. This method has a long history dating back to the foundational work of \citet{Koenker1978} and has found widespread applications across a variety of areas \citep{Koenker2001, Koenker2017}. Classical results demonstrate that as the sample size increases quantile regression estimates are consistent, normally distributed around their population analogs \citep{Koenker1978, Angrist2006}, and, perhaps most critically,  achieve their target coverage level \citep{Jung2023, Duchi2025}.

\begin{figure}[ht]
    \centering
    \includegraphics[width=\textwidth]{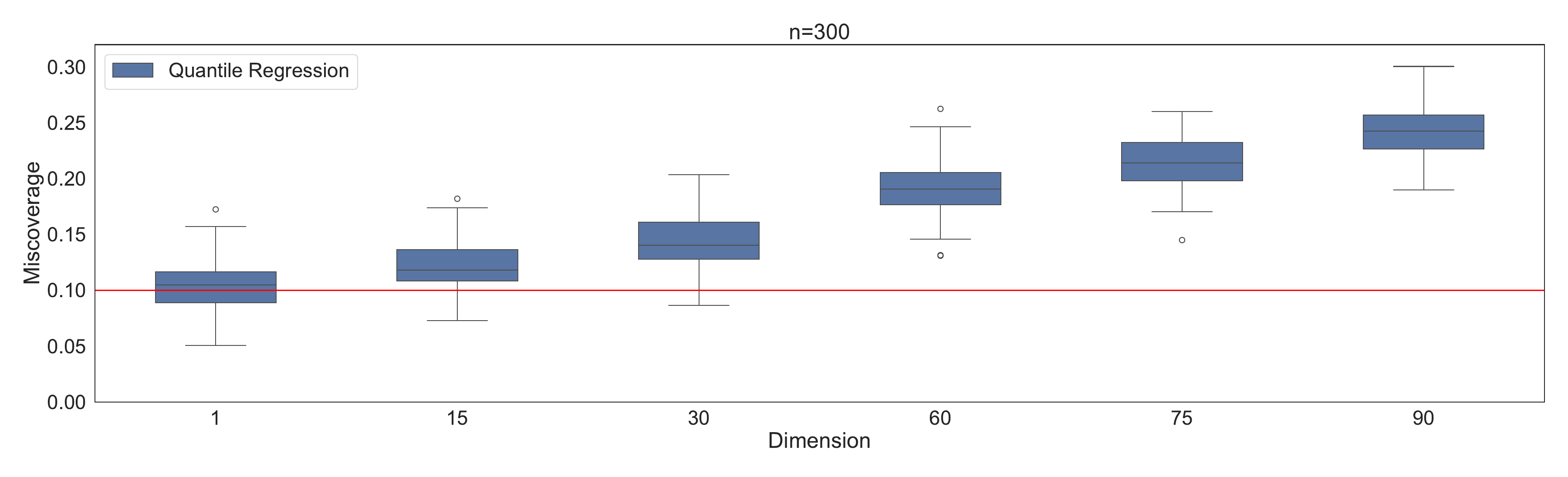}
    \caption{Miscoverage of (unregularized) quantile regression fit with model $Y_i \sim \beta_0 + X_i^\top \beta$ on i.i.d. data $\{(X_i,Y_i)\}_{i=1}^n$ sampled as $Y_i = X_i^\top\tilde{\beta} + \epsilon_i$ for $X_i \sim \mathcal{N}(0,I_d)$, $\epsilon \sim \mathcal{N}(0,1)$, and $\epsilon_i \independent X_i$. Boxplots in the figure show the empirical distribution of the training-conditional coverage, $\mmp(Y_{n+1} \leq \hat{\beta}_0 + X_{n+1}^\top \hat{\beta} \mid \{(X_i,Y_i)\}_{i=1}^n)$ where $(\hat{\beta}_0, \hat{\beta})$ denote the estimated coefficients at quantile level $\tau = 0.9$ and $(X_{n+1},Y_{n+1})$ is an independent sample from the same model. The results come from 100 trials where in each trial the coverage is evaluated over a test set of size 2000 and the population coefficients are sampled as $\tilde{\beta} \sim \mathcal{N}(0,I_d/d)$. The red line shows the target miscoverage level of $1-\tau = 0.1$.}
    \label{fig:qr_cov}
\end{figure}

Although the classical theory can be accurate for large sample sizes, it is often insufficient to characterize the realities of finite datasets. Figure \ref{fig:qr_cov} shows the realized miscoverage of quantile estimates fit at target level $\tau = 0.9$ in a well specified linear model $Y_i = X_i^\top \tilde{\beta} + \epsilon_i$ with $\epsilon_i \independent X_i$ and $X_i \in \mmr^d$.\footnote{Code for reproducing this figure along with all other experiments in this article can be found at \url{https://github.com/isgibbs/adjusted_qr}.} In agreement with the classical theory, we see that when $X_i$ has very low dimension (e.g., $d=1$) quantile regression reliably obtains the target miscoverage rate of $1-\tau = 0.1$. However, the scope of this result is limited, and the coverage shows visible bias in what might be typically considered to be small or moderate dimensions (e.g., $d \in \{15,30\}$ compared to a sample size of $n=300$). Perhaps unsurprisingly, this issue only worsens as the dimension increases and quantile regression exhibits over two times the target error rate when $d=90$.

Formal characterization of the coverage bias of quantile regression was first given in \citet{Bai2021}. They eschew classical theory and instead work under a proportional asymptotic framework in which the ratio of the dimension of the data and the sample size converges to a constant. Under a stylized linear model, they show that in this regime the coverage of quantile regression converges to a value different from the target level and provide an exact formula for quantifying this bias. Interestingly, while both under- and overcoverage are possible, they demonstrate that in most settings quantile regression will tend to undercover.\footnote{As a matter of terminology, if $\hat{q}_{\tau}$ is an estimate of the $\tau \in [1/2,1]$ quantile of $Y$ we say that $\hat{q}_{\tau}$ undercovers if $\mmp(Y \leq \hat{q}_{\tau}) < \tau$ and overcovers if $\mmp(Y \leq \hat{q}_{\tau}) > \tau$. For $\tau < 1/2$ this terminology is reversed and we say that $\hat{q}_{\tau}$ undercovers if $\mmp(Y \leq \hat{q}_{\tau}) > \tau$ and overcovers otherwise. This is motivated by the fact that for $\tau > 1/2$ (resp. $\tau < 1/2$) the $\tau$-quantile is designed to be a high probability upper (resp. lower) bound on $Y$. We use the terms undercoverage and overcoverage to reflect these goals.} This is consistent with the results in Figure \ref{fig:qr_cov} as well as additional empirical evaluations that we will present in Section \ref{sec:empirics}. 

Two proposals have been made in the literature for correcting quantile regression's bias. Under the same linear model assumptions, \citet{Bai2021} derive a simple method for adjusting the nominal level to account for overfitting. While quite effective, this procedure is limited in scope to small aspect ratios and a restrictive model for the data. A more generic procedure that does not require any such modeling assumptions was given in \citet{GCC2025}. They employ a technique known as full conformal inference, which augments the regression fit with a guess of the unseen test point. This mimics the effect of overfitting the training data on the test point, thereby eliminating the resulting bias. In general, this approach has two main drawbacks. First, it requires randomization in order to obtain the desired coverage level. As we will show in Section \ref{sec:dual_thresh}, this randomization can be significant and may cause the quantile estimate to vary substantially. Second, considerable additional computation is required for every test point. This contrasts sharply with standard quantile regression, which once fitted can issue new predictions at the cost of computing just a single inner product. Depending on the application, additional test-time computational complexity may not be permissible. 

In this article, we develop three alternative procedures for adjusting the quantile regression fit. All of these methods are deterministic, and two of them require per test point computation that is identical to standard quantile regression. Briefly, our methods can be summarized as 1) a level-adjustment procedure that tunes the nominal level of the quantile regression loss, 2) an additive adjustment that adds a constant bias to the quantile estimates, and 3) a deterministic analog of the procedure proposed in \citet{GCC2025}. To tune the parameters of these first two methods, we will utilize leave-one-out cross-validation. A central contribution of our work is a new connection between the leave-one-out coverage and a set of dual variables to the quantile regression. This will enable us to compute the entire set of leave-one-out coverage values in time identical to that of running a single regression fit and facilitate hyperparameter tuning at significantly reduced computational costs.

The remainder of this article is structured as follows. In Section \ref{sec:methods} we introduce our main methods. Section \ref{sec:loo_cv} then gives the formal connection between the quantile dual and leave-one-out coverage. Theoretical results showing the consistency of our proposals in the proportional asymptotic regime are presented in Section \ref{sec:theory}, while Sections \ref{sec:simulated_example} and \ref{sec:empirics} give empirical results demonstrating the accuracy of our methods in finite samples. Overall, our results show that all of our proposed methods are robust and provide reliable coverage irrespective of the dimension of the data.

The theoretical results in this paper contribute to a growing literature on characterizing and correcting for overfitting bias in high dimensions (e.g., \cite{EK2013a, Donoho2013, Zhang2013, Javanmard2014, Geer2014, Thram2018, Hastie2022}). Of particular relevance to our work are the Gaussian comparison inequalities of \citet{Gordon1985, Gordon1988} and their development for high dimensional M-estimation problems in \citet{Thram2018}. These tools will allow us to characterize the asymptotic behaviour of the quantile regression dual variables and, through their connection to leave-one-out coverage, to prove the consistency of our cross-validation estimates. There is a large body of literature investigating the performance of cross-validation in high-dimensional parameter tuning (e.g., \cite{Steinberger2016, Kamiar2020, Bayle2020, Austern2020, Xu2021, Patil2021, Patil2022, Steinberger2023, Zou2025}). On a technical level, these articles often require smoothness and/or strong convexity assumptions on the loss in order to derive exact formulas for the leave-one-out coefficients. In contrast, we will be interested in the behaviour of the leave-one-out coverage of quantile regression, which is a discontinuous objective taken over parameter estimates coming from a non-differentiable loss. Here, our connection to the dual program will be critical in allowing us to avoid technical problems present in prior work and facilitating the application of tools that are typically unavailable in studies of cross-validation.

\textbf{Notation:} In the remainder of this article we let $\{(X_i,Y_i)\}_{i=1}^{n+1} \in \mmr^d \times \mmr$ denote a set of covariate-response pairs, where the first $n$ points denote the training set and the last entry is the test point for which $Y_{n+1}$ is unobserved. Given a target level $\tau \in (0,1)$, we will be interested in quantile regression estimates of the form
\[
(\hat{\beta}_0, \hat{\beta}) = \underset{(\beta_0, \beta) \in \mmr^{d+1}}{\text{argmin}} \sum_{i=1}^n \ell_{\tau}(Y_i - \beta_0 - X_i^\top\beta) + \mathcal{R}(\beta),
\]
where $\ell_{\tau}(r) = \tau r - \min\{r,0\}$ is the usual pinball loss and $\mathcal{R} : \mmr^{d} \to \mmr$ is an optional regularization function. For $d$ fixed and $n$ tending to infinity, the quantile regression estimates satisfy the target coverage guarantee $\mmp(Y_{n+1} \leq \hat{\beta}_0 + X_{n+1}^\top \hat{\beta}) \to \tau$. Our goal in this article is to adjust the regression procedure to recover this guarantee even in cases where $d/n \to \gamma \in (0,\infty)$ converges to a constant.

\section{Methods}\label{sec:methods}

We will now introduce our three methods for debiasing quantile regression. As shown theoretically in Section \ref{sec:theory} and empirically in Sections \ref{sec:simulated_example} and \ref{sec:empirics}, all of these procedures provide (asymptotically) exact coverage. Notably, this does not mean that their performance is identical. In Section \ref{sec:empirics} we compare the three approaches across a number of additional metrics (e.g., prediction set length, conditional coverage properties) and observe considerable variability. After reading the introduction to each method below, readers who are primarily interested in practical recommendations may choose to skip ahead to these results.

\subsection{Level adjustment}

The first method we will consider is to modify the nominal level used in the quantile regression loss. In particular, let 
\[
(\hat{\beta}_0(\tau^{\text{adj.}}),\hat{\beta}(\tau^{\text{adj.}})) = \underset{(\beta_0,\beta) \in \mmr^{d+1}}{\text{argmin}} \sum_{i=1}^n \ell_{\tau^{\text{adj.}}}(Y_i - \beta_0 -  X_i^\top \beta),
\]
denote the quantile estimates fit at adjusted level $\tau^{\text{adj.}}$. Let  $(\hat{\beta}^{(-i)}_0(\tau^{\text{adj.}}), \hat{\beta}^{(-i)} (\tau^{\text{adj.}}))$ denote the corresponding leave-one-out coefficients obtained when the $i_{\text{th}}$ sample is excluded from the fit. Then, we define
\begin{equation}\label{eq:adjusted_level}
\hat{\tau}^{\text{adj.}} = \underset{\tau^{\text{adj.}} \in [0,1]}{\text{argmin}}\left| \frac{1}{n} \sum_{i=1}^n \bone\left\{Y_i \leq  \hat{\beta}_0^{(-i)}(\tau^{\text{adj.}}) + X_i^\top \hat{\beta}^{(-i)}(\tau^{\text{adj.}}) \right\} - \tau\right| ,
\end{equation}
as the level that obtains the smallest leave-one-out coverage gap. This gives us the adjusted quantile estimate,
\[
\hat{q}_{\text{level-adj.}}(X_{n+1}) = \hat{\beta}_0(\hat{\tau}^{\text{adj.}}) + X_{n+1}^\top \hat{\beta}(\hat{\tau}^{\text{adj.}}).
\]
As an aside, we remark that in practice the leave-one-out coverage is typically a non-decreasing function of $\tau^{\text{adj.}}$. Using this observation, in the experiments that follow we will compute (\ref{eq:adjusted_level}) using binary search. 

A method for adjusting the quantile regression level has also been previously proposed by \citet{Bai2021}. They showed that when the aspect ratio is small and the data come from a stylized linear model, the value $\hat{\tau}^{\text{adj.}} = (\tau - \frac{1}{2}\frac{d}{n})/(1-\frac{1}{2}\frac{d}{n})$ asymptotically provides the desired coverage. The method above can be seen as a generalization of this procedure that replaces their modeling assumptions with a generic leave-one-out cross-validation based approach.

\begin{figure}[ht]
    \centering
    \includegraphics[width=\textwidth]{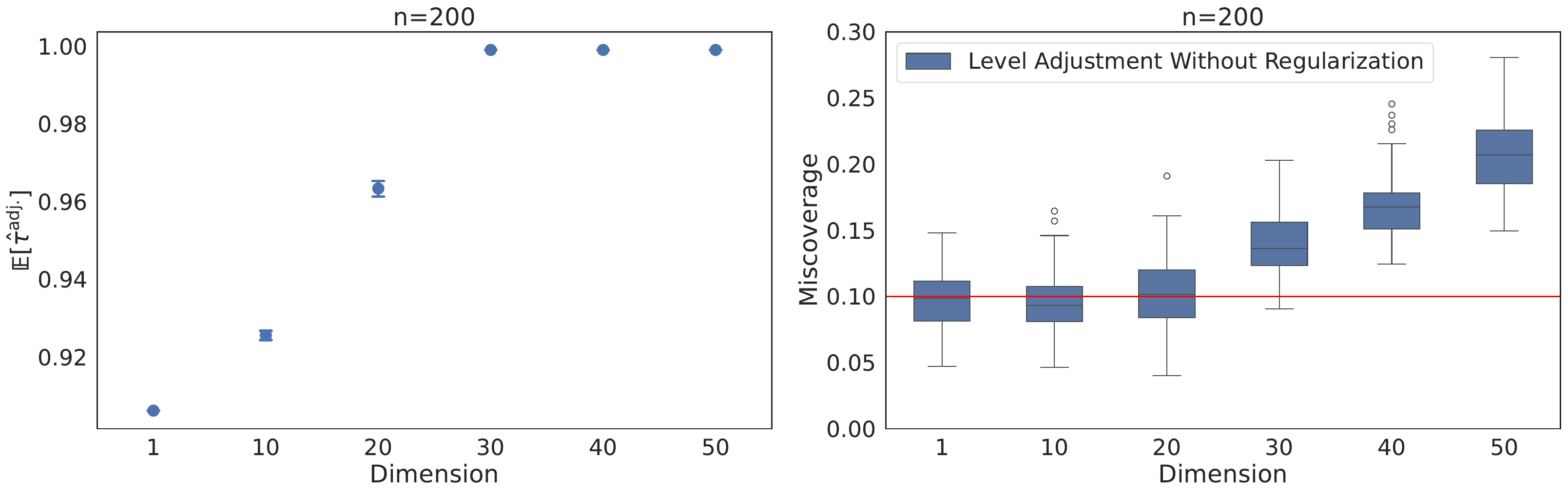}
     \caption{Average value of $\hat{\tau}^{\text{adj.}}$ (left panel) and empirical miscoverage (right panel) of quantile regression fit with an adjusted level as the dimension of the data varies. Data for these experiments are sampled from the Gaussian linear model $Y_i = X_i^\top\tilde{\beta} + \epsilon_i$ with $X_i \sim \mathcal{N}(0,I_d)$, $\epsilon_i \sim \mathcal{N}(0,1)$, and $\epsilon_i \independent X_i$. Dots and error bars in the left  panel show estimated means and 95\% confidence intervals from 100 trials where in each trial the population coefficients are sampled as $\tilde{\beta} \sim \mathcal{N}(0,I_d/d)$. Boxplots in the right panel show the empirical distribution of the training-conditional miscoverage evaluated over the same 100 trials where in each trial the miscoverage is estimated on a test set of size 2000. The red line shows the target miscoverage of $1-\tau = 0.1$. }
    \label{fig:level_tuning}
\end{figure}

Unfortunately, tuning the level alone is not sufficient to regain coverage at higher aspect ratios. The right panel of Figure \ref{fig:level_tuning} shows the realized miscoverage of $\hat{q}_{\text{level-adj.}}(X_{n+1})$ for increasing values of $d/n$ on data generated from the Gaussian linear model. We see that for $d/n \leq 0.1$ leave-one-out cross-validation successfully finds an adjusted level that restores coverage. On the other hand, for larger aspect ratios all values of $\hat{\tau}^{\text{adj.}}$ undercover. As a result, despite selecting the largest possible adjustment of $\hat{\tau}^{\text{adj.}} \approx 1$,\footnote{In this case, we set $\hat{\tau}^{\text{adj.}}$ to be slightly less than $1$ to have a well-defined quantile regression fit.} this method still realizes a significant bias.

To obtain uniform coverage across higher aspect ratios, we will add regularization to the regression. For simplicity, we focus our experiments on ridge regularization, though we anticipate that other choices would also be effective. Proceeding as above, let
\[
(\hat{\beta}_0(\lambda,\tau^{\text{adj.}}), \hat{\beta}(\lambda,\tau^{\text{adj.}})) = \underset{(\beta_0,\beta) \in \mmr^{d+1}}{\text{argmin}} \sum_{i=1}^n \ell_{{\tau}^{\text{adj.}}}(Y_i - \beta_0 - X_i^\top \beta) + {\lambda} \|\beta\|_2^2,
\]
denote the coefficients fit with regularization $\lambda$ and adjusted level $\tau^{\text{adj.}}$, and $\hat{\beta}^{(-i)}_0(\lambda,\tau^{\text{adj.}})$ and $ \hat{\beta}^{(-i)} (\lambda,\tau^{\text{adj.}})$ denote the corresponding coefficients obtained when the $i_{\text{th}}$ sample is omitted. Let 
\[
\text{LOOCov}(\lambda,\tau^{\text{adj.}}) = \frac{1}{n} \sum_{i=1}^n \bone\left\{ Y_i \leq \hat{\beta}^{(-i)}_0(\lambda,\tau^{\text{adj.}}) +  X_i^\top\hat{\beta}^{(-i)} (\lambda,\tau^{\text{adj.}}) \right\},
\]
denote the leave-one-out coverage at parameters $(\lambda,\tau^{\text{adj.}})$. Then, our goal is to find a specific choice $(\hat{\lambda},\hat{\tau}^{\text{adj.}})$ such that $\text{LOOCov}(\hat{\lambda},\hat{\tau}^{\text{adj.}})$ is close to $\tau$. 

In general, there will be more than one setting of $(\lambda,\tau^{\text{adj.}})$ that provides valid coverage. To choose amongst these values, we will use an auxiliary multiaccuracy target. Briefly, we aim to ensure that the miscoverage of the quantile estimate is uncorrelated with the covariates. We defer a detailed discussion of the motivation behind this metric to Section \ref{sec:empirics_metrics} where we discuss other goals for quantile regression beyond marginal coverage. Now, given a discrete grid of candidate values $\Lambda$ for $\lambda$, we select the parameters using the following two-step procedure:
\begin{enumerate}
    \item For $\lambda \in \Lambda$ define
    \[
    \hat{\tau}^{\text{adj.}}(\lambda) = \underset{\tau^{\text{adj.}} \in [0,1]}{\text{argmin}} \left|\text{LOOCov}(\lambda, \tau^{\text{adj.}}) - \tau \right|,
    \]
    as the adjusted level that gives the smallest leave-one-out coverage error.
    \item Let $\Lambda_{\tau} = \{\lambda \in \Lambda : \left|\text{LOOCov}(\lambda, \hat{\tau}^{\text{adj.}}(\lambda)) - \tau \right| \leq 1/n\}$ denote the set of regularization levels that provide a leave-one-out coverage of approximately $\tau$ and
    \begin{equation}\label{eq:lambda_hat_level}
    \hat{\lambda} = \underset{\lambda \in \Lambda_{\tau}}{\text{argmin}}  \max_{j \in \{1,\dots,d\}} \frac{\left|\frac{1}{n} \sum_{i=1}^n X_{i,j} \left(\bone\left\{Y_i \leq \hat{\beta}_0^{(-i)}(\lambda, \hat{\tau}^{\text{adj.}}(\lambda)) + X_i^\top\hat{\beta}^{(-i)}(\lambda, \hat{\tau}^{\text{adj.}}(\lambda)) \right\} - \tau\right) \right|}{ \frac{1}{n} \sum_{i=1}^n |X_{i,j}|},
    \end{equation}
    as the regularization level that minimizes the leave-one-out multiaccuracy error (see Section \ref{sec:empirics_metrics} for a detailed explanation of this error metric). 
\end{enumerate}
As above, in our experiments we implement the first step of this procedure using binary search. This gives us the final adjusted quantile estimate,
\[
\hat{q}_{\text{level-reg.}}(X_{n+1}) = \hat{\beta}_0(\hat{\lambda},\hat{\tau}^{\text{adj.}}(\hat{\lambda})) + X_{n+1}^\top \hat{\beta}(\hat{\lambda},\hat{\tau}^{\text{adj.}}(\hat{\lambda})).
\]

\begin{figure}[ht]
    \centering
    \includegraphics[width=\textwidth]{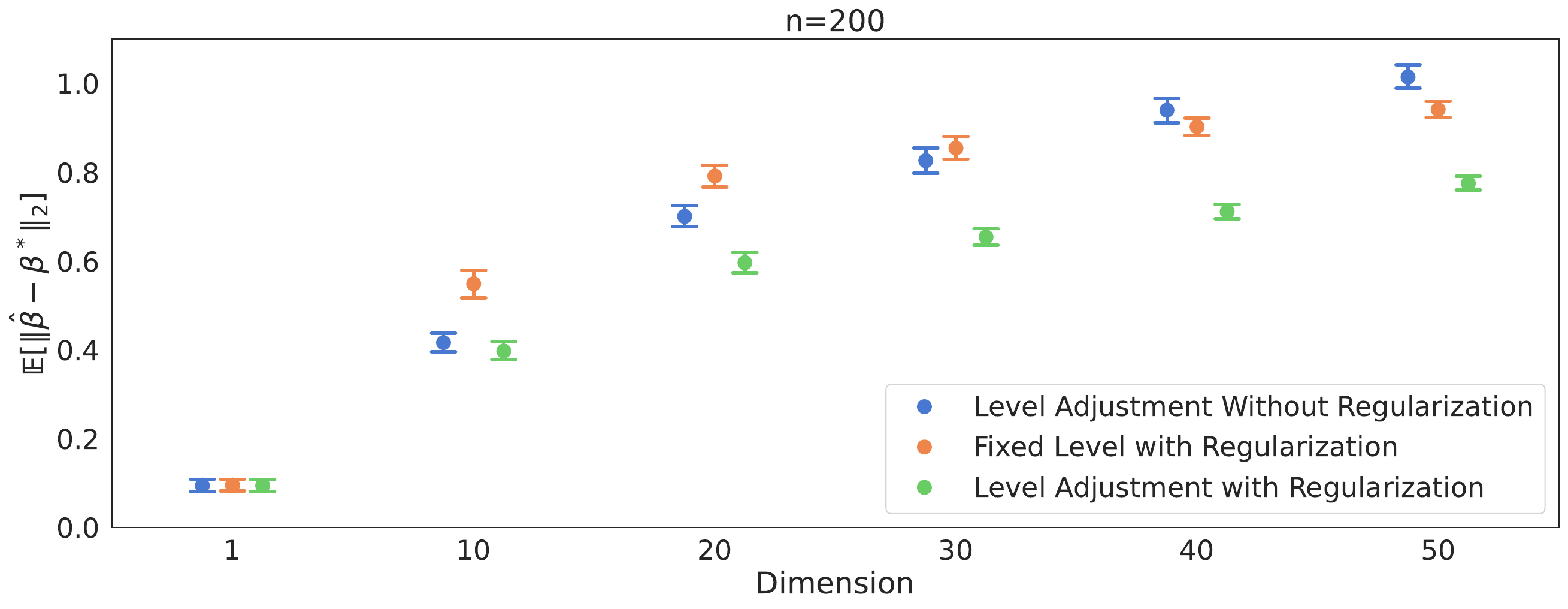}
    \caption{Average coefficient estimation error of quantile regression fit with an adjusted level (blue), adjusted regularization (orange), and a joint level and regularization adjustment (green) as the dimension of the data varies. Data for these experiments are sampled as in Figure \ref{fig:level_tuning} and the target miscoverage is set as $1-\tau = 0.9$. Dots and error bars show estimated means and 95\% confidence intervals from 100 trials. All regularization levels are chosen from the grid $n^{-1}\Lambda = \{0,0.005,0.01,\dots,0.1\}$.}
    \label{fig:level_reg_accuracy}
\end{figure}

Before moving on, it is worthwhile to ask if level-tuning is necessary or if coverage could be more easily obtained by simply holding $\tau^{\text{adj.}} = \tau$ fixed and adjusting the regularization alone. Empirically, we find that while such a strategy is feasible, it typically leads to over regularization. To illustrate this, Figure \ref{fig:level_reg_accuracy} compares the estimation error of $\hat{\beta}(\hat{\tau}^{\text{adj.}})$ and $\hat{\beta}(\hat{\lambda},\hat{\tau}^{\text{adj.}}(\hat{\lambda}))$ against that of $\hat{\beta}(\hat{\lambda}_{\tau},\tau)$ where
\[
\hat{\lambda}_{\tau} = \min \left\{\lambda \in [0,\infty) : \text{LOOCov}(\lambda,\tau) \geq \tau-1/n \right\},
\]
denotes the smallest regularization level that obtains a leave-one-out coverage of at least $\tau-1/n$. Data for this experiment are sampled from a well-specified Gaussian linear model and we target a coverage level of $\tau = 0.9$. We find that joint regularization and level tuning provides the smallest estimation error uniformly across all aspect ratios. As a result, we will prefer this method in the sections that follow and omit further investigation of sole regularization adjustment.

\subsection{Additive adjustment}

The second method we will consider is applying an additive adjustment to the quantile estimate. One way to implement such an adjustment would be to fit the parameters $(\hat{\beta}_0,\hat{\beta})$ using a standard quantile regression and then, at prediction time, output the corrected estimate $c + \hat{\beta}_0 + X_{n+1}^\top \hat{\beta}$ for some constant $c \in \mmr$. This approach has been previously considered by \citet{Romano2019} under the name (split) conformalized quantile regression. They propose to fit the parameter $c$ using a held out subset of the training data that is not used in the quantile regression. In high-dimensional problems where data is scarce, withholding data from the initial regression may lead to a considerable drop in efficiency. In the following section, we will develop a computationally efficient leave-one-out cross-validation procedure that facilitates accurate parameter tuning without data splitting. To leverage that theory here, we now introduce an alternative method for computing an additive adjustment. 

\begin{figure}[ht]
    \centering
    \includegraphics[width=\textwidth]{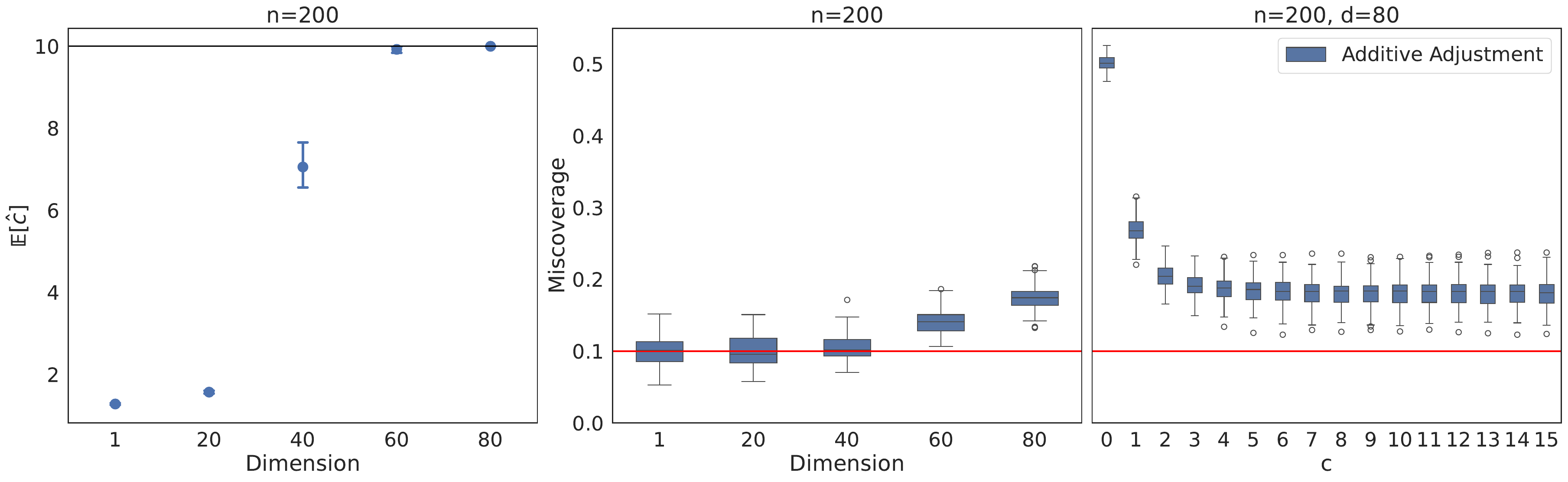}
    \caption{Empirical estimate of the mean selected value of $\hat{c}$ (left panel), realized miscoverage for varying dimension (center panel), and realized miscoverage as $c$ varies (right panel) of the unregularized additive adjustment. Data for this experiment are sampled from the Gaussian linear model $Y_i = X_i^\top\tilde{\beta} + \epsilon_i$ with $X_i \sim \mathcal{N}(0,I_d)$, $\epsilon_i \sim \mathcal{N}(0,1)$, and $\epsilon_i \independent X_i$.  Dots and error bars in the left panel show estimated means and 95\% confidence intervals taken over 100 trials where in each trial the population coefficients are sampled as $\tilde{\beta} \sim \mathcal{N}(0,I_d/d)$. Boxplots in the center and right panel show the empirical distribution of the training-conditional miscoverage evaluated over the same 100 trials where in each trial the miscoverage is estimated on a test set of size 2000. The black line in the left panel shows the maximum allowable value for $\hat{c}$, while red lines in the center and right panel show the target miscoverage of $1-\tau = 0.1$.}
    \label{fig:additive_adjustment}
\end{figure}

For any $c \in \mmr$, let $\hat{\beta}^c$ denote the coefficients fit in the intercept-less quantile regression,
\begin{equation}\label{eq:add_adj_no_reg_params}
\hat{\beta}^c = \underset{\beta \in \mmr^d}{\text{argmin}} \sum_{i=1}^n \ell_{\tau}(Y_i - c - X_i^\top\beta).
\end{equation}
Let $\hat{\beta}^{c,(-i)}$ denote the corresponding coefficients obtained when the $i_{\text{th}}$ sample is excluded from the fit. Similar to the previous section, one reasonable proposal is to select the adjustment
\begin{equation}\label{eq:sole_add_adj}
\hat{c} = \underset{c \in C}{\text{argmin}} \left| \frac{1}{n} \sum_{i=1}^n \bone\left\{ Y_i \leq c + X_i^\top \hat{\beta}^{c,(-i)} \right\}  -\ \tau\right|,
\end{equation}
that provides the smallest leave-one-out coverage gap over some appropriate set of candidate values $C$. This would then give us the adjusted quantile estimate $\hat{q}_{\text{add.-adj.}}(X_{n+1}) = c + X_{n+1}^\top \hat{\beta}^{\hat{c}}$. Unfortunately, as with the level adjustment procedure, we find that at larger aspect ratios this is insufficient to ensure coverage. Figure \ref{fig:additive_adjustment} demonstrates this on simulated data from a Gaussian linear model. For simplicity, in this experiment we restrict the set of candidate values for $c$ to $C = [-10,10]$. Similar to the previous section, we find empirically that the leave-one-out coverage is non-decreasing in $c$ and thus we solve (\ref{eq:sole_add_adj}) using binary search. We find that for $d/n \geq  0.3$ this method almost always selects the maximum value of $\hat{c} = 10$ (left panel) and, despite selecting such a large value, still undercovers (center panel). This issue cannot be alleviated by increasing the cap on $\hat{c}$ as larger values do not change the coverage (right panel). 

To overcome this shortcoming, we will once again add regularization to the regression. Let 
\[
\hat{\beta}^{\lambda,c} = \underset{\beta \in \mmr^d}{\text{argmin}} \sum_{i=1}^n \ell_{\tau}(Y_i - c - X_i^\top \beta) + \lambda \|\beta\|_2^2,
\]
denote the coefficients fit with regularization level $\lambda$ and additive adjustment $c$, and $\hat{\beta}^{\lambda,c,(-i)}$ denote the corresponding coefficients obtained when the $i_{\text{th}}$ sample is excluded from the fit. Let $\text{LOOCov}^{\text{add}}(\lambda,c) = \frac{1}{n} \sum_{i=1}^n \bone\{Y_i \leq c + X_i^\top\hat{\beta}^{\lambda,c,(-i)}\}$ denote the leave-one-out coverage. As above, we search for a pair $(\lambda,c)$ that obtains the desired leave-one-out coverage while minimizing multiaccuracy error. Namely, we fix a grid $\Lambda$ of possible values for $\lambda$ and consider the two-step procedure:
\begin{enumerate}
    \item For $\lambda \in \Lambda$ define
    \[
    \hat{c}(\lambda) = \underset{c \in C}{\text{argmin}} \left|\text{LOOCov}^{\text{add}}(\lambda, c) - \tau \right|,
    \]
    as the additive adjustment that gives the smallest leave-one-out coverage error. 
    \item Let $\Lambda_{\tau} = \{\lambda \in \Lambda : \left|\text{LOOCov}^{\text{add}}(\lambda, \hat{c}(\lambda)) - \tau \right| \leq 1/n\}$ denote the set of regularization levels that provide a leave-one-out coverage of approximately $\tau$ and
    \begin{equation}\label{eq:lambda_hat_add}
    \hat{\lambda} = \underset{\lambda \in \Lambda_{\tau}}{\text{argmin}}  \max_{j \in \{1,\dots,d\}} \frac{\left|\frac{1}{n} \sum_{i=1}^n X_{i,j} \left(\bone\left\{Y_i \leq \hat{c}(\lambda) + X_i^\top\hat{\beta}^{\lambda, \hat{c}(\lambda),(-i)}\right\} - \tau\right) \right|}{ \frac{1}{n} \sum_{i=1}^n |X_{i,j}|},
    \end{equation}
    as the regularization level that minimizes the leave-one-out multiaccuracy error. 
\end{enumerate}
As above, in our experiments step one of this procedure is computed using binary search. This gives us the final quantile adjustment,
\[
\hat{q}_{\text{add.-reg.}}(X_{n+1})  = \hat{c}(\hat{\lambda}
) + X_{n+1}^\top\hat{\beta}^{\hat{c}(\hat{\lambda}),\hat{\lambda}}.
\]

\subsection{Fixed dual thresholding}\label{sec:dual_thresh}

The final method we will consider is a derandomized variant of the full conformal quantile regression procedure proposed in \citet{GCC2025}. Unlike the previous two methods which used leave-one-out estimates to adjust the quantile fit, \citet{GCC2025} instead propose to augment the regression with an imputed guess for the test point. Concretely, they consider unpenalized regressions of the form
\begin{equation}\label{eq:adusted_qr}
(\hat{\beta}^{\text{adj.},y}_0,\hat{\beta}^{\text{adj.},y}) = \underset{(\beta_0,\beta) \in \mmr^{d+1}}{\text{argmin}} \sum_{i=1}^{n} \ell_{\tau}(Y_i - \beta_0 - X_i^\top\beta) + \ell_{\tau}(y - \beta_0 - X_{n+1}^\top\beta),
\end{equation}
and define the adjusted quantile estimate
\[
\hat{q}_{\textup{GCC}}(X_{n+1}) = \sup\{y : y \leq \hat{\beta}^{\text{adj.},y}_0 + X_{n+1}^T\hat{\beta}^{\text{adj.},y}\},
\]
as the maximum value of $y$ that is covered by the regression fit with $y$ in place of $Y_{n+1}$. Under no assumptions on the data beyond that they are i.i.d., this adjustment has the conservative coverage guarantee $\mmp(Y_{n+1} \leq \hat{q}_{\textup{GCC}}(X_{n+1})) \geq \tau$. 

Unfortunately, this guarantee is not typically tight and the authors find that $\hat{q}_{\textup{GCC}}(X_{n+1})$ can exhibit significant overcoverage bias in high dimensions. To further correct this estimate, they additionally introduce a smaller, randomized threshold that is constructed using the quantile regression dual. More formally, let $r_{n+1} = y- \beta_0 - X_{n+1}^\top\beta$ and $r_i = Y_i - \beta_0 - X_i^\top\beta$ for $i \in \{1,\dots,n\}$ denote a set of primal variables that are constrained to be equal to the residuals. Let $\eta \in \mmr^{n+1}$ denote the corresponding dual variables for these constraints. Then, the adjusted quantile regression (\ref{eq:adusted_qr}) can be equivalently written in its primal form as
\begin{align*}\label{}
 (\hat{\beta}_0^{\text{adj.,y}},\hat{\beta}^{\text{adj.,y}} , \hat{r}^{\text{adj.,y}}) = \underset{(\beta_0,\beta) \in \mmr^{d+1}, r \in \mmr^{n+1}}{\text{argmin}}\  & \sum_{i=1}^{n+1} \ell_{\tau}(r_i)\\
 \text{subject to } \hspace{0.5cm} & r_{n+1} = y - \beta_0 - X_{n+1}^\top\beta,\\
 & r_i = Y_i - \beta_0 - X_i^\top\beta,\  \forall i \in \{1,\dots,n\},
\end{align*}
with associated Lagrangian,
\[
L(\beta_0,\beta,r,\eta) = \sum_{i=1}^{n+1} \ell_{\tau}(r_i) + \sum_{i=1}^n \eta_i(Y_i - \beta_0 - X_i^\top \beta - r_i) + \eta_{n+1}(y-\beta_0 - X_{n+1}^\top \beta - r_{n+1}),
\]
and dual program,
\begin{align*}
\hat{\eta}^{\text{adj.},y} = & \underset{\eta \in \mmr^{n+1}}{\text{ argmax}} \sum_{i=1}^n \eta_i Y_i + \eta_{n+1} y\\
& \text{subject to } \sum_{i=1}^{n+1} \eta_i = 0,\  \sum_{i=1}^{n+1} \eta_i X_i = 0,\ -(1-\tau) \preceq \eta \preceq \tau.
\end{align*}
To connect the dual variables to coverage, note that differentiating the Lagrangian with respect to $r_{n+1}$ gives the first-order condition
\[
\hat{\eta}^{\text{adj.},y}_{n+1} \in \begin{cases}
 \{\tau\},\ & y > \hat{\beta}^{\text{adj.},y}_0 + X_{n+1}^\top \hat{\beta}^{\text{adj.},y},\\
 \{-(1-\tau)\},\ & y < \hat{\beta}^{\text{adj.},y}_0 + X_{n+1}^\top \hat{\beta}^{\text{adj.},y},\\
 [-(1-\tau),\tau],\ & y = \hat{\beta}^{\text{adj.},y}_0 + X_{n+1}^\top \hat{\beta}^{\text{adj.},y}.\\
\end{cases}
\]
This connection, along with some additional calculations, motivates the randomized quantile adjustment $\hat{q}_{\textup{GCC, rand.}}(X_{n+1}) = \sup\{y : \hat{\eta}_{n+1}^y \leq U\}$, where $U \sim \text{Unif}(-(1-\tau),\tau)$ is uniformly distributed on the interval $[-(1-\tau),\tau]$. Crucially, this method has the desired exact coverage guarantee, $\mmp(Y_{n+1} \leq \hat{q}_{\textup{GCC, rand.}}(X_{n+1})) = \tau$.

\begin{figure}[ht]
    \centering
    \includegraphics[width=\textwidth]{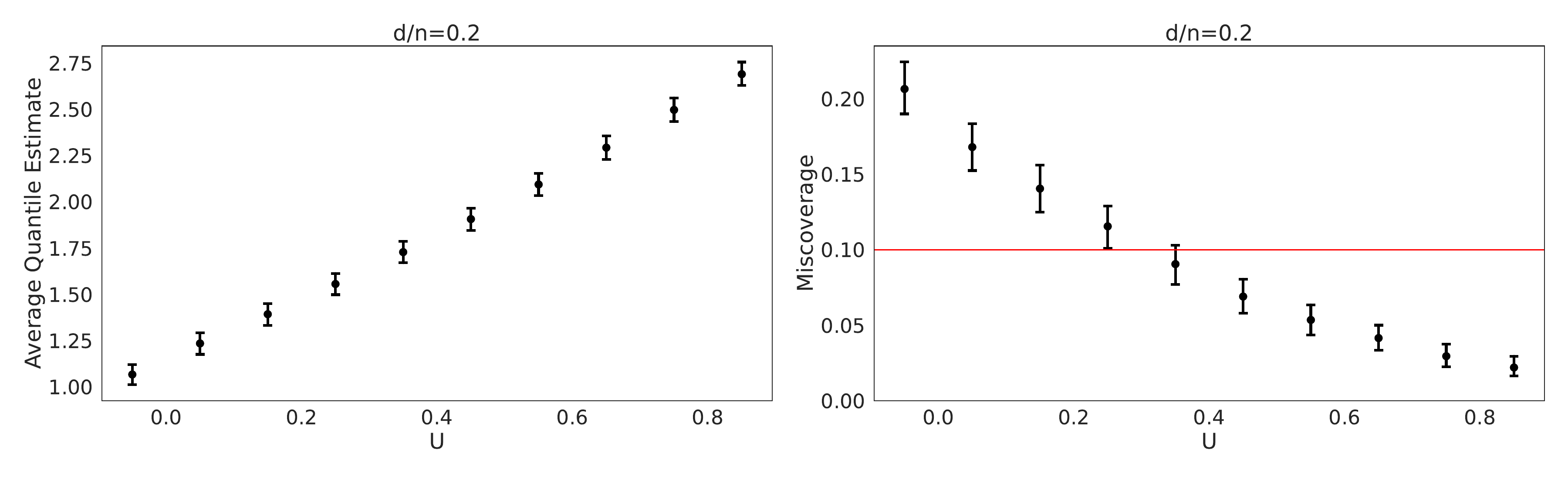}
    \caption{Empirical estimates of the average adjusted quantile (left panel) and miscoverage (right panel) of the randomized method of \citet{GCC2025} conditional on the cutoff $U$. Data for this experiment are sampled from the Gaussian linear model $Y_i = X_i^\top\beta + \epsilon_i$ where $X_i \sim \mathcal{N}(0,I_d)$ and $\epsilon_i \sim \mathcal{N}(0,1)$ with $X_i \independent \epsilon_i$. Dots and error bars show means and 95\% confidence intervals obtained over 2000 samples of the combined training and test dataset $\{(X_i,Y_i)\}_{i=1}^{n+1}$ where in each sample the population coefficients are generated as $\tilde{\beta} \sim \mathcal{N}(0,I_d/d)$. Throughout, we set $d=40$ and $n=200$. The red line in the right panel indicates the target miscoverage level of $1-\tau = 0.1$.} 
    \label{fig:gcc_randomization_dependence}
\end{figure}

As discussed in the introduction, this method has two shortcomings. The first is that to compute the cutoff we need to evaluate the solution path of $\hat{\eta}_{n+1}^y$ as $y$ varies. Although \citet{GCC2025} give some strategies for accomplishing this in an efficient manner, their methods still typically require additional computational time of at least $\Omega(d^3)$\footnote{This comes from the cost of inverting a $d \times d$ matrix, which we shorthand as requiring $\Omega(d^3)$ time, although some algorithms with faster scaling are known.} per test point. Adapting their methods to penalized regressions is more challenging and requires even higher computational complexity. This contrasts sharply with both standard quantile regression and the level and additive adjustment methods proposed in the previous sections which can issue predictions quickly at the low cost of computing a single inner product of the form $X_{n+1}^\top \hat{\beta}$.  The second major shortcoming of $\hat{q}_{\textup{GCC, rand.}}(X_{n+1})$ is that its value depends heavily on the randomized choice of $U$. Figure \ref{fig:gcc_randomization_dependence} displays estimates of the average conditional cutoff, $\mme[\hat{q}_{\textup{GCC, rand.}}(X_{n+1}) \mid U]$ and miscoverage, $\mmp(Y_{n+1} > \hat{q}_{\textup{GCC, rand.}}(X_{n+1}) \mid U)$ as $U$ varies on data sampled from a Gaussian linear model with $d/n = 0.2$. We see that the average cutoff can change by a factor of almost $2.5$ and the miscoverage can vary by over $0.5-2$ times the target level depending on the sampled value of $U$. As an aside, we note that the exact magnitude of these values depends directly on the aspect ratio. In the classical case where $d/n \to 0$ the randomization disappears and the method (asymptotically) produces a fixed cutoff, while larger aspect ratios produce greater variability.

The methods proposed in the previous sections correct both of the above shortcomings. As an alternative approach to contrast with these procedures, we now also propose an unrandomized dual thresholding method that adjusts $\hat{q}_{\text{GCC, rand.}}(X_{n+1})$ by replacing the random cutoff $U$ with a fixed threshold. This method has the advantage of being deterministic, but still suffers similar computational complexity to $\hat{q}_{\text{GCC, rand.}}(X_{n+1})$.

To define our adjustment, let 
\[
\hat{q}_{\text{dual thresh.}}(X_{n+1}; t) = \sup\left\{y : \hat{\eta}_{n+1}^{\text{adj.}, y} \leq t\right\},
\]
denote the quantile estimate obtained when we threshold the dual variable at level $t$. The coverage of this estimate is given by 
\[
\mmp(Y_{n+1} \leq \hat{q}_{\text{dual thresh.}}(X_{n+1}; t)) = \mmp(\hat{\eta}_{n+1}^{\text{adj.}, Y_{n+1}} \leq t).
\]
In particular, to obtain the target coverage level of $\tau$ we see that we should set $t$ as the $\tau$ quantile of $\hat{\eta}^{\text{adj.}, Y_{n+1}}$. Since this quantity is unknown, we replace it with the empirical estimate
\begin{equation}\label{eq:quantile_plug_in}
\hat{t} = \text{Quantile}\left( \tau, \frac{1}{n} \sum_{i=1}^n \delta_{\hat{\eta}_i} \right),
\end{equation}
where $\hat{\eta}$ denotes the dual variables fit using just the training data $\{(X_i,Y_i)\}_{i=1}^n$ and  $\text{Quantile}(\tau,P)$ denotes the $\tau$ quantile of the distribution $P$. This gives us the final adjusted quantile estimate
\[
\hat{q}_{\text{fixed thresh.}}\left(X_{n+1}\right) = \hat{q}_{\text{dual thresh.}}\left(X_{n+1}; \hat{t}\right).
\]

Our theoretical results stated in Theorem \ref{thm:main_asym_consistency} and Corollary \ref{corr:quantile_consistency} show that the quantile estimate given in (\ref{eq:quantile_plug_in}) is consistent and thus that this method obtains the desired coverage. Notably, the proofs of these results are quite different from the approach taken in \citet{GCC2025} to derive a coverage guarantee for $\hat{q}_{\text{GCC, rand.}}(X_{n+1})$. While the arguments given there center on the exchangeability of the fitted dual variables, here we will derive a set of asymptotic consistency results in the proportional asymptotic regime that more directly elucidate the behaviour of our estimates in high dimensions.

\subsection{Simulated example}\label{sec:simulated_example}

We conclude this section with a brief simulated example demonstrating that all of the methods proposed above give accurate coverage in high dimensions. More extensive comparisons that evaluate these procedures across a number of additional metrics are given in Section \ref{sec:empirics}. Similar to Figure \ref{fig:qr_cov} from the introduction, we generate data from the Gaussian linear model $Y_i = X_i^\top \tilde{\beta} + \epsilon_i$ with $X_i \sim \mathcal{N}(0,I_d)$, $\epsilon_i \sim \mathcal{N}(0,1)$, and $\epsilon_i \independent X_i$. Figure \ref{fig:simulated_example} shows the resulting coverage of both our methods and that of standard quantile regression. We see that all three of our methods offer robust coverage irrespective of the dimension (left panel). As anticipated by the theory presented below in Section \ref{sec:theory}, this coverage becomes more tightly concentrated on the target level as $n$ and $d$ increase (right panel).

\begin{figure}[ht]
    \centering
    \includegraphics[width=\textwidth]{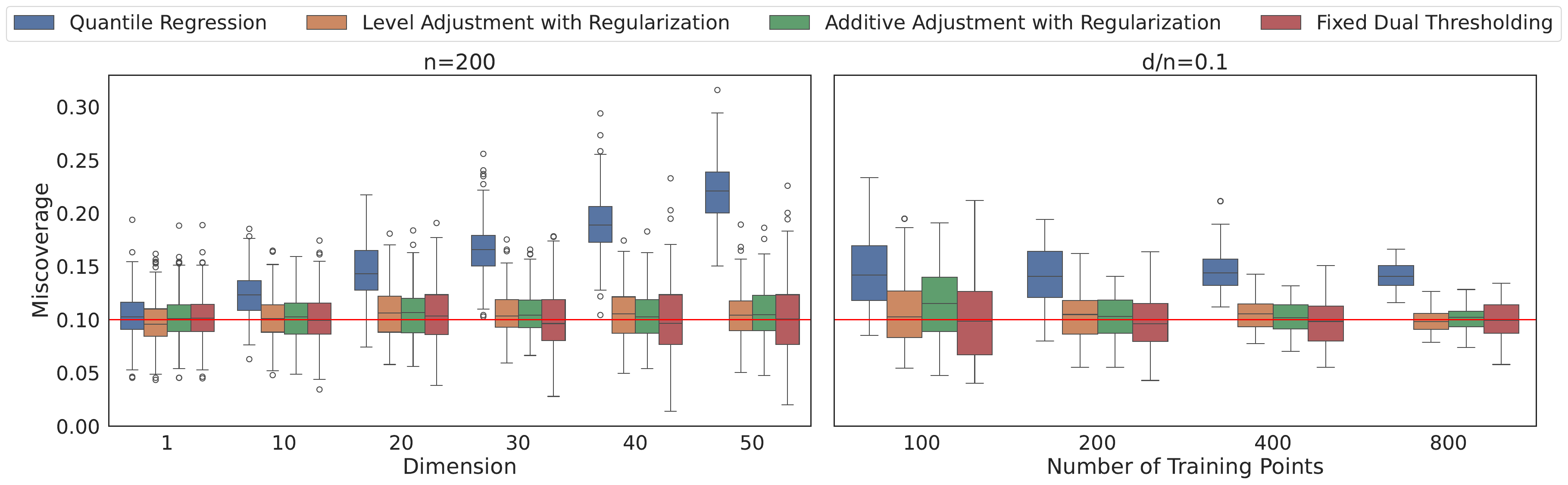}
    \caption{Empirical distribution of the training-conditional miscoverage of quantile regression (blue) and our level adjustment (orange), additive adjustment (green), and fixed dual thresholding (red) methods. The left panel shows results obtained with varying dimension and a fixed sample size of $n=200$, while the right panel varies $n$ and $d$ together at a fixed aspect ratio of $d/n=0.1$. Boxplots show results from 200 trials where in each trial the miscoverage is evaluated empirically over a test set of size 2000 and the population coefficients are sampled as $\tilde{\beta} \sim \mathcal{N}(0,I_d/d)$. The additive adjustment procedure is implemented with range $C = [-10,10]$ for $c$ and all regularization levels are chosen from the grid $n^{-1}\Lambda = \{0,0.005,0.01,\dots,0.1\}$.}
    \label{fig:simulated_example}
\end{figure}

\section{Efficient leave-one-out cross-validation}\label{sec:loo_cv}

Two of the methods developed in the previous section use leave-one-out cross-validation to select their hyperparameters. The typical implementation of these procedures requires fitting $n$ quantile regressions across a range of hyperparameter settings. In this section, we derive a connection between the leave-one-out coverage and the quantile regression dual variables that allows us to obtain all $n$ leave-one-out coverage indicators with just a single fit. Hyperparameter tuning can then be performed at the cost of just a few regression fits across different parameter values.

To introduce this method, we define a few pieces of additional notation. Throughout this section, we consider quantile regressions of the form
\begin{equation}\label{eq:loo_main_qr}
\hat{w} = \underset{w \in \mmr^p}{\text{argmin}} \sum_{i=1}^n \ell_{\tau}(\tilde{Y}_i - \tilde{X}_i^\top \hat{w}) + \mathcal{R}(w).
\end{equation}
Note that unlike the previous sections, here we have chosen to omit an explicit intercept parameter. This allows us to unify the notation to encompass both our  level-based adjustment, in which $\tilde{Y}_i = Y_i$, $\tilde{X}_i = (1,X_i)$, and $p = d+1$, and our additive adjustment, in which $\tilde{Y}_i = Y_i - c$, $\tilde{X}_i = X_i$, and $p = d$. Following the same steps as in Section \ref{sec:dual_thresh}, a useful dual for this regression can be obtained by defining the additional primal variables $r_i = \tilde{Y}_i - \tilde{X}_i^\top w$  for $i \in \{1,\dots,n\}$ and corresponding dual variables $\eta \in \mmr^{n}$ for these constraints.  This gives the dual program
\begin{align*}
 \hat{\eta} =\ & \underset{\eta \in \mmr^n}{\text{argmax}}  \sum_{i=1}^n \eta_i \tilde{Y}_i - \mathcal{R}^*\left( \sum_{i=1}^n \eta_i\tilde{X}_i \right) \numberthis \label{eq:general_dual}\\
& \text{subject to } -(1-\tau) \preceq \eta_i \preceq \tau,
\end{align*}
where $\mathcal{R}^*$ denotes the convex conjugate of $\mathcal{R}$. Finally, we let $\hat{w}^{(-i)}$ denote the corresponding primal solution when the $i_{\text{th}}$ sample is omitted from the fit. In what follows, we will assume without further comment that all of the primal and dual solutions defined above always exist and satisfy strong duality. For common regularization functions (e.g., $\mathcal{R}(w) = \lambda \|w\|_2^2$ or $\mathcal{R}(w) = \lambda\|w\|_1$ for $\lambda \geq 0$), it is straightforward to verify that the primal and dual programs satisfy Slater's condition and thus that this assumption holds (cf.~Section 5.3.2 of \citet{Boyd2004}). 

Our first result derives a general connection between the leave-one-out coverage and the sign of the dual variables. 

\begin{proposition}\label{prop:initial_loo_comparison}
    Assume that $\mathcal{R}$  is convex. Then, all dual solutions $\hat{\eta}$ and leave-one-out primal solutions $\hat{w}^{(-i)}$  satisfy the conditions
    \[
    \tilde{Y}_i < \tilde{X}_i^\top \hat{w}^{(-i)} \implies \hat{\eta}_i \leq 0,
    \]
    and
    \[
    \tilde{Y}_i > \tilde{X}_i^\top \hat{w}^{(-i)} \implies \hat{\eta}_i \geq 0.
    \]
\end{proposition}

Now, recall that our goal is to compute the leave-one-out coverage, $\frac{1}{n} \sum_{i=1}^n \bone\{\tilde{Y}_i \leq \tilde{X}_i^\top \hat{w}^{(-i)}\}$. The above proposition suggests that this quantity should be comparable to $\frac{1}{n} \sum_{i=1}^n \bone\{\hat{\eta}_i \leq 0\}$. Unfortunately however, deriving an exact equivalence between these two quantities is not possible due to the ambiguity around the edge cases $\tilde{Y}_i = \tilde{X}_i^\top \hat{w}^{(-i)}$ and $\hat{\eta}_i = 0$. We are not aware of any simple method for resolving these cases in full generality. One of the key difficulties is that without additional assumptions both the primal and dual solutions may not be unique and at these edge cases the coverage can vary depending on which solution we select. The following example illustrates one such instance where this occurs. 

\begin{example}\label{ex:x_cont_necessary}
Consider fitting an intercept only quantile regression with features $\tilde{X}_i = 1$ and level $\tau = 1/2$ to find the median of the three data points $(\tilde{Y}_1,\tilde{Y}_2,\tilde{Y}_3)$. For simplicity, assume that $\tilde{Y}_1 < \tilde{Y}_2 < \tilde{Y}_3$. The primal solution is $\hat{w} = \tilde{Y}_2$ with corresponding dual solution $\hat{\eta} = (-1/2,0,1/2)$. Critically, we have that $\hat{\eta}_2 = 0$. Now, consider the leave-one-out problem when $(1,\tilde{Y}_2)$ is omitted. Then, the median is any point $\hat{w}^{(-2)} \in [\tilde{Y}_1,\tilde{Y}_3]$ and it is ambiguous whether $\tilde{Y}_2$ is covered. 
\end{example}

We will now introduce two different techniques for modifying the regression to avoid the above ambiguity. For simplicity, we focus on cases where $\mathcal{R}$ is a quadratic regularizer, although we expect similar results to hold for other choices. Our first method is to perturb the covariates by adding a small amount of independent noise to each of their values. The magnitude of this noise is not critical and can be made arbitrarily small such that it has a vanishing impact on the quantile regression objective. Our insight is that even a small amount of noise is sufficient to push the dual solutions away from zero and, correspondingly, to enforce a unique value for the leave-one-out coverage. To illustrate this, the following demonstrates how added noise removes the ambiguity observed in Example \ref{ex:x_cont_necessary}. 

\begin{example}\label{ex:effect of x_noise}
    Consider adding noise to the intercept feature in Example \ref{ex:x_cont_necessary}, i.e., consider fitting a quantile regression at level $\tau = 1/2$ with features $\tilde{X}_i = 1 + \xi_i$, where $\{\xi_i\}_{i=1}^3$ are i.i.d. continuously distributed random variables independent of $\{\tilde{Y}_i\}_{i=1}^3$. As before, assume for simplicity that $\tilde{Y}_1 < \tilde{Y}_2 < \tilde{Y}_3$. For sufficiently small values of $(\xi_1, \xi_2,\xi_3)$, the dual solution is uniquely specified as $\hat{\eta} = (-1/2,\frac{\xi_1 - \xi_3}{2(1+\xi_2)} ,1/2)$ and the leave-one-out primal solution with point $(1+\xi_2,\tilde{Y}_2)$ omitted is (with probability one) unique and given by
    \[
    \hat{w}^{(-2)} = \begin{cases}
         \frac{\tilde{Y}_1}{1+\xi_1}, \text{ if $|1+\xi_1| > |1+\xi_3|$},\\
         \frac{\tilde{Y}_3}{1+\xi_3}, \text{ if $|1+\xi_1| < |1+\xi_3|$}.
    \end{cases}
    \]
    For $(\xi_1,\xi_2,\xi_3)$ sufficiently small, we see that with probability one $\tilde{Y}_2 \neq (1+\xi_2)\hat{w}^{(-2)}$ and thus there is no ambiguity in the coverage of the leave-one-out solution.
\end{example}

The second method we will consider is to add non-zero $L_2$ regularization to all of the primal variables. Similar to the added noise, the magnitude of this regularization is not critical and, in particular, can be taken to be vanishingly small such that it has almost no impact on the regression. The only important consideration is that the regularization makes the fitted leave-one-out solutions unique and thus removes ambiguity in the coverage.

Assumptions \ref{assump:x_pertubation} and \ref{assump:reg_condition} give more formal statements of our two approaches for ensuring leave-one-out uniqueness. We note that both of these assumptions require that the distribution of $\tilde{Y}_i \mid \tilde{X}_i$ is continuous. This can always be guaranteed by adding a small amount of noise to $\tilde{Y}_i$. The main result of this section is stated in Theorem \ref{thm:loo_equivalence}, which shows that these assumptions are sufficient to ensure a one-to-one equivalence between the leave-one-out coverage and the signs of the dual variables.

\begin{assumption}\label{assump:x_pertubation}
    The distribution of $\tilde{Y}_i \mid \tilde{X}_i$ is continuous. Moreover, the regularization can be written as $\mathcal{R}(w) = \sum_{j=1}^p \lambda_j w_j^2$ for some non-negative constants $\lambda_1,\dots,\lambda_p \geq 0$ and the covariates can be written as $\tilde{X}_i = Z_i + \xi_i$ where $\xi_i \independent (Z_i,Y_i)$ has independent, continuously distributed entries. Finally, we have that $p < n$.
\end{assumption}

\begin{assumption}\label{assump:reg_condition}
    The distribution of $\tilde{Y}_i \mid \tilde{X}_i$ is continuous. Moreover, the regularization can be written as $\mathcal{R}(w) = \sum_{j=1}^p \lambda_j w_j^2$ for some positive constants $\lambda_1,\dots,\lambda_p > 0$. 
\end{assumption}

\begin{theorem}\label{thm:loo_equivalence}
Assume that $\mathcal{R}$ is convex, $\{(\tilde{X}_i,\tilde{Y}_i)\}_{i=1}^n$ are i.i.d., and that the conditions of either Assumption \ref{assump:x_pertubation} or \ref{assump:reg_condition} are satisfied. Then, with probability one we have that for all $i \in \{1,\dots,n\}$ either all dual solutions satisfy $\hat{\eta}_i < 0$ or all dual solutions satisfy $\hat{\eta}_i > 0$. Similarly, with probability one either all leave-one-out primal solutions satisfy $\tilde{Y}_i <  \tilde{X}_i^\top \hat{w}^{(-i)}$ or all leave-one-out primal solutions satisfy $\tilde{Y}_i > \tilde{X}^\top_i\hat{w}^{(-i)} $. Finally, letting $\hat{\eta}$ and $\{\hat{w}^{(-i)}\}_{i=1}^n$ denote any such solutions we have that
\[
 \frac{1}{n} \sum_{i=1}^n \bone\{\hat{\eta}_i \leq 0\} \stackrel{a.s.}{=} \frac{1}{n} \sum_{i=1}^n \bone\{\tilde{Y}_i \leq \tilde{X}_i^\top \hat{w}^{(-i)}\}.
\]
\end{theorem}

In general, on real data we find that the conditions outlined in Assumptions \ref{assump:x_pertubation} and \ref{assump:reg_condition} tend to be redundant. In our experiments in Sections \ref{sec:simulated_example} and \ref{sec:empirics} we ignore these assumptions and use the dual variables to estimate the leave-one-out coverage and perform hyperparameter selection across a variety of different datasets and regularization settings that do not satisfy these conditions. In all cases, our results show that the dual estimate is accurate and facilitates the selection of hyperparameter values that yield reliable coverage. As a result, outside of rare edge cases we find that $\frac{1}{n}\sum_{i=1}^n\bone\{\hat{\eta}_i \leq 0\}$ can typically be used to estimate the leave-one-out coverage without the need to modify the data or estimation procedure.

\section{High-dimensional consistency}\label{sec:theory}

We now develop our main theoretical results establishing the high-dimensional consistency of the estimates proposed in the previous sections. Throughout, we will work in a stylized linear model with Gaussian covariates that is commonly used in work in this area (e.g., \cite{Bayati2012, Donoho2013, Thram2015, Thram2018, Dicker2016, Sur2019}). While we will not pursue this in detail, universality results derived for similar problems suggest that one should expect our results to also hold under more relaxed assumptions (e.g.~$X_i$ having i.i.d.~entries) \citep{Han2023}. This is validated by the empirical results presented in the following section which demonstrate the robustness of our methods on real datasets. 

\begin{assumption}\label{assump:data_high_dim}
    The data $\{(X_i,Y_i)\}_{i=1}^{n}$ are i.i.d.~and distributed as $Y_i = X_i^\top \tilde{\beta} + \epsilon_i$ with $X_i \sim \mathcal{N}(0,I_d)$, $\epsilon_i \sim P_{\epsilon}$, and $\epsilon_i \independent X_i$. Moreover, the error distribution $P_{\epsilon}$ is continuous, mean zero, and has at least two bounded moments. Additionally, the density of $P_{\epsilon}$ is bounded, continuous, and positive on $\mmr$. Finally, the population coefficients are themselves random and sampled as $(\sqrt{d}\tilde{\beta}_j)_{j=1}^d \overset{i.i.d.}{\sim} P_{\beta}$ independent of $\{(X_i,\epsilon_i)\}_{i=1}^n$. 
\end{assumption}

We will focus on quantile regressions of the form 
\begin{equation}\label{eq:high_dim_qr}
  \underset{(\beta_0,\beta) \in \mmr^{d+1}}{\min} \sum_{i=1}^n \ell_{\tau}(Y_i - \beta_0 - X_i^\top \beta) + \mathcal{R}_d(\beta).  
\end{equation}
We make two remarks about this set-up. First, for simplicity, we have chosen to focus on regressions containing an intercept. To obtain results for our additive adjustment method we will also need to consider cases where $\beta_0$ is replaced by a fixed constant. This extension is stated at the end of this section in Theorem \ref{thm:add_adj_high_dim}. Second, here we have allowed the regularization function $\mathcal{R}_d$ to depend explicitly on the dimension. This is done to account for the fact that the regularization level may need to be rescaled as $n$ and $d$ increase. Our formal assumptions on the regularizer are stated in Assumption \ref{assump:reg_high_dim_assumptions} in the appendix. At a high-level, we require that $\mathcal{R}_d$ is convex and that the data have enough bounded moments to ensure that various functions of $\mathcal{R}_d$ satisfy the law of large numbers. As an example, Lemma \ref{lem:reg_verification} verifies that our assumptions are met if $P_{\beta}$ has four bounded moments and $\mathcal{R}_d(\beta) = \sqrt{d}\lambda \|\beta\|_1$ or $\mathcal{R}_d(\beta) = d\lambda \|\beta\|^2_2$ is $L_1$ or $L_2$ regularization.

We now state the main result of this section, which establishes that the coordinate-wise empirical distribution of the dual variables converges to an asymptotic limit. Although we only state this result for aspect ratios $d/n \to \gamma \in (0,2/\pi)$, we expect similar
conclusions to hold for $\gamma \geq 2/\pi$ under appropriate assumptions on the regularization.

\begin{theorem}\label{thm:main_asym_consistency}
    Fix any $\tau \in (0,1)$ and suppose that the data and regularizer satisfy the conditions of Assumptions \ref{assump:data_high_dim} and \ref{assump:reg_high_dim_assumptions}. Suppose that $d,n \to \infty$ with $d/n \to \gamma \in (0,2/\pi)$. Then, there exists a limiting distribution $P_{\eta}$ such that for any bounded, Lipschitz function $\psi$ and any $\delta > 0$,
    \[
    \mmp\left(\text{For all dual solutions $\hat{\eta}$ to (\ref{eq:high_dim_qr}), }
    \left| \frac{1}{n} \sum_{i=1}^n \psi(\hat{\eta}_i) - \mme_{Z \sim P_{\eta}}[\psi(Z)] \right| \leq \delta \right) \to 1.
    \]
    Moreover, the distribution $P_{\eta}$ is supported on $[-(1-\tau),\tau]$ with discrete masses at $-(1-\tau)$ and $\tau$ and a continuous distribution on $(-(1-\tau),\tau)$.
\end{theorem}

As an aside, we remark that explicit formulas for the asymptotic distribution $P_{\eta}$ are given in Proposition \ref{prop:M_u_is_0} and equation (\ref{eq:asymptotic_dist}) in the appendix. The exact definition of this distribution is somewhat involved and thus we defer a more precise description of the relevant quantities to Appendix \ref{sec:app_main_asymp}.

Theorem \ref{thm:main_asym_consistency} has two critical corollaries for our debiasing methods. The first establishes the consistency of our leave-one-out coverage estimates. Unlike in Section \ref{sec:loo_cv} where we restricted our attention to $L_2$ regularization, here our extra assumptions on the data allow us to derive a result for much more general regularizers.

\begin{corollary}\label{corr:loo_cov_consistency}
    Let $(X_{n+1}, Y_{n+1})$ denote an independent sample from the same distribution as $\{(X_i,Y_i)\}_{i=1}^n$. Let $(\hat{\beta}_0,\hat{\beta})$ denote any primal solution to (\ref{eq:high_dim_qr}) chosen independently of $(X_{n+1},Y_{n+1})$. Then, under the assumptions of Theorem \ref{thm:main_asym_consistency}, it holds that for any $\delta > 0$,
    \begin{align*}
    & \mmp\left(\text{For all dual solutions $\hat{\eta}$ to (\ref{eq:high_dim_qr}), }  \left| \frac{1}{n} \sum_{i=1}^n \bone\{\hat{\eta}_i \leq 0\} - \mmp\left(Y_{n+1} \leq \hat{\beta}_0 + X_{n+1}^\top\hat{\beta} \right)  \right| \leq \delta \right)\to 1.
    \end{align*}
\end{corollary}

Our second corollary shows that the quantile estimate used by our fixed dual thresholding method is consistent. 

\begin{corollary}\label{corr:quantile_consistency}
    Consider unregularized quantile regression with $\mathcal{R}_d(\beta) = 0$. Under the assumptions of Theorem \ref{thm:main_asym_consistency}, it holds that for any $\delta > 0$, 
    \[
    \mmp\left(\text{For all dual solution $\hat{\eta}$ to (\ref{eq:high_dim_qr}), }
    \left|  \textup{Quantile}\left(\tau, \frac{1}{n} \sum_{i=1}^n \delta_{\hat{\eta}_i}  \right) -  \textup{Quantile}\left(\tau, P_{\eta} \right) \right| \leq \delta \right) \to 1.
    \]
\end{corollary}

Proofs of Theorem \ref{thm:main_asym_consistency} as well as Corollaries \ref{corr:loo_cov_consistency}  and \ref{corr:quantile_consistency} are given in the appendix. Our arguments build heavily on Gordon's comparison inequalities \citep{Gordon1985, Gordon1988} and their application to high-dimensional regression developed in \citet{Thram2018}. At a high level, these results allow us to derive a correspondence between the dual quantile regression program and a simplified auxiliary optimization problem that replaces the covariate matrix with vector-valued random variables. The main technical difficulty is then to characterize the solutions of this auxiliary program. One key difference between our result and that of the original work of \citet{Thram2018} is that we consider the behaviour of the solutions under arbitrary bounded, Lipschitz functions. We are not the first to derive an extension of this type. However, to the best of our knowledge previous extensions typically rely on strong convexity of the auxiliary optimization (see e.g., \citet{Abassi2016, Miolane2021, Celentano2023}). Here, we derive a similar result under weaker conditions. 

It is worthwhile to contrast Theorem \ref{thm:loo_equivalence} with the results of \citet{Bai2021}. In that paper, the authors derive a number of asymptotic consistency results for the primal quantile regression estimates $(\hat{\beta}_0, \hat{\beta})$. Here, we provide a set of complementary asymptotics for the dual. In addition, we also treat a more general setting that removes the restrictions to small aspect ratios and unregularized regressions present in their work. While not our main focus, a corollary of our analysis is that the intercept, $\hat{\beta}_0$ and estimation error, $\|\hat{\beta} - \tilde{\beta}\|_2$ both converge to constants under the assumptions of Theorem \ref{thm:main_asym_consistency}. This is formally stated in Theorem \ref{thm:primal_convergence} in the appendix, which directly generalizes Theorem C.1 of \citet{Bai2021}.

Finally, as a last remark, we note that all of the conclusions stated above also hold for the intercept-less regressions used in the additive adjustment procedure. The proof of this result is nearly identical to that of Theorem \ref{thm:main_asym_consistency} and thus is omitted.

\begin{theorem}\label{thm:add_adj_high_dim}
    Under identical assumptions, the conclusions of Theorem \ref{thm:main_asym_consistency} and Corollaries  \ref{corr:loo_cov_consistency} and \ref{corr:quantile_consistency} remain true when the intercept $\beta_0$ is replaced by a fixed, real-valued constant. 
\end{theorem}

\section{Real data experiments}\label{sec:empirics}

\subsection{Methods and metrics}\label{sec:empirics_metrics}

We now undertake a series of empirical comparisons of our proposed methods. As baselines, we also evaluate the performance of standard quantile regression, the randomized method of \citet{GCC2025}, and the (split) conformalized quantile regression (CQR) method of \citet{Romano2019}. In all experiments, we implement CQR so that 75\% of the data is used to fit the quantile regression and 25\% is used to calibrate its coverage. 

To evaluate these methods, we compare the quality of prediction sets constructed using their estimated quantiles. More precisely, for a given miscoverage level $\alpha \in (0,1/2)$ (taken to be 0.1 in our experiments) we compute the (adjusted) estimates $\hat{q}^{\alpha/2}(X_{n+1})$ and $\hat{q}^{1-\alpha/2}(X_{n+1})$ of the $\alpha/2$ and $1-\alpha/2$ quantiles using each of the methods. We then evaluate the resulting prediction interval $[\hat{q}^{\alpha/2}(X_{n+1}),\hat{q}^{1-\alpha/2}(X_{n+1})]$ in terms of three criteria: 1) marginal coverage, $\mmp(\hat{q}^{\alpha/2}(X_{n+1}) \leq Y_{n+1} \leq \hat{q}^{1-\alpha/2}(X_{n+1}))$, 2) interval length, $\max\{\hat{q}^{1-\alpha/2}(X_{n+1}) - \hat{q}^{\alpha/2}(X_{n+1}),0\} $, and 3) maximum multiaccuracy error.

Multiaccuracy as introduced in \citet{hebert2018multicalibration} and \citet{kim2019multiaccuracy} is a general criteria for measuring the bias of a predictor over reweightings of the covariate space. In our context, we will consider linear reweightings and thus our goal will be to obtain quantile estimates whose miscoverage events are uncorrelated with the features. This is motivated by results from the classical regime in which $d\log(n)/n \to 0$. There, \citet{Duchi2025} showed that (under appropriate tail bounds on the data) the canonical quantile regression estimates $(\hat{q}_{\text{QR}}^{\alpha/2}, \hat{q}_{\text{QR}}^{1-\alpha/2})$ satisfy the multiaccuracy condition,\footnote{See also \citet{Jung2023} for an earlier appearance of a similar result.}
\begin{equation}\label{eq:qr_ma_conv}
\sup_{\|v\|_2 \leq 1}  \mme\left[X_{n+1}^\top v(\bone\{ Y_{n+1}  \in [\hat{q}^{\alpha/2}_{\text{QR}}(X_{n+1}) ,\hat{q}^{1-\alpha/2}_{\text{QR}}(X_{n+1})] \} - (1-\alpha)) \mid \{(X_i,Y_i)\}_{i=1}^n \right]  \stackrel{\mmp}{\to} 0.
\end{equation}
As a concrete example to motivate the utility of this condition, consider fitting quantile regression with a feature $X_{i,j} = \bone\{X_{i} \in G\}$ that indicates whether sample $i$ falls into group $G$. Then, applying (\ref{eq:qr_ma_conv}) with $v = e_i$ gives the conditional coverage statement,
\[ \mmp( Y_{n+1} \in [\hat{q}^{\alpha/2}_{\text{QR}}(X_{n+1}), \hat{q}^{1-\alpha/2}_{\text{QR}}(X_{n+1})] \mid X_{n+1} \in G, \{(X_i,Y_i)\}_{i=1}^n) \stackrel{\mmp}{\to} 1-\alpha.
\]
More generally, by designing the features appropriately multiaccuracy conditions of this form can be used to ensure that the prediction set provides accurate performance across sensitive attributes of the population.  

Motivated by this, \citet{GCC2025} extend (\ref{eq:qr_ma_conv}) to the high-dimensional setting and show that their randomized adjustment satisfies 
\[
\mme\left[X_{n+1}^\top v(\bone\{ Y_{n+1} \in [\hat{q}^{\alpha/2}_{\text{GCC, rand.}}(X_{n+1}),  \hat{q}^{1-\alpha/2}_{\text{GCC, rand.}}(X_{n+1})]\} - (1-\alpha)) \right] = 0, \ \forall v \in \mmr^d.
\]
Notably, this statement is not directly comparable to (\ref{eq:qr_ma_conv}) since here the expectation is taken marginally  over the random draw of the training set. In general, one cannot expect to obtain training-conditional convergence uniformly over $v$ in high dimensions. Nevertheless, as we will see shortly, empirically  $\hat{q}_{\text{GCC, rand.}}(X_{n+1})$ can provide approximate validity when $v$ is restricted to a smaller set (e.g.~to the coordinate axes). 

The methods developed in the previous section are not designed to explicitly guarantee multiaccuracy. Regardless, since they are built on top of quantile regression one may hope that they still approximately satisfy these conditions. To evaluate this, we will examine the coordinatewise multiaccuracy error of each method defined as
\begin{equation}\label{eq:ma_target}
 \max_{j \in \{1,\dots,d\}} \frac{\mme[X_{n+1,j} (\bone\{ Y_{n+1} \in [\hat{q}^{\alpha/2}(X_{n+1}),  \hat{q}^{1-\alpha/2}(X_{n+1})]\} - (1-\alpha)) \mid \{(X_i,Y_i)\}_{i=1}^n ]}{\mme[|X_{n+1,j}|]}.
\end{equation}
We recall that in order to improve the performance on this metric, in Section \ref{sec:methods} we defined the parameters for our regularized level and additive adjustment procedures to minimize a leave-one-out estimate of (\ref{eq:ma_target}) (cf.~equations (\ref{eq:lambda_hat_level}) and (\ref{eq:lambda_hat_add})).

\subsection{Results}

We compare the methods on two datasets in which the goals are to predict the per capita violent crime rate of various communities \citep{Redmond2002} and the number of times a news article was shared online \citep{Fernandes2015}. Both datasets are publicly available from the University of California, Irvine machine learning repository \citep{Dua2019}. After filtering out features with missing values and removing (approximately) linearly dependent columns, the datasets have 99 and 55 covariates, respectively. We normalize both the features and the target to have mean zero and variance one and then compare the methods discussed above in terms of their miscoverage, median length, and multiaccuracy error.

\begin{figure}[ht]
    \centering
    \includegraphics[width=\textwidth]{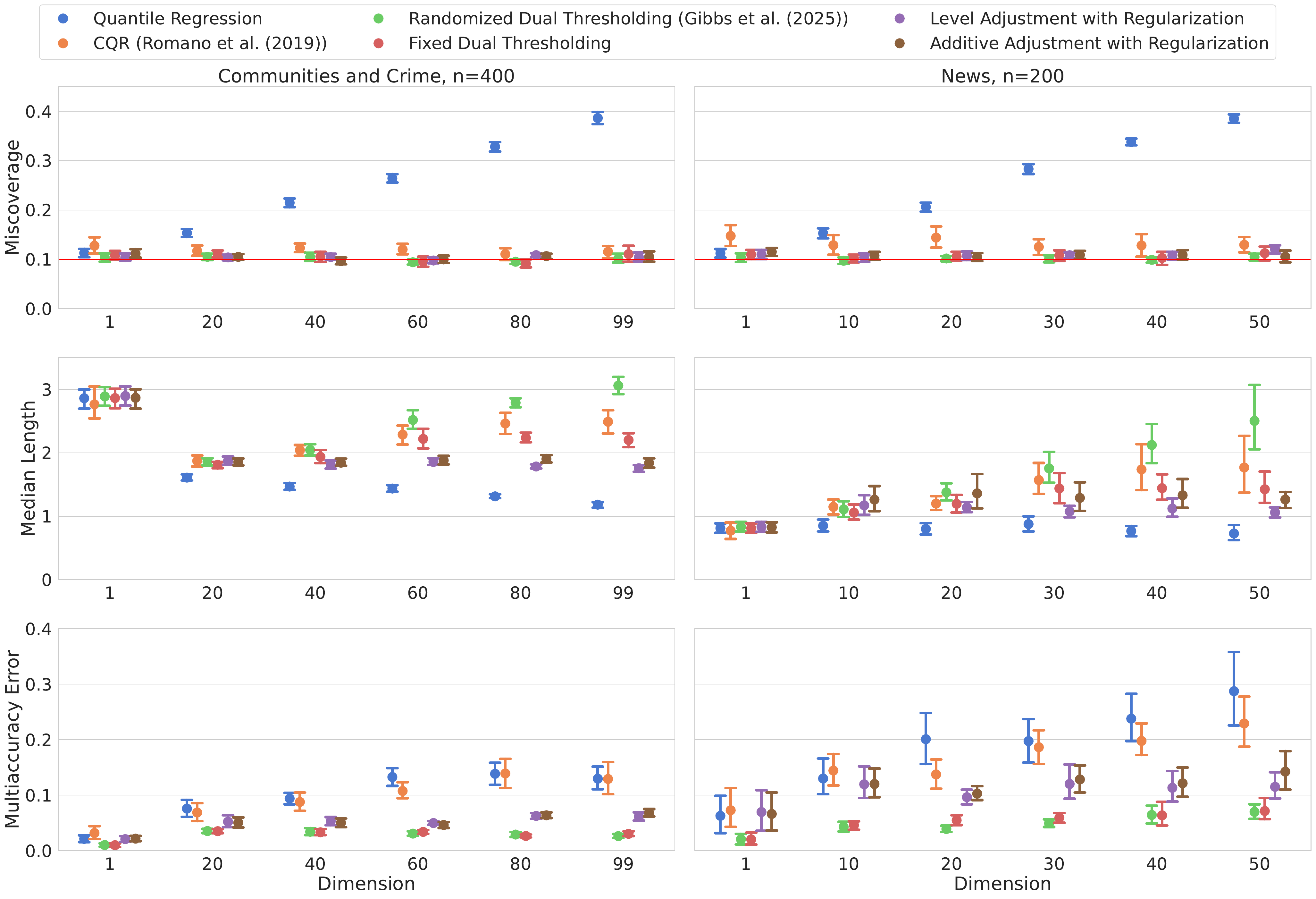}
    \caption{Empirical miscoverage (top row), median length (center row), and multiaccuracy error (bottom row) of quantile regression (blue), the baseline methods of \citet{Romano2019} (orange) and \citet{GCC2025} (green) and our fixed dual thresholding (red), level adjustment (purple), and additive adjustment (brown) methods on the communities and crime (left panels) and news (right panels) datasets as the dimension varies. Dots and error bars show means and 95\% confidence intervals obtained over 20 trials. The red lines in the top panels indicate the target level of $\alpha = 0.1$. In all experiments, the additive adjustment procedure is implemented with range $C = [-10,10]$ for $c$ and regularization levels are chosen from the grid $n^{-1}\Lambda = \{0,0.005,0.01,\dots,0.2\}$.}
    \label{fig:all_comparisons}
\end{figure}

Figure \ref{fig:all_comparisons} shows the outcome of this experiment. For the communities and crimes (resp.~news) dataset results in the figure summarize 20 trials where in each trial the data are randomly split into a training set of size 400 (resp.~200) and a test set of size 1594 (resp.~2000)\footnote{The communities and crimes dataset only has 1994 samples, so we utilize a smaller test set of size 1594 for its experiments.} and a random subset of the features are selected for use. As shown in the top row,  all methods provide the desired coverage except for standard quantile regression which realizes significant bias as the dimension increases. Among the methods with accurate coverage, our level and additive adjustment procedures (purple and brown) yield the smallest intervals (center row). The largest intervals are output by the randomized method of \citet{GCC2025} (green), which obtains a median interval length of up to two times that of the level adjustment procedure in higher dimensions. 

In terms of multiaccuracy, the lowest error is obtained by the dual thresholding methods (bottom row). Interestingly, while randomization is necessary to obtain a theoretical multiaccuracy bias of zero, we find that the fixed thresholding method (red) offers nearly identical performance in practice. On the other hand, our level and additive adjustment procedures realize a higher multiaccuracy error (purple and brown). This is to be expected since by adding regularization to these methods we have introduced bias. To see this, note that letting $(\hat{\beta}_0(\lambda),\hat{\beta}(\lambda))$ denote the fitted coefficients at quantile level $\tau$ with $L_2$ regularization $\lambda$ and $\tilde{\beta}$ denote the population quantile regression coefficients, we have that in the classical regime where $d\log(n)/n \to 0$,
\[
\mme\left[X_{n+1}^\top v(\bone\{Y_{n+1} \leq \hat{\beta}_0(\lambda) + X_{n+1}^\top \hat{\beta}(\lambda)\} - \tau)  \mid \{(X_i,Y_i)\}_{i=1}^n \right] \stackrel{\mmp}{\to} -2\lambda v^\top\tilde{\beta}.
\]
This follows directly from the first-order conditions of quantile regression and the arguments of \citet{Duchi2025}. Notably, while non-negligible, we find that this bias is small relative to the effects of overfitting and our level and additive adjustment procedures still produce much lower multiaccuracy errors than the baseline approaches of quantile regression and conformalized quantile regression. 

\section{Conclusions}

In this paper we developed three methods for correcting the coverage bias of quantile regression. Theoretical and empirical results show that all of these procedures provide robust coverage irrespective of the dimension of the data. In terms of prediction interval length and multiaccuracy error, none of the methods dominate. Across our empirical results we find that the fixed dual thresholding method offers the lowest multiaccuracy error. However, this comes at the cost of wider prediction intervals and greater test-time computational complexity as compared to the level and additive adjustment procedures.

\section{Acknowledgments}

E.C.~was supported by the Office of Naval Research grant N00014-24-1-2305, the National Science Foundation grant DMS-2032014, and the Simons Foundation under award 814641. J.J.C.~was supported by
the John and Fannie Hertz Foundation.

%

\newpage

\bibliography{qr_bib}

\newpage





\appendix

\section{Proofs for Section \ref{sec:loo_cv}}

We will now give formal proofs for the results stated in Section \ref{sec:loo_cv}. Throughout, we use the same notation as was defined in the main text. Namely, we let $\{(\tilde{X}_i,\tilde{Y}_i)\}_{i=1}^n\subseteq \mmr^p \times \mmr$ denote the training data and we consider quantile regressions of the form
\[
\min_{w \in \mmr^p} \sum_{i=1}^n \ell_{\tau}(\tilde{Y}_i - \tilde{X}_i^\top w) + \mathcal{R}(w).
\]
We use $\hat{w}$ and $\hat{\eta}$ to denote primal and dual solutions to this regression and $\hat{w}^{(-i)}$ and $\hat{\eta}^{(-i)}$ to denote corresponding leave-one-out primal and dual solutions when the $i_{\text{th}}$ sample is omitted from the fit. To begin, we prove a useful technical lemma that relates a dual value of zero to interpolation of the leave-one-out prediction. 

\begin{lemma}\label{lem:dual_at_0_to_interpolation}
    Fix any $i \in \{1,\dots,n\}$ and suppose there exists a leave-one-out primal solution with $\tilde{Y}_i = \tilde{X}_i^\top \hat{w}^{(-i)}$. Then, there exists a dual solution to the full program with $\hat{\eta}_i = 0$. Symmetrically, if there exists a dual solution to the full program with $\hat{\eta}_i = 0$, then there exists a leave-one-out primal solution with $\tilde{Y}_i = \tilde{X}_i^\top \hat{w}^{(-i)}$.
\end{lemma}
\begin{proof}
    For notational simplicity, we will focus on the case $i = n$. Suppose there exists a leave-one-out primal solution with $\tilde{Y}_n = \tilde{X}_n^\top \hat{w}^{(-n)}$. For $i \in \{1,\dots,n-1\}$, let $\hat{r}_i^{(-n)} =  \tilde{Y}_i - \tilde{X}_i^\top \hat{w}^{(-n)}$ denote the additional primal variables. Let $\hat{\eta}^{(-n)} \in \mmr^{n-1}$ denote a corresponding dual solution to the leave-one-out program. The Lagrangian for the leave-one-out program is
    \[
    L^{(-n)}\left(w^{(-n)}, r^{(-n)}, \eta^{(-n)}\right) = \sum_{j=1}^{n-1} \ell_{\tau}\left(r_j^{(-n)}\right) + \sum_{j=1}^{n-1} \eta^{(-n)}_j \left(\tilde{Y}_j - \tilde{X}_j^\top w^{(-n)} -r^{(-n)}_j\right) + \mathcal{R}\left(w^{(-n)}\right),
    \]
    and the Lagrangian for the full program is
    \begin{equation}\label{eq:qr_lagrangian}
    L(w, r, \eta) = \sum_{j=1}^{n} \ell_{\tau}(r_j) + \sum_{j=1}^{n} \eta_j\left(\tilde{Y}_j - \tilde{X}_j^\top w -r_j\right)+ \mathcal{R}(w).
    \end{equation}
    By assumption, $(\hat{w}^{(-n)},\hat{r}^{(-n)},\hat{\eta}^{(-n)})$ is a saddle point of $L^{(-n)}$. Using this fact, and taking first-order derivatives, it is straightforward to verify that $(\hat{w}^{(-n)}, (\hat{r}^{(-n)},0), (\hat{\eta}^{(-n)},0))$ is a saddle point of $L$. Thus, $\hat{\eta} = (\hat{\eta}^{(-n)},0)$ is a solution to the full dual program, as desired.

    For the reverse direction, suppose there exists a solution to the full dual program with $\hat{\eta}_n = 0$. Let $(\hat{w}, \hat{r}, \hat{\eta})$ denote the corresponding saddle point of $L$. By differentiating $L$ with respect to $r_n$, we see that we must have $\hat{r}_n = 0$. Moreover, differentiating $L$ with respect to $\hat{\eta}_n$, we then also find that $\tilde{Y}_n - \tilde{X}_n^\top \hat{w} = \hat{r}_n = 0$. So, using the notation $v_{1:(n-1)}$ to denote the first $n-1$ entries of a vector $v \in \mmr^n$ and taking first-order derivatives of $L^{(-n)}$, it is straightforward to verify that $(\hat{w}, \hat{r}_{1:(n-1)}, \hat{\eta}_{1:(n-1)})$ is a saddle point of $L^{(-n)}$. Since $\tilde{Y}_n = \tilde{X}_n^\top \hat{w}$, this proves the desired result. 
\end{proof}

To prove Proposition \ref{prop:initial_loo_comparison}, we will need one additional technical result demonstrating that the $i_{\text{th}}$ coordinate of the dual solution, $\hat{\eta}_i$, behaves monotonically in $\tilde{Y}_i$. This result was originally derived in \citet{GCC2025} where it was used to obtain efficient algorithms for computing $\hat{q}_{\text{GCC, rand.}}(\cdot)$. To state the result formally, let
\begin{equation}\label{eq:app_dual_with_swap}
\hat{\eta}^{\tilde{Y}_i \to y} = \underset{\eta \in [-(1-\tau),\tau]^n}{\text{argmax}} \sum_{j \neq i} \eta_j \tilde{Y}_j + \eta_i y - \mathcal{R}^*\left(\sum_{j=1}^n \eta_j \tilde{X}_j\right),
\end{equation}
denote the dual solution obtained when $\tilde{Y}_i$ is replaced with $y \in \mmr$. We have the following lemma.

\begin{lemma}\label{lem:monotone_dual}[Theorem 4 of \citet{GCC2025}]
Fix any $i \in \{1,\dots,n\}$ and let $\{\hat{\eta}^{\tilde{Y}_i \to y}\}_{y \in \mmr}$ denote any collection of solutions to (\ref{eq:app_dual_with_swap}). Then, $y \mapsto \hat{\eta}^{\tilde{Y}_i \to y}$ is non-decreasing.  
\end{lemma}

We are now ready to prove Proposition \ref{prop:initial_loo_comparison}.

\begin{proof}[Proof of Proposition \ref{prop:initial_loo_comparison}]
    Fix any $i \in \{1,\dots,n\}$. Suppose there exists a leave-one-out primal solution with $\tilde{Y}_i < \tilde{X}_i^\top \hat{w}^{(-i)}$. By Lemma \ref{lem:dual_at_0_to_interpolation}, when $y = \tilde{X}_i^\top \hat{w}^{(-i)}$ there exists a solution to (\ref{eq:app_dual_with_swap}) with $\hat{\eta}^{\tilde{Y}_i \to \tilde{X}_i^\top \hat{w}^{(-i)}}_i = 0$. By the monotonicity of the dual solutions (Lemma \ref{lem:monotone_dual}), this immediately implies that any dual solution to the full program must satisfy 
    \[
    \hat{\eta}_i \leq \hat{\eta}^{\tilde{Y}_i \to \tilde{X}_i^\top \hat{w}^{(-i)}}_i = 0,
    \]
    as desired. 
    
    The case where $\tilde{Y}_i > \tilde{X}_i^\top \hat{w}^{(-i)}$ follows by an identical argument.
\end{proof}

We conclude this section with a proof of Theorem \ref{thm:loo_equivalence}. To aid in this proof, we introduce a number of additional pieces of notation. We let $\tilde{X} \in \mmr^{n \times p}$ denote the matrix with rows $\tilde{X}_1,\dots,\tilde{X}_n$ and $\tilde{Y} \in \mmr^n$ denote the vector with entries $\tilde{Y}_1,\dots,\tilde{Y}_n$. For any vector $v \in \mmr^k$ and set $I \subseteq \{1,\dots,k\}$ we let $(v_{I}) = (v_i)_{i \in I}$ denote the subvector consisting of the entries in $I$. Similarly, for any $I \subseteq \{1,\dots,n\}$ and $J \in \{1,\dots,p\}$ we let $\tilde{X}_{I, J} = (\tilde{X}_{i,j})_{i \in I, j \in J}$ denote the submatrix consisting of the rows in $I$ and columns in $J$. Finally, for any $k \in \mmn$ we let $[k]$ denote the set $\{1,\dots,k\}$.

We begin by presenting a preliminary lemma which controls the rank of the submatrix of $\tilde{X}$ corresponding to the interpolated points of the quantile regression.

\begin{lemma}\label{lem:interpolated_set_size}
    Assume that $\{(\tilde{X}_i,\tilde{Y}_i)\}_{i=1}^n$ are i.i.d.~and that the distribution of $\tilde{Y}_i \mid \tilde{X}_i$ is continuous. Then, with probability one all primal solutions $\hat{w}$ satisfy
    \[
\textup{rank}\left(\tilde{X}_{\{i : \tilde{Y}_i = \tilde{X}_i^\top \hat{w}\},[p]} \right) = |\{i : \tilde{Y}_i = \tilde{X}_i^\top \hat{w}\}|.
    \]
\end{lemma}
\begin{proof}
    Fix any primal solution $\hat{w}$. Let $I_{=}(\hat{w}) = \{i \in \{1,\dots,n\} : \tilde{Y}_i = \tilde{X}_i^\top \hat{w}\}$ denote the set of interpolated points. By definition, we have that 
    \[
    \tilde{Y}_{I_{=}(\hat{w})} = \tilde{X}_{I_{=}(\hat{w}), [p]}\hat{w}.
    \]
    For the sake of deriving a contradiction, suppose that $\textup{rank}\left(\tilde{X}_{\{i : \tilde{Y}_i = \tilde{X}_i^\top \hat{w}\},[p]}\right) < |I_{=}(\hat{w})|$. Let $I(\hat{w}) \subset |I_{=}(\hat{w})|$ be such that $\tilde{X}_{I(\hat{w}), [p]}$ is of maximal rank. Then, there exists a matrix $A(\tilde{X})$ such that $\tilde{X}_{I_{=}(\hat{w}) \setminus I(\hat{w}), [p]} = A(\tilde{X})\tilde{X}_{I(\hat{w}),[p]}$ and, in particular,
    \[
    \tilde{Y}_{I_{=}(\hat{w})\setminus I(\hat{w}) } = \tilde{X}_{I_{=}(\hat{w}) \setminus I(\hat{w}), [p]} \hat{w}  = A(\tilde{X})\tilde{X}_{I(\hat{w}),[p]} \hat{w} = A(\tilde{X}) \tilde{Y}_{I(\hat{w})}
    \]
    However, for any fixed sets $ I' \subset I'_{=} \subseteq [n]$, the distribution of $\tilde{Y}_{I'_{=} \setminus I' } \mid (\tilde{X}, \tilde{Y}_{I'})$ is continuous. So, taking a union bound over all choices of $I_{=}(\hat{w})$ and $I(\hat{w})$ we find that this occurs with probability zero, as claimed.
    
\end{proof}

We now prove Theorem \ref{thm:loo_equivalence}.

\begin{proof}[Proof of Theorem \ref{thm:loo_equivalence}]
    We split into two cases corresponding to the two sets of assumptions.

    \noindent \textbf{Case 1, Assumption \ref{assump:reg_condition} holds:} In this case the primal program is strictly convex. Thus, for any $i \in \{1,\dots,n\}$, the leave-one-out solution $\hat{w}^{(-i)}$ is unique and by the continuity of the distribution of $\tilde{Y}_{i} \mid \tilde{X}_{i}$ we must have that $\mmp(\tilde{Y}_{i} = \tilde{X}_{i}^\top \hat{w}^{(-i)} )=0$. By Lemma \ref{lem:dual_at_0_to_interpolation}, this implies that $\mmp(\text{Any dual solution satisfies }\hat{\eta}_{i} = 0) = 0$. The desired result then follows from Proposition \ref{prop:initial_loo_comparison} and the convexity of the space of primal and dual solutions.

    \noindent \textbf{Case 2, Assumption \ref{assump:x_pertubation} holds:} This case is considerably more involved. Without loss of generality, it is sufficient to prove the result for $i = n$. To begin, note that by the convexity of the set of primal and dual solutions and the results of Proposition \ref{prop:initial_loo_comparison} and Lemma \ref{lem:dual_at_0_to_interpolation}, it is sufficient to show that $\mmp(\text{There exists a dual solution with } \hat{\eta}_n = 0) = 0$. Recall that all dual solutions are supported on the domain $[-(1-\tau),\tau]^n$ (cf.~(\ref{eq:general_dual}) in the main text). Fix any dual solution $\hat{\eta}$ and let $I_{\text{int.}}(\hat{\eta}) = \{i \in \{1,\dots,n\} : -(1-\tau) < \hat{\eta}_i < \tau \}$ denote the set of coordinates that lie in the interior of the feasible region. We need to show that with probability one $\hat{\eta}_{I_{\text{int.}}(\hat{\eta})}$ has all non-zero entries. 
    
    To do this, let $(\hat{w},\hat{r})$ denote any corresponding primal solution. Recall that the Lagrangian for this optimization problem is 
     \[
    L(w,r,\eta) = \sum_{i=1}^n \ell_{\tau}(r_i) + \sum_{i=1}^n \eta_i(\tilde{Y}_i - \tilde{X}_i^\top w - r_i) + \sum_{j=1}^p \lambda_j w_j^2.
    \]
    Let $J_{+} = \{j \in \{1,\dots,d\} : \lambda_j > 0\}$ denote the set of coordinates which receive positive regularization. Let $\Lambda_{J_+} = \text{diag}((\lambda_j)_{j \in J_{+}})$ be the diagonal matrix with diagonal entries $(\lambda_j)_{j \in J_{+}}$. Differentiating $L$ with respect to $w$ gives us that
     \[
    \tilde{X}^\top \hat{\eta} = (2\lambda_j \hat{w}_j)_{j=1}^p \iff  \tilde{X}_{I_{\text{int.}}(\hat{\eta}),J_+^c}^\top \hat{\eta}_{I_{\text{int.}}(\hat{\eta})} = - \tilde{X}_{I_{\text{int.}}(\hat{\eta})^c,J_+^c}^\top \hat{\eta}_{I_{\text{int.}}(\hat{\eta})^c}  \text { and } \hat{w}_{J_+} = \frac{1}{2} \Lambda_{J_+}^{-1} \tilde{X}_{[n],J_+}^\top \hat{\eta}.
    \]
    Moreover, differentiating $L$ with respect to $r_i$ gives the first-order condition $\hat{\eta}_i \in \partial \ell_{\tau}(\hat{r}_i)$. Now, recall that by construction $\hat{r}_i = \tilde{Y}_i - \tilde{X}_i \hat{w}$. Combining these two conditions, we find that for all $i \in I_{\text{int.}}(\hat{\eta})$, $\tilde{Y}_i = \tilde{X}_i \hat{w} $ and thus,
    \begin{align*}
    & \tilde{Y}_{I_{\text{int.}}(\hat{\eta})}  = \tilde{X}_{I_{\text{int.}}(\hat{\eta}),[p]}\hat{w} = \tilde{X}_{I_{\text{int.}}(\hat{\eta}),J_+^c}\hat{w}_{J_+^c} + \frac{1}{2} \tilde{X}_{I_{\text{int.}}(\hat{\eta}),J_+} \Lambda_{J_+}^{-1} \tilde{X}_{[n],J_+}^\top \hat{\eta}\\
    & = \tilde{X}_{I_{\text{int.}}(\hat{\eta}),J_+^c}\hat{w}_{J_+^c} + \frac{1}{2} \tilde{X}_{I_{\text{int.}}(\hat{\eta}),J_+} \Lambda_{J_+}^{-1} X_{I_{\text{int.}}(\hat{\eta}),J_+}^\top \hat{\eta}_{I_{\text{int.}}(\hat{\eta})} + \frac{1}{2} \tilde{X}_{I_{\text{int.}}(\hat{\eta}),J_+} \Lambda_{J_+}^{-1} \tilde{X}_{I_{\text{int.}}(\hat{\eta})^c,J_+}^\top \hat{\eta}_{I_{\text{int.}}(\hat{\eta})^c}.
    \end{align*}
    Combining all of the above observations, we arrive at the system of equations
    \begin{align*}
    & \left[\begin{matrix}
        \frac{1}{2} \tilde{X}_{I_{\text{int.}}(\hat{\eta}),J_+} \Lambda_{J_+}^{-1} \tilde{X}_{I_{\text{int.}}(\hat{\eta}),J_+}^\top  & \tilde{X}_{I_{\text{int.}}(\hat{\eta}),J_+^c} \\
        \tilde{X}_{I_{\text{int.}}(\hat{\eta}),J_+^c}^\top & \pmb{0}_{|J_{+}^c|, |J_{+}^c|} 
    \end{matrix} \right] \left[ \begin{matrix} \hat{\eta}_{I_{\text{int.}}(\hat{\eta})} \\ \hat{w}_{J_+^c} \end{matrix} \right] = \left[\begin{matrix}
        \tilde{Y}_{I_{\text{int.}}(\hat{\eta})} - \frac{1}{2} \tilde{X}_{I_{\text{int.}}(\hat{\eta}),J_+} \Lambda_{J_+}^{-1} \tilde{X}_{I_{\text{int.}}(\hat{\eta})^c,J_+}^\top \hat{\eta}_{I_{\text{int.}}(\hat{\eta})^c}\\
       -\tilde{X}_{I_{\text{int.}}(\hat{\eta})^c,J_+^c}^\top \hat{\eta}_{I_{\text{int.}}(\hat{\eta})^c} 
    \end{matrix} \right],
    \end{align*}
    where $\pmb{0}_{|J_{+}^c|, |J_{+}^c|} \in \mmr^{|J_{+}^c| \times |J_{+}^c|}$ denotes the zero matrix. We claim that with probability one the matrix appearing on the left-hand side above is invertible. To see this, suppose that $(v_1,v_2) \in \mmr^{|I_{\text{int.}}(\hat{\eta})| \times |J_+^c|}$ is in the kernel of this matrix, i.e., suppose that 
    \begin{equation}\label{eq:kernel_matrix_system}
    \frac{1}{2} \tilde{X}_{I_{\text{int.}}(\hat{\eta}),J_+} \Lambda_{J_+}^{-1} \tilde{X}_{I_{\text{int.}}(\hat{\eta}),J_+}^\top v_1 +   \tilde{X}_{I_{\text{int.}}(\hat{\eta}),J_+^c} v_2 = 0 \ \ \ \text{ and } \ \ \ \tilde{X}^{\top}_{I_{\text{int.}} (\hat{\eta}),J_+^c} v_1 = 0.
    \end{equation}
    Taking the inner product of the first equation with $v_1$, we find that 
    \begin{align*}
    & 0 = v_1^\top \frac{1}{2} \tilde{X}_{I_{\text{int.}}(\hat{\eta}),J_+} \Lambda_{J_+}^{-1} \tilde{X}_{I_{\text{int.}}(\hat{\eta}),J_+}^\top v_1 +   v_1^\top \tilde{X}_{I_{\text{int.}}(\hat{\eta}),J_+^c} v_2 =  v_1^\top \frac{1}{2} \tilde{X}_{I_{\text{int.}}(\hat{\eta}),J_+} \Lambda_{J_+}^{-1} \tilde{X}_{I_{\text{int.}}(\hat{\eta}),J_+}^\top v_1\\
    & \implies \tilde{X}_{I_{\text{int.}}(\hat{\eta}),J_+}^\top v_1 = 0.
    \end{align*}
    Combining this fact with the second equation in (\ref{eq:kernel_matrix_system}) gives $v_1^\top \tilde{X}_{I_{\text{int.}}(\hat{\eta}),[p]} = 0$. By Lemma \ref{lem:interpolated_set_size}, with probability one $\tilde{X}_{I_{\text{int.}}(\hat{\eta}),[p]}$ has linearly independent rows, and thus, with probability one we must have that $v_1 = 0$. Returning to  (\ref{eq:kernel_matrix_system}), this implies that $\tilde{X}_{I_{\text{int.}}(\hat{\eta}), J_+^c} v_2 = 0$. Lemma \ref{lem:full_rank_interp_unreg} below shows that with probability one $\tilde{X}_{I_{\text{int.}}(\hat{\eta}), J_+^c}$ is of rank $|J_+^c|$. Thus, with probability one we must have that $v_2 = 0$. This proves the desired claim.

    Applying this claim, we find that 
    \begin{align*}
     & \left[ \begin{matrix} \hat{\eta}_{I_{\text{int.}}(\hat{\eta})} \\ \hat{w}_{J_+^c} \end{matrix} \right]\\
     & =  \left[\begin{matrix}
        \frac{1}{2} \tilde{X}_{I_{\text{int.}}(\hat{\eta}),J_+} \Lambda_{J_+}^{-1} \tilde{X}_{I_{\text{int.}}(\hat{\eta}),J_+}^\top  & \tilde{X}_{I_{\text{int.}}(\hat{\eta}),J_+^c} \\
        \tilde{X}_{I_{\text{int.}}(\hat{\eta}),J_+^c}^\top & \pmb{0}_{|J_{+}^c|, |J_{+}^c|} 
    \end{matrix} \right]^{-1} \left[\begin{matrix}
        \tilde{Y}_{I_{\text{int.}}(\hat{\eta})} - \frac{1}{2} \tilde{X}_{I_{\text{int.}}(\hat{\eta}),J_+} \Lambda_{J_+}^{-1} \tilde{X}_{I_{\text{int.}}(\hat{\eta})^c,J_+}^\top \hat{\eta}_{I_{\text{int.}}(\hat{\eta})^c}\\
       -\tilde{X}_{I_{\text{int.}}(\hat{\eta})^c,J_+^c}^\top \hat{\eta}_{I_{\text{int.}}(\hat{\eta})^c} 
    \end{matrix} \right]
    \end{align*}
    
    Now, let us consider the behavior of the random variable appearing on the last line above when $I_{\text{int.}}(\hat{\eta})$ is a fixed set and $\hat{\eta}_{I_{\text{int.}}(\hat{\eta})^c}$ is a fixed vector. By Lemma \ref{lem:interp_set_size_less_p} below, $|I_{\text{int.}}(\hat{\eta})| <n$. So, fix any set $\emptyset \subset I_{\text{int.}} \subset [n]$ and vector $\eta_{I_{\text{int.}}^c} \in \{-(1-\tau),\tau\}^{|I_{\text{int.}}^c|}$ with the property that the matrix inverse above exists and consider the behaviour of the random variable
    \begin{equation}\label{eq:fixed_set_eta_formula}
    \left[\begin{matrix}
        \frac{1}{2} \tilde{X}_{I_{\text{int.}},J_+} \Lambda_{J_+}^{-1} \tilde{X}_{I_{\text{int.}},J_+}^\top  & \tilde{X}_{I_{\text{int.}},J_+^c} \\
        \tilde{X}_{I_{\text{int.}},J_+^c}^\top & \pmb{0}_{|J_{+}^c|, |J_{+}^c|} 
    \end{matrix} \right]^{-1} \left[\begin{matrix}
        \tilde{Y}_{I_{\text{int.}}} - \frac{1}{2} \tilde{X}_{I_{\text{int.}},J_+} \Lambda_{J_+}^{-1} \tilde{X}_{I_{\text{int.}}^c,J_+}^\top {\eta}_{I_{\text{int.}}^c}\\
       -\tilde{X}_{I_{\text{int.}}^c,J_+^c}^\top {\eta}_{I_{\text{int.}}^c} 
    \end{matrix} \right].
    \end{equation}
    Recall that by assumption the covariates can be written as $\tilde{X} = Z + \xi$ where $\xi \in \mmr^{n\times p}$ has i.i.d., continuously distributed entries independent of $Z \in \mmr^{n\times p}$ and $\tilde{Y}$. Additionally, recall that $\tilde{Y} \mid \tilde{X}$ is continuously distributed. Applying these facts, we find that conditional on the  $(\tilde{X}_{I_{\text{int.}},J_+}, \tilde{X}_{I_{\text{int.}},J^c_+}, Z)$ the vector appearing in (\ref{eq:fixed_set_eta_formula}) is continuously distributed. In particular, conditional on $(\tilde{X}_{I_{\text{int.}},J_+}, \tilde{X}_{I_{\text{int.}},J^c_+}, $Z$)$ and the event that the matrix appearing in this expression is invertible, we find that with probability one the entire random variable appearing in (\ref{eq:fixed_set_eta_formula}) is continuously distributed, and thus, with probability one, has all non-zero entries. Taking a union bound over the values of $I_{\text{int.}}$ and ${\eta}_{I_{\text{int.}}^c}$ gives the desired result.

\end{proof}

We close this section with two lemmas that were helpful in the proof of Theorem \ref{thm:loo_equivalence}.

\begin{lemma}\label{lem:interp_set_size_less_p}
    Suppose that $\{(\tilde{X},\tilde{Y})\}_{i=1}^n$ are i.i.d.~and that Assumption \ref{assump:x_pertubation} holds. Then,
    \[
    \mmp(\text{For all dual solutions $\hat{\eta}$, } |\{i : -(1-\tau) < \hat{\eta}_i < \tau \}| < n) = 1.
    \]
\end{lemma}
\begin{proof}
    Let $(\hat{w},\hat{r},\hat{\eta})$ denote any primal-dual solution, i.e., any saddle point of the Lagrangian 
    \[
    L(w,r,\eta) = \sum_{i=1}^n \ell_{\tau}(r_i) + \sum_{i=1}^n \eta_i(\tilde{Y}_i - \tilde{X}_i^\top w - r_i) + \sum_{j=1}^p \lambda_j w_j^2.
    \]
   Differentiating $L$ with respect to $r$, we must have that for all $i \in [n]$, $\hat{\eta}_i \subseteq \partial \ell_{\tau}(\hat{r}_i)$. Recalling the constraint $\hat{r}_i = \tilde{Y}_i - \tilde{X}_i \hat{w}$, this in particular implies that 
   \[
   \{i : -(1-\tau) < \hat{\eta}_i < \tau \} \subseteq \{i : \tilde{Y}_i = \tilde{X}_i \hat{w}\}.
   \]
   Now, by Lemma  \ref{lem:interpolated_set_size} we have that with probability one all primal solutions are such that $| \{i : \tilde{Y}_i = \tilde{X}_i \hat{w}\}| \leq p$. Thus, with probability one all dual solutions satisfy
   \[
   |\{i : -(1-\tau) < \hat{\eta}_i < \tau \}| \leq p.
   \]
   Since by assumption $p < n$, this gives the desired result.  
\end{proof}

\begin{lemma}\label{lem:full_rank_interp_unreg}
    Suppose that $\{(\tilde{X},\tilde{Y})\}_{i=1}^n$ are i.i.d.~and that Assumption \ref{assump:x_pertubation} holds. For any dual solution $\hat{\eta}$, let $I_{\text{int.}}(\hat{\eta}) = \{i : -(1-\tau) < \hat{\eta}_i < \tau\}$ denote the set of coordinates of $\hat{\eta}$ that lie in the interior of the feasible region. Let $J_+^c = \{j : \lambda_j = 0\}$ denote the set of unregularized coordinates of the primal variables. Then,
    \[
    \mmp\left(\text{For all dual solutions $\hat{\eta}$, }  \textup{rank}(\tilde{X}_{I_{\text{int.}}(\hat{\eta}), J_+^c}) = |J_+^c|\right) = 1.
    \]
\end{lemma}
\begin{proof}
    By our assumptions on the distribution of $\tilde{X}$, we have that with probability one $\text{rank}(X_{A,B}) = \min\{|A|, |B|\}$ for all $A \subseteq [n]$ and $B \subseteq [p]$. Thus, it is sufficient to show that with probability one all dual solutions are such that $|J_+^c| \leq |I_{\text{int.}}(\hat{\eta})$|. 

    Fix any dual solution $\hat{\eta}$ and corresponding primal solution $(\hat{w},\hat{r})$. Recall that $(\hat{w},\hat{r},\hat{\eta})$ is a saddle point of the Lagrangian
     \[
    L(w,r,\eta) = \sum_{i=1}^n \ell_{\tau}(r_i) + \sum_{i=1}^n \eta_i(\tilde{Y}_i - \tilde{X}_i^\top w - r_i) + \sum_{j=1}^p \lambda_j w_j^2.
    \]
    Differentiating $L$ with respect to $r$ gives us that $\hat{\eta} \in [-(1-\tau),\tau)]^n$. Moreover, differentiating $L$ with respect to $w$ gives 
    \[
    \tilde{X}_{[n], J_+^c}^\top \hat{\eta} = 0 \iff \tilde{X}_{I_{\text{int.}}(\hat{\eta}), J_+^c}^\top \hat{\eta}_{I_{\text{int.}}(\hat{\eta})} = -\tilde{X}_{I_{\text{int.}}(\hat{\eta})^c, J_+^c}^\top \hat{\eta}_{I_{\text{int.}}(\hat{\eta})^c}.
    \]
    For the sake of deriving a contradiction, suppose that $|I_{\text{int.}}(\hat{\eta})| < J_+^c$. Let $J_{\text{sub.}}(\hat{\eta}) \subseteq J_+^c$ be a subset of size $|J_{\text{sub.}}(\hat{\eta})| = |I_{\text{int.}}(\hat{\eta})|$. Rearranging the above, we have that 
    \begin{align*}
    & \bigg( \hat{\eta}_{I_{\text{int.}}(\hat{\eta})} = -(\tilde{X}_{I_{\text{int.}}(\hat{\eta}),J_{\text{sub.}}(\hat{\eta})}^\top)^{-1} \tilde{X}_{I_{\text{int.}}(\hat{\eta})^c ,J_{\text{sub.}}(\hat{\eta})}^\top \hat{\eta}_{I_{\text{int.}}(\hat{\eta})^c},\\
     & \hspace{2cm} \text{and }  \ \  \tilde{X}_{I_{\text{int.}}(\hat{\eta}), J_{\text{sub.}}(\hat{\eta})^c}^\top \hat{\eta}_{I_{\text{int.}}(\hat{\eta})} = -\tilde{X}_{I_{\text{int.}}(\hat{\eta})^c, J_{\text{sub.}}(\hat{\eta})^c}^\top \hat{\eta}_{I_{\text{int.}}(\hat{\eta})^c} \bigg)\\
    & \implies \tilde{X}_{I_{\text{int.}}(\hat{\eta}), J_{\text{sub.}}(\hat{\eta})^c}^\top (\tilde{X}_{I_{\text{int.}}(\hat{\eta}),J_{\text{sub.}}(\hat{\eta})}^\top)^{-1} \tilde{X}_{I_{\text{int.}}(\hat{\eta})^c ,J_{\text{sub.}}(\hat{\eta})}^\top \hat{\eta}_{I_{\text{int.}}(\hat{\eta})^c} = \tilde{X}_{I_{\text{int.}}(\hat{\eta})^c, J_{\text{sub.}}(\hat{\eta})^c}^\top \hat{\eta}_{I_{\text{int.}}(\hat{\eta})^c} \numberthis \label{eq:last_eta_equal_cond}
    \end{align*}

Now, recall that by Lemma \ref{lem:interp_set_size_less_p}, $|I_{\text{int.}}(\hat{\eta})| < n$. Moreover, recall that by assumption the covariates can be written as $\tilde{X} = Z + \xi$ where $\xi \in \mmr^{n\times p}$ has i.i.d., continuously distributed entries independent of $Z \in \mmr^{n\times p}$. Thus, in particular, for any fixed sets $I_{\text{int.}} \subset [n]$ and $J_{\text{sub.}} \subseteq [p]$ with $|J_{\text{sub.}}| = |I_{\text{int.}}|$ and any fixed vector $\eta_{I_{\text{int.}}^c} \in \{-(1-\tau),\tau\}^n$, the random variable $\tilde{X}_{I_{\text{int.}}^c, J_{\text{sub.}}^c}^\top \eta_{I_{\text{int.}}^c}$ has a continuous distribution conditional on  $(Z,\tilde{X}_{I_{\text{int.}}, J_{\text{sub.}}^c}^\top (\tilde{X}_{I_{\text{int.}},J_{\text{sub.}}}^\top)^{-1} \tilde{X}_{I_{\text{int.}}^c ,J_{\text{sub.}}}^\top {\eta}_{I_{\text{int.}}^c}$). Therefore,
\[
\mmp\left(\tilde{X}_{I_{\text{int.}}, J_{\text{sub.}}^c}^\top (\tilde{X}_{I_{\text{int.}},J_{\text{sub.}}}^\top)^{-1} \tilde{X}_{I_{\text{int.}}^c ,J_{\text{sub.}}}^\top {\eta}_{I_{\text{int.}}^c} = \tilde{X}_{I_{\text{int.}}^c, J_{\text{sub.}}^c}^\top {\eta}_{I_{\text{int.}}^c}\right) = 0.
\]
Taking a union bound over all possible choices of the sets $I_{\text{int.}}$ and $J_{\text{sub.}}$ and vector $\eta_{I_{\text{int.}}^c}$, we find that with probability one no dual solution can satisfy (\ref{eq:last_eta_equal_cond}). This proves the desired result.
\end{proof}

\section{Proofs for Section \ref{sec:theory}}\label{sec:app_main_asymp}

The bulk of this section is devoted to a proof of Theorem \ref{thm:main_asym_consistency}. Proofs of Corollaries \ref{corr:loo_cov_consistency} and \ref{corr:quantile_consistency}  are then given at the end. In what follows, we use $X \in \mmr^{n \times p}$ to denote the matrix with rows $X_1,\dots,X_n$ and $Y \in \mmr^n$ and $\epsilon \in \mmr^n$ to denote the vectors with entries $(Y_1,\dots,Y_n)$ and $(\epsilon_1,\dots,\epsilon_n)$, respectively. With some abuse of notation, we will often use $\sqrt{d}\tilde{\beta}_1$ to denote a generic sample from the distribution of population coefficients $P_{\beta}$. Additionally, for any convex function $f: \mmr^k \to \mmr$, $x \in \mmr^k$ and $\rho > 0$ we recall the definition of the Moreau envelope,
\[
e_f(x;\rho) = \min_{v \in \mmr^k} \frac{1}{2\rho} \|x - v\|^2_2 + f(v).
\]
For ease of notation, we will also define the Moreau envelope at $\rho = 0$ using the continuous extension $e_f(x;0) = f(x)$ (cf.~Lemma \ref{lem:envelope_extension} below). Finally, for $f: \mmr^k \to \mmr$ we recall the definition of the convex conjugate,
\[
f^*(x) = -\inf_{v \in \mmr^k} f(v) - v^\top x.
\]
We have the following assumptions on the regularizer and population coefficients.

\begin{assumption}\label{assump:reg_high_dim_assumptions}
    The distribution of population coefficients $P_{\beta}$ has two bounded moments. Moreover, the regularization function and $P_{\beta}$ are such that:
   \begin{enumerate}
   \item $\mathcal{R}_d$ is convex. Moreover, for all $\beta \in \mmr^d$, $\mathcal{R}(\beta) \geq 0$ and $\mathcal{R}(0) = 0$.
   \item For any $C>0$, the subderivatives of $\mathcal{R}_d$ are bounded as
   \[
   \sup_{d \in \mmn} \sup_{\|\beta\|_2 \leq C} \frac{1}{d}\|\partial \mathcal{R}_d(\beta)\|_2 < \infty.
   \]
       \item For any $\gamma \in (0,2/\pi)$ there exists a function $\nu : \mmr \to \mmr$ with the property that for any  $c \in \mmr$, $\rho \geq 0$, and $h \sim \mathcal{N}(0,I_d)$,
   \[
   \lim_{(d,n \to \infty,\ d/n \to \gamma)} \frac{1}{d} e_{\mathcal{R}_d(\cdot/\sqrt{n})  }\left(ch_i + \sqrt{n}\tilde{\beta}_i ; \rho \right) \stackrel{\mmp}{=} \mme[e_{\nu}(ch_1 + \gamma\sqrt{d}\tilde{\beta}_1 ; \rho) ] < \infty.
   \]
   \item 
   The function $(c,\rho) \mapsto \mme[e_{\nu}(ch_1 + \gamma\sqrt{d}\tilde{\beta}_1 ; \rho) ]$ is jointly continuous on $\mmr \times \mmr_{\geq 0}$.
   \item 
   For any $\rho > 0$, $\partial_x  e_{\nu^*}(x; \rho)$ and $\partial^2_x  e_{\nu^*}(x;\rho)$ exist almost everywhere and satisfy the equations
   \[
   \frac{d}{dc}\mme[ e_{\nu^*}(ch_1 + \rho \gamma\sqrt{d}\tilde{\beta}_1 ; \rho) ] = \mme[h_1 \partial_x  e_{\nu^*}(ch_1 + \rho \gamma\sqrt{d}\tilde{\beta}_1 ; \rho) ] = c\mme[ \partial_x^2 e_{\nu^*}(ch_1 + \rho \gamma\sqrt{d}\tilde{\beta}_1 ; \rho) ].
   \]
   \item
   For any compact set $A \subseteq \mmr_{> 0}$ and constant $C_{\rho}>0$, 
   \[
   \inf_{c \in A, 0 < \rho \leq C_{\rho}} \mme[ \partial_x^2 e_{\nu^*}(ch_1 + \rho \gamma\sqrt{d}\tilde{\beta}_1 ; \rho) ] > 0.
   \]
   \end{enumerate}
\end{assumption}

All of the assumptions above are fairly generic and will hold for most common separable regularizers. As an example to illustrate this, the following lemma verifies that all these conditions are satisfied by $L_1$ and $L_2$ regularization.

\begin{lemma}\label{lem:reg_verification}
    Assume $P_{\beta}$ has four bounded moments. Then, for any $\lambda \geq 0$ the conditions of Assumption \ref{assump:reg_high_dim_assumptions} are met for $\mathcal{R}_d(\beta) = \sqrt{d} \lambda \|\beta\|_1$ and $\mathcal{R}_d(\beta) = d\lambda \|\beta\|_2^2$.
\end{lemma}
\begin{proof}
    For $\mathcal{R}_d(\beta) = d\lambda \|\beta\|_2^2$ define $\nu(b) = \lambda  \gamma b^2$. By a direct calculation, we have that 
    \[
     e_{\nu}(x;\rho) = \frac{\lambda \gamma x^2}{1+2\lambda \gamma \rho}, \ \ \ \nu^*(b) = \frac{b^2}{4\lambda \gamma}, \ \ \text{ and } \ \ e_{\nu^*}(x;\rho) = \frac{x^2}{4\lambda \gamma +  2\rho}.
    \]
    Parts one, two, and four of Assumption \ref{assump:reg_high_dim_assumptions} are immediate. Part 3 follows from the law of large numbers.  Part five follows by the dominated convergence theorem and Stein's lemma (Lemma 1 of \citet{Stein1981}). Part six is also immediate since $\partial_x^2 e_{\nu^*}(x;\rho) = (2\lambda \gamma + \rho)^{-1} > 0$.

    On the other hand, suppose $\mathcal{R}_d(\beta) = \sqrt{d}\lambda \|\beta\|_1$. Define $\nu(b) = \lambda \sqrt{\gamma} |b|$. Then,
    \[
     e_{\nu}(x;\rho) = 
     \begin{cases}
     \lambda \sqrt{\gamma} x-\frac{(\lambda\sqrt{\gamma})^2\rho}{2},\ & x > \lambda \sqrt{\gamma}\rho,\\
     \frac{x^2}{2\rho},\ & |x| \leq \lambda \sqrt{\gamma} \rho,\\
     -\lambda \sqrt{\gamma} x -\frac{(\lambda \sqrt{\gamma})^2\rho}{2},\ & x < -\lambda \sqrt{\gamma}\rho,
     \end{cases} \ \ \ \   \nu^*(b) = \begin{cases}
         0,\ & |x| \leq \lambda\sqrt{\gamma},\\
         \infty,\ & |x| > \lambda\sqrt{\gamma},
     \end{cases}
    \]
    and
    \[
    e_{\nu^*}(x;\rho) = \begin{cases}
         0,\ & |x| \leq \lambda\sqrt{\gamma},\\
         \frac{(|x| - \lambda\sqrt{\gamma})^2}{2\rho},\ & |x| > \lambda\sqrt{\gamma}.
     \end{cases}.
    \]
    Moreover, one can verify that $e_{\nu^*}(x;\rho)$ is twice piecewise continuously differentiable with 
    \[
    \partial_x e_{\nu^*}(x;\rho) = \begin{cases}
         0,\ & |x| \leq \lambda\sqrt{\gamma},\\
         \frac{x - \text{sgn}(x)\lambda\sqrt{\gamma}}{\rho},\ & |x| > \lambda\sqrt{\gamma},
     \end{cases} \ \   \text{ and } \ \ \partial^2_x e_{\nu^*}(x;\rho) = \begin{cases}
         0,\ & |x| < \lambda\sqrt{\gamma},\\
         \frac{1}{\rho},\ & |x| > \lambda\sqrt{\gamma}.
     \end{cases} 
    \]
    The desired results once again follow by the law of large numbers, the dominated convergence theorem and Stein's lemma.
    
\end{proof}

Our main point of study is the joint min-max formulation of the quantile regression,
\[
\max_{\eta} \min_{\beta_0,\beta,r} \frac{1}{n}\sum_{i=1}^n \ell(r_i) + \frac{1}{n} \eta^\top (Y - \beta_0 \pmb{1}_n - X\beta - r) + \frac{1}{n} \mathcal{R}_d(\beta).
\]
Letting $u = \beta - \tilde{\beta}$ and re-writing $\mathcal{R}_d$ in terms of the convex conjugate, this can be equivalently formulated as
\begin{equation}\label{eq:app_min_max_qr}
\max_{\eta, s} \min_{\beta_0,u,r} \frac{1}{n}\sum_{i=1}^n \ell(r_i) + \frac{1}{n} \eta^\top (\epsilon - \beta_0 \pmb{1}_n - Xu - r)  + \frac{1}{\sqrt{n}}  s^\top(\tilde{\beta} + u) -  \frac{1}{n} \mathcal{R}_d^*(\sqrt{n}s).
\end{equation}
To prove Theorem \ref{thm:main_asym_consistency} we will need to study the solutions of this optimization program. We proceed in four main steps. First, in Section \ref{sec:prelims_to_theorem} we give a number of preliminary simplifications which demonstrate that the optimization domain can be restricted to a compact set. Section \ref{sec:auxiliary_reduc} then begins our main study of (\ref{eq:app_min_max_qr}). We show that the solutions to this problem are characterized by an auxiliary optimization program in which the matrix $X$ is replaced by vector-valued Gaussian random variables. Moreover, we additionally demonstrate that the solutions to this auxiliary problem are themselves characterized by the deterministic asymptotic program
\begin{equation}\label{eq:asymp_program}
\begin{split}
    & \min_{(|\beta_0| \leq C_{\beta_0}, 0 \leq M_{u} \leq C_u, 0 \leq \rho_1 \leq C_1)} \max_{(0 < \rho_2 \leq C_2, c_{\eta} \leq M_{\eta} \leq  C_{\eta} )} \bigg( \mme\left[ e_{\ell_{\tau}}\left( M_ug + \epsilon - \beta_0 ; \frac{\rho_1}{M_{\eta}} \right) \right] - \frac{M_{\eta}^2M_u\gamma}{2\rho_2}\\
    & \hspace{5cm}  + \frac{M_{\eta}\rho_1}{2} - \frac{M_u\rho_2}{2} - \mme\left[ e_{\nu}\left( \frac{M_uM_{\eta}}{\rho_2}h_1 + \gamma \sqrt{d}\tilde{\beta}_1; \frac{M_u}{\rho_2} \right) \right] \bigg),
\end{split}
\end{equation}
where  
  $(\beta_0, M_u, \rho_1, \rho_2, M_{\eta})$ are the optimization variables, ($C_{\beta_0}, C_u, C_1, C_2, c_{\eta}, C_{\eta})$ are constants that we will define shortly, and $g_1, h_1 \sim \mathcal{N}(0,1)$  are independent of $\tilde{\beta}_1$ and $\epsilon_1$. The solutions to this asymptotic program are characterized in Section \ref{sec:asymp_program}. Section \ref{sec:final_theorem_proof} then gives a proof of Theorem \ref{thm:main_asym_consistency} and  Section \ref{sec:corollaries} gives proofs of Corollaries \ref{corr:loo_cov_consistency} and \ref{corr:quantile_consistency}.

  Our overall analysis framework is based on the work of \citet{Thram2018}. In what follows, we will focus on the aspects of the analysis that are new to our work and omit the proofs of some results that are minor variations of those appearing in that article.

\subsection{Preliminaries}\label{sec:prelims_to_theorem}

We begin our proof of Theorem \ref{thm:main_asym_consistency} by giving three lemmas which bound the domains of the optimization variables appearing in (\ref{eq:app_min_max_qr}). In what follows, we use the notation $(\hat{\beta}_0,\hat{u},\hat{r},\hat{\eta},\hat{s})$ to denote a generic primal-dual solution to (\ref{eq:app_min_max_qr}), where $(\hat{\beta}_0,\hat{u},\hat{r})$ and $(\hat{\eta},\hat{s})$ are the primal and dual solutions, respectively. Our first result bounds the range of $\hat{\eta}$.

\begin{lemma}\label{lem:dual_norm_bounds}
    Under the assumptions of Theorem \ref{thm:main_asym_consistency}, there exist $C_{\eta} > c_{\eta} > 0$ such that 
    \[
    \mmp\left(\text{For all dual solutions to (\ref{eq:app_min_max_qr}), } \sqrt{n} c_{\eta} \leq \|\hat{\eta}\|_2 \leq \sqrt{n} C_{\eta}\right) \to 1. 
    \]
\end{lemma}
\begin{proof}
    Let $(\hat{\beta}_0,\hat{u},\hat{r},\hat{\eta},\hat{s})$ denote any primal-dual solution. First, note that differentiating (\ref{eq:app_min_max_qr}) with respect to $r$ gives us that for all $i \in \{1,\dots,n\}$, $\hat{\eta}_i \in \partial \ell_{\tau}(\hat{r}_i) \subseteq [-(1-\tau),\tau]$. So, taking $C_{\eta} = \max\{(1-\tau),\tau\}$ gives the upper bound. 
    
    To get the lower bound, note that differentiating (\ref{eq:app_min_max_qr}) with respect to $\eta_i$ gives us that $\hat{r}_i = Y_i - \hat{\beta}_0 - X_i^\top \hat{\beta}$. Combining this with the fact that $\hat{\eta}_i \in \partial\ell_{\tau}(\hat{r}_i)$, we find that 
    \[
    Y_i \neq \hat{\beta}_0 + X_i^\top \hat{\beta} \implies \hat{\eta}_i \in \{-(1-\tau),\tau)\},
    \]
    and thus,
    \[
    \|\hat{\eta}\|_2 \geq \min\{(1-\tau),\tau \} \sqrt{n - |\{i : Y_i = \hat{\beta}_0 + X_i^\top \hat{\beta}\}|}.
    \]
    By Lemma \ref{lem:interpolated_set_size}, with probability one all primal solutions interpolate at most $d+1$ points. Thus, we must have that with probability one all dual solutions satisfy $\|\hat{\eta}\|_2 \geq \sqrt{n-d-1}\min\{(1-\tau),\tau\}$ and so setting $c_{\eta} = (1/2)\sqrt{1-\gamma}\min\{(1-\tau),\tau\}$ gives the desired result.

\end{proof}

Our next lemma gives a similar set of bounds on $\hat{u}$ and $\hat{\beta}_0$. For ease of notation, we state this result in terms of the original primal variable $\hat{\beta} = \hat{u} + \tilde{\beta}$.

\begin{lemma}\label{lem:primal_bound}
Suppose the assumptions of Theorem \ref{thm:main_asym_consistency} hold. Then, there exist constants $C_u, C_{\beta_0} > 0$ such that 
\[
\mmp\left(\text{For all primal solutions to (\ref{eq:app_min_max_qr}), } \|\hat{\beta} - \tilde{\beta}\|_2 \leq C_u \text{ and } |\hat{\beta}_0| \leq C_{\beta_0} \right) \to 1.
\]  
\end{lemma}
\begin{proof}
    Let $(\hat{\beta}_0,\hat{\beta})$ denote any primal solution. By the law of large numbers and the optimality of $(\hat{\beta}_0,\hat{\beta})$,  we have that 
    \begin{align*}
        \mme[\ell_{\tau}(Y_1)] & \geq  \frac{1}{n} \sum_{i=1}^n \ell_{\tau}(Y_i) - o_{\mmp}(1) \geq \frac{1}{n} \sum_{i=1}^n \ell_{\tau}(Y_i - X_i^\top \hat{\beta} - \hat{\beta}_0) + \mathcal{R}_d(\hat{\beta}) - o_{\mmp}(1)\\
        & \geq \min\{1-\tau,\tau\} \frac{1}{n} \sum_{i=1}^n |X_i^\top(\hat{\beta} - \tilde{\beta}) + \hat{\beta}_0| -\min\{1-\tau,\tau\} \frac{1}{n}\sum_{i=1}^n|\epsilon_i| - o_{\mmp}(1)\\
        & \geq \min\{1-\tau,\tau\} \max\{\|\hat{\beta} - \tilde{\beta}\|_2, |\hat{\beta}_0|\} \inf_{(\|u\|_2 \leq 1, |\beta_0| \leq 1, \max\{\|u\|_2, |\beta_0|\} = 1)} \frac{1}{n} \sum_{i=1}^n |X_i^\top u + \beta_0|\\
        & \ \ \ \ \ - \min\{1-\tau,\tau\} \mme[\epsilon_1] - o_{\mmp}(1).
    \end{align*}
    Lemma \ref{lem:gaussian_abs_lower} below shows that 
    \[
   \liminf_{(n,d \to \infty, d/n \to \gamma)} \inf_{(\|u\|_2 \leq 1, |\beta_0| \leq 1, \max\{\|u\|_2, |\beta_0|\} = 1)} \frac{1}{n} \sum_{i=1}^n |X_i^\top u + \beta_0| \stackrel{\mmp}{\geq} \sqrt{\frac{2}{\pi}} - \sqrt{\gamma}.
    \]
    Applying this to the above, we conclude that 
     \[
     \max\{\|\hat{\beta} - \tilde{\beta}\|_2, |\hat{\beta}_0|\} \leq \frac{\mme[\ell_{\tau}(Y_1)] +  \min\{1-\tau,\tau\} \mme[\epsilon_1]}{\min\{1-\tau,\tau\}(\sqrt{2/\pi} - \sqrt{\gamma})} + o_{\mmp}(1),
     \]
     where it should be understood that the $o_{\mmp}(1)$ term on the right hand side is uniform over all primal solutions. Taking \[
     C_u = C_{\beta_0}= 2\frac{\mme[\ell_{\tau}(Y_1)] +  \min\{1-\tau,\tau\} \mme[\epsilon_1]}{\min\{1-\tau,\tau\}(\sqrt{2/\pi} - \sqrt{\gamma})} 
     \]
     gives the desired result.  
\end{proof}

Our final preliminary lemma bounds the size of the solutions for $r$ and $s$.

\begin{lemma}\label{lem:r_s_bounded}
Suppose the assumptions of Theorem \ref{thm:main_asym_consistency} hold. Then, there exist constants $ C_r > 0$ and $C_s > 0$ such that 
\[
\mmp\left(\text{For all solutions to (\ref{eq:app_min_max_qr}), } \|\hat{s}\|_2 \leq  C_s \sqrt{n} \text{ and } \|\hat{r}\|_2 \leq C_r \sqrt{n}\right) \to 1.
\]
\end{lemma}
\begin{proof}
    Fix any primal-dual solution  $(\hat{\beta}_0,\hat{u},\hat{r},\hat{\eta},\hat{s})$ to (\ref{eq:app_min_max_qr}). By the first-order conditions of $(\ref{eq:app_min_max_qr})$ in $\eta$ we must have 
    \[
    \|\hat{r}\|_2 = \| \epsilon - \hat{\beta}_0 \pmb{1}_n - X\hat{u}\|_2 \leq \|\epsilon\|_2 + \sqrt{n}|\hat{\beta}_0| + \lambda_{\text{max}}(X) \|\hat{u}\|_2.
    \]
    By standard results (e.g.~Theorem 3.1 of \cite{Yin1988}) we have that $\lambda_{\text{max}}(X)/\sqrt{n}$ is converging in probability to a constant. Moreover, by the law of large numbers, $\|\epsilon\|_2/\sqrt{n} \stackrel{\mmp}{\to} \sqrt{\mme[\epsilon_i^2]}$. Combining these facts with the bounds on $|\hat{\beta}_0|$ and $\|\hat{u}\|_2$ given by Lemma \ref{lem:primal_bound} gives the desired bound on $\|\hat{r}\|_2$. 

    To bound $\|\hat{s}\|_2$, note that by standard facts regarding the convex conjugate (e.g.~Proposition 11.3 of \citet{Rockefeller1997}), we have $\hat{s} \in n^{-1/2}\partial \mathcal{R}_d(\tilde{\beta} + \hat{u})$. Moreover, by Lemma \ref{lem:primal_bound} and the law of large numbers there exists $C>0$ such that with probability converging to one all primal solutions satisfy $\|\tilde{\beta} + \hat{u}\|_2 \leq C$. So, with probability converging to one,
    \[
    \frac{1}{\sqrt{n}}\|\hat{s}\|_2 \leq   \sup_{\|v\|_2 \leq C}\left\|\frac{1}{n}\partial \mathcal{R}_d(v) \right\|_2.
    \]
    This last quantity is bounded by part 2 of Assumption \ref{assump:reg_high_dim_assumptions}.
\end{proof}

\subsection{Reduction to the auxiliary optimization problem}\label{sec:auxiliary_reduc}

We will now reduce (\ref{eq:app_min_max_qr}) to a simpler asymptotic program that is easier to study. Our main tool will be the Gaussian comparison inequalities of \citet{Gordon1985, Gordon1988} and their application to regression problems developed in \citet{Thram2018}. In particular, we will apply the following proposition. Since this result is a minor extension of Theorem 3 of \citet{Thram2018} (see also Theorem 3 of \citet{Thram2015}) we omit its proof. 
\begin{proposition}[Extension of Theorem 3 of \citet{Thram2018}]\label{prop:gordon}
    Fix any $d,n \in \mmn$. Let $X \in \mmr^{n \times d}$ be distributed as $(X_{i,j})_{i \in [n], j \in [d]}\stackrel{\text{i.i.d.}}{\sim} \mathcal{N}(0,1)$ and define $g \sim \mathcal{N}(0,I_n)$ and $h \sim \mathcal{N}(0,I_d)$ to be independent Gaussian vectors. Let $Q(\beta_0,u,r,\eta,s) : \mathcal{\mmr} \times \mmr^d \times \mmr^n \times \mmr^n \times \mmr^d \to \mmr$ be jointly continuous, convex in $(r,\beta_0,u)$, and concave in $(s,\eta)$. Fix any compact sets $A \subseteq \mathcal{\mmr} \times \mmr^d \times \mmr^n$ and $B \subseteq \mmr^n \times \mmr^d$ and define the values 
    \begin{align*}
    & \Phi  = \max_{(\eta,s) \in B} \min_{(\beta_0,u,r) \in A}  \eta^\top X u + Q(\beta_0,u,r,\eta,s),\\
    & \phi  =  \max_{(\eta,s) \in B} \min_{(\beta_0,u,r) \in A} \|u\|_2\eta^\top g + \|\eta\|_2 u^\top h + Q(\beta_0,u,r,\eta,s).
    \end{align*}
    Then, for all $c \in \mmr$,
    \[
    \mmp(\Phi > c) \leq 2\mmp(\phi \geq c).
    \]
    If in addition $A$ and $B$ are convex, then for all $c \in \mmr$,
    \[
    \mmp(\Phi < c) \leq 2\mmp(\phi \leq c).
    \]
\end{proposition}

To apply this result in our context, let
\begin{align*}
\Phi(S) & =
\max_{(\eta \in S, \|s\|_2 \leq  C_s \sqrt{n})} \min_{(|\beta_0| \leq C_{\beta_0},\|u\|_2 \leq C_{u},\|r\|_2 \leq  C_r \sqrt{n})} \frac{1}{n}\sum_{i=1}^n \ell(r_i) + \frac{1}{n} \eta^\top (\epsilon - \beta_0 \pmb{1}_n - Xu - r)\\
& \hspace{8cm} + \frac{1}{\sqrt{n}} s^\top(\tilde{\beta} + u) -  \frac{1}{n} \mathcal{R}^*_d(\sqrt{n}s),
\end{align*}
where the constants $C_{s},C_{\beta_0},C_u,C_r$ satisfy the conclusions of Lemmas \ref{lem:primal_bound} and \ref{lem:r_s_bounded}. Let $C_{\eta}$ satisfy the conclusion of Lemma \ref{lem:dual_norm_bounds}. We know that asymptotically the solutions of $\Phi(\{\eta : \|\eta\|_2 \leq \sqrt{n}C_{\eta}\})$ agree with those of (\ref{eq:app_min_max_qr}). Our goal will be to compare the value of $\Phi(\{\eta : \|\eta\|_2 \leq \sqrt{n}C_{\eta}\})$ to that of $\Phi(S)$ when $S$ is a more restricted set. The key insight of \citet{Thram2018} is that for this purpose it is sufficient to study the value of the auxiliary optimization,
\begin{align*}
\phi(S) & := \min_{(\|r\|_2 \leq C_r \sqrt{n},|\beta_0| \leq C_{\beta_0}, 0 \leq M_u \leq C_u)} \max_{(\|s\|_2 \leq   C_s \sqrt{n},\eta \in S)} \min_{(u : \|u\|_2 = M_u)} \bigg( \frac{1}{n} \|u\|_2 \eta^\top g + \frac{1}{n} \|\eta\|_2 u^\top h\\
& \hspace{2cm} + \frac{1}{n} \eta^\top \epsilon - \frac{1}{n} \beta_0 \eta^\top\pmb{1}_n  - \frac{1}{n} \eta^\top r + \frac{1}{n} \sum_{i=1}^n \ell_{\tau}(r_i)  + \frac{1}{\sqrt{n}} s^\top(\tilde{\beta} + u) - \frac{1}{n} \mathcal{R}^*_d(\sqrt{n}s) \bigg),
\end{align*}
where $h \sim \mathcal{N}(0,I_d)$ and $g \sim \mathcal{N}(0,I_n)$ are Gaussian vectors sampled such that $(g,h,\epsilon,\tilde{\beta})$ is jointly independent. The following proposition formalizes this.

\begin{proposition}[Extension of Lemma 7 in \citet{Thram2018}]\label{prop:original_to_aux}
    Suppose the assumptions of Theorem \ref{thm:main_asym_consistency} hold and let $C_{s}$, $C_{\beta_0}$, $C_u$, $C_r$, and $C_{\eta}$ be constants satisfying the conclusion of Lemmas \ref{lem:dual_norm_bounds}, \ref{lem:primal_bound}, and \ref{lem:r_s_bounded}. Let $S$ be any set such that 
    \begin{enumerate}
        \item $S$ is compact.
        \item There exists $v \in \mmr$ and $\xi, \delta > 0$ such that 
        \[
        \min\{\mmp(\phi(\{\eta : \|\eta\|_2 \leq \sqrt{n}C_{\eta}\}) \geq v + \delta), \mmp(\phi(S) \leq v - \delta)  \} \geq 1-\xi.
        \]
    \end{enumerate}
    Then,
        \[
        \mmp(\text{For all dual solutions to (\ref{eq:asymp_program}), } \hat{\eta} \notin S) \geq 1-4\xi.
        \]
\end{proposition}
\begin{proof}
    This result follows immediately by applying Proposition \ref{prop:gordon} and repeating the steps of Lemma 7 in \citet{Thram2018}.
\end{proof}

Our goal now is to lower bound $\phi(\{\eta : \|\eta\|_2 \leq \sqrt{n}C_{\eta}\})$ and upper bound $\phi(S)$ for a more restricted set $S$. We will focus initially on $\phi(\{\eta : \|\eta\|_2 \leq \sqrt{n}C_{\eta}\})$. Let $c_{\eta}$ and $C_{\eta}$ be constants satisfying the conclusion of Lemma \ref{lem:dual_norm_bounds}. We have that 
\begin{align*}
& \phi(\{\eta: \|\eta\|_2 \leq \sqrt{n}C_{\eta}\})  \geq \phi(\{\eta : c_{\eta} \sqrt{n} \leq \|\eta\|_2 \leq C_{\eta} \sqrt{n}\})\\
& =  \min_{(\|r\|_2 \leq C_r \sqrt{n},|\beta_0| \leq C_{\beta_0}, 0 \leq M_u \leq C_u)} \max_{(\|s\|_2 \leq  C_s \sqrt{n},c_{\eta}\sqrt{n} \leq \|\eta\|_2 \leq C_{\eta} \sqrt{n})}  \bigg( - M_u\left\|\frac{1}{n} \|\eta\|_2  h + \frac{1}{\sqrt{n}} s \right\|_2\\
& \hspace{1cm} \frac{1}{n} M_u \eta^\top g  + \frac{1}{n} \eta^\top \epsilon - \frac{1}{n} \beta_0 \eta^\top\pmb{1}_n   - \frac{1}{n} \eta^\top r + \frac{1}{n} \sum_{i=1}^n \ell_{\tau}(r_i)  + \frac{1}{\sqrt{n}} s^\top\tilde{\beta}  -\frac{1}{n} \mathcal{R}^*_d(\sqrt{n}s) \bigg)\\
& =  \min_{(\|r\|_2 \leq C_r \sqrt{n},|\beta_0| \leq C_{\beta_0}, 0 \leq M_u \leq C_u)} \max_{(\|s\|_2 \leq  C_s \sqrt{n}, c_{\eta}  \leq M_{\eta} \leq C_{\eta} )}  \bigg( - M_u\left\|\frac{1}{\sqrt{n}} M_{\eta}  h + \frac{1}{\sqrt{n}} s \right\|_2 \\
& \hspace{0.1cm} + M_{\eta} \left\| \frac{1}{\sqrt{n}} M_u  g + \frac{1}{\sqrt{n}} \epsilon - \frac{1}{\sqrt{n}} \beta_0 \pmb{1}_n - \frac{1}{\sqrt{n}} r \right\|_2 + \frac{1}{n} \sum_{i=1}^n \ell_{\tau}(r_i)   + \frac{1}{\sqrt{n}} s^\top\tilde{\beta}  - \frac{1}{n} \mathcal{R}^*_d(\sqrt{n}s)  \bigg).
\end{align*}
Notably, this last optimization problem is convex-concave. Now, note that for any vector $x$ and $C \geq \|x\|_2$, $\|x\|_2 = \min_{0<\tau \leq C} \frac{\|x\|_2^2}{2\tau} + \frac{\tau}{2}$. Moreover, by the weak law of large numbers, there exist constants $C_1, C_2 >0$ such that with probability converging to one,
\begin{align*}
 & \max_{(\|r\|_2 \leq C_r\sqrt{n}, |\beta_0| \leq C_{\beta_0}, 0 \leq M_u \leq C_u)} \left\|\frac{1}{\sqrt{n}} M_u  g + \frac{1}{\sqrt{n}} \epsilon - \frac{1}{\sqrt{n}} \beta_0 \pmb{1}_n - \frac{1}{\sqrt{n}} r\right\|_2 \\
 & \leq C_u \frac{\|g\|_2}{\sqrt{n}} + \frac{\|\epsilon\|_2}{\sqrt{n}} + C_{\beta_0} + C_r \leq C_1,
\end{align*}
and
\[
\max_{(\|s\|_2 \leq C_s\sqrt{n}, c_{\eta} \leq  M_{\eta} \leq C_{\eta} )} \left\|\frac{1}{\sqrt{n}} M_{\eta}  h + \frac{1}{\sqrt{n}} s \right\|_2 \leq C_{\eta} \frac{\|h_2\|_2}{\sqrt{n}} + C_s \leq C_2.
\]
So, applying these facts and using Sion's minimax theorem to swap the order of minimization and maximization \citep{Sion1958}, we have that the above can be rewritten as
\begin{align*}
    & \min_{(|\beta_0| \leq C_{\beta_0}, 0 \leq M_u \leq C_u)} \max_{( c_{\eta} \leq M_{\eta} \leq C_{\eta})} \min_{0 < \rho_1 \leq C_1}  \max_{0 < \rho_2 \leq C_2} \min_{\|r\|_2 \leq C_r \sqrt{n} } \max_{\|s\|_2 \leq C_s \sqrt{n}} \bigg( \frac{M_{\eta}}{2 n \rho_1} \left\|  M_u  g +  \epsilon -  \beta_0 \pmb{1}_n -  r \right\|^2_2\\
    & \hspace{2cm} + \frac{M_{\eta}\rho_1}{2} - \frac{M_u}{2n\rho_2}\left\|M_{\eta}  h + s \right\|^2_2 - \frac{M_u\rho_2}{2} + \frac{1}{n} \sum_{i=1}^n \ell_{\tau}(r_i)  + \frac{1}{\sqrt{n}} s^\top\tilde{\beta}  - \frac{1}{n} \mathcal{R}^*_d(\sqrt{n}s)\bigg).
\end{align*}
To simplify this further, we will rewrite the optimizations over $r$ and $s$ in terms of the Moreau envelope. This is done using the following lemma.
\begin{lemma}\label{lem:swap_to_envelope}
    Fix any constants $C_{\beta_0}, C_u, C_{\eta}, c_{\eta} > 0$ with $C_{\eta} > c_{\eta}$. Under the assumptions of Theorem \ref{thm:main_asym_consistency}, there exist constants $\tilde{C}_s, \tilde{C}_r > 0$, such that with probability tending to 1, it holds that for any $|\beta_0| \leq C_{\beta_0}$, $0 \leq M_u \leq C_u$, $c_{\eta} \leq M_{\eta} \leq C_{\eta}$, and $\rho_2, \rho_1 > 0$,
    \[
    \min_{\|r\|_2 \leq \tilde{C}_r \sqrt{n}} \frac{M_{\eta}}{2 n \rho_1} \left\|  M_u  g +  \epsilon -  \beta_0 \pmb{1}_n -  r \right\|^2_2 + \frac{1}{n} \sum_{i=1}^n \ell_{\tau}(r_i) = \frac{1}{n} \sum_{i=1}^n e_{\ell_{\tau}}\left( M_u  g_i +  \epsilon_i - \beta_0 ; \frac{\rho_1}{M_{\eta}} \right),
    \]
    and
    \begin{align*}
    & \max_{\|s\|_2 \leq \tilde{C}_s \sqrt{n}}   - \frac{M_u}{2n\rho_2}\left\|M_{\eta}  h + s \right\|^2_2 + \frac{1}{\sqrt{n}}s^\top \tilde{\beta} - \frac{1}{n} \mathcal{R}^*_d(\sqrt{n}s)\\
    & = -\frac{M_{\eta}^2 M_u}{2n\rho_2} \|h\|_2^2 + \frac{1}{n} e_{\mathcal{R}_d(\cdot/\sqrt{n})}\left( \sqrt{n}\tilde{\beta} - \frac{M_{\eta}M_u}{\rho_2} h  ; \frac{M_{u}}{\rho_2} \right).
    \end{align*}
\end{lemma}
\begin{proof}
    Recall the definition of the proximal function,
    \[
    \text{prox}_f(x;\rho) = \underset{v}{\text{argmin}} \frac{1}{2\rho} \|x-v\|_2^2 +  f(v).
    \]
    For any vector $x \in \mmr^n$ and $\rho > 0$, let 
    \[
    \text{prox}^n_{\ell_{\tau}}(x;\tau) = \underset{v}{\text{argmin}} \frac{1}{2\rho} \|x-v\|_2^2 +  \sum_{i=1}^n \ell_{\tau}(v_i),
    \]
    denote the proximal map of the function $v \mapsto \sum_{i=1}^n \ell_{\tau}(v_i)$.
    By definition of $\ell_{\tau}$, we have $\text{prox}^n_{\ell_{\tau}}(0;\rho) = 0$ for any $\rho > 0$. Since the proximal function is non-expansive (Proposition 12.28 of \citet{Bauschke2017}), it holds that for any $x \in \mmr$ and $\rho > 0$,
    \[
    \|\text{prox}^n_{\ell_{\tau}}(x;\rho)\|_2 =   \|\text{prox}^n_{\ell_{\tau}}(x;\rho) - \text{prox}^n_{\ell_{\rho}}(0;\rho) \|_2 \leq \|x - 0 \|_2 = \|x\|_2.
    \]
    So, in particular, 
    \[
    \|\text{prox}^n_{\ell_{\tau}}(  M_u  g +  \epsilon -  \beta_0 \pmb{1}_n  ;\rho_1/M_{\eta}) \|_2 \leq \|M_u  g +  \epsilon -  \beta_0 \pmb{1}_n  \|_2 \leq C_u \|g\|_2 + \|\epsilon\|_2 + C_{\beta_0}\sqrt{n}.
    \]
    The first part of the lemma then follows immediately by using the law of large numbers to bound $\|g\|_2$ and $\|\epsilon\|_2$.

    For the second part of the lemma, write 
    \begin{align*}
     & \max_{\|s\| \leq \tilde{C}_s\sqrt{n}} - \frac{M_u}{2n\rho_2}\left\|M_{\eta}  h + s \right\|^2_2 + \frac{1}{\sqrt{n}}s^\top \tilde{\beta} - \frac{1}{n} \mathcal{R}^*_d(\sqrt{n}s)\\
     & =   \max_{\|s\| \leq \tilde{C}_s\sqrt{n}} - \frac{M_u}{2n\rho_2}\left\|M_{\eta}  h - \frac{\rho_2}{M_u}\sqrt{n}\tilde{\beta} + s \right\|^2_2 + \frac{\rho_2}{2nM_u}\|\sqrt{n} \tilde{\beta}\|_2^2 - \frac{1}{n} M_{\eta} h^\top(\sqrt{n}\tilde{\beta})  - \frac{1}{n} \mathcal{R}^*_d(\sqrt{n}s).
    \end{align*}
    Suppose that $\tilde{C}_s$ is sufficient large so that $\left\|\text{prox}_{x \mapsto \mathcal{R}^*_d(\sqrt{n}x)}\left(\frac{\rho_2}{M_u}\sqrt{n}\tilde{\beta} - M_{\eta}  h  ; \frac{\rho_2}{M_u} \right)\right\|_2  \leq \tilde{C}_s \sqrt{n}$. Then, the above can be rewritten as
    \[
- \frac{1}{n}\text{e}_{x \mapsto \mathcal{R}^*_d(\sqrt{n}x)}\left(\frac{\rho_2}{M_u}\sqrt{n}\tilde{\beta} - M_{\eta}  h  ; \frac{\rho_2}{M_u} \right) +  \frac{\rho_2}{2nM_u}\|\sqrt{n} \tilde{\beta}\|_2^2 - \frac{1}{n} M_{\eta} h^\top(\sqrt{n}\tilde{\beta}).
    \]
    Moreover, recalling the identity (Lemma \ref{lem:envelope_identity} below),
    \[
    e_f(x;\rho) + e_{f^*}(x/\rho;1/\rho) = \frac{\|x\|_2^2}{2\rho},
    \]
    this can be equivalently written as
    \[
    \frac{1}{n} e_{\mathcal{R}_d(\cdot/\sqrt{n})}\left(\sqrt{n}\tilde{\beta} -  \frac{M_uM_{\eta}}{\rho_2}h ; \frac{M_u}{\rho_2} \right) - \frac{M_{\eta}^2 M_u}{2n \rho_2} \|h\|_2^2,
    \]
    as desired. 
    
    It remains to bound $\left\|\text{prox}_{x \mapsto \mathcal{R}^*_d(\sqrt{n}x)}\left(\frac{\rho_2}{M_u}\sqrt{n}\tilde{\beta} - M_{\eta}  h  ; \frac{\rho_2}{M_u} \right)\right\|_2 $. For ease of notation, let $s^* = \text{prox}_{x \mapsto \mathcal{R}^*_d(\sqrt{n}x)}\left(\frac{\rho_2}{M_u}\sqrt{n}\tilde{\beta} - M_{\eta}  h  ; \frac{\rho_2}{M_u} \right)$. Fix any $s' \in n^{1/2}\partial R_d(\tilde{\beta})$. By definition of the proximal function, we have that 
    \begin{align*}
     & \frac{M_u}{2\rho_2}\left\|M_{\eta}  h - \frac{\rho_2}{M_u}\sqrt{n}\tilde{\beta} + s^* \right\|^2_2 +  \mathcal{R}_d^*(\sqrt{n}s^*) \leq  \frac{M_u}{2\rho_2}\left\|M_{\eta}  h - \frac{\rho_2}{M_u}\sqrt{n}\tilde{\beta} + s' \right\|^2_2 + \mathcal{R}_d^*(\sqrt{n}s')\\
     & \implies  \frac{M_u}{2\rho_2}\left\|M_{\eta}  h  + s^* \right\|^2_2 \leq \frac{M_u}{2\rho_2}\left\|M_{\eta}  h  + s' \right\|^2_2 + (s^*)^\top(\sqrt{n}\tilde{\beta}) - \mathcal{R}_d^*(\sqrt{n}s^*)\\
     & \hspace{4cm} - ((s')^\top(\sqrt{n}\tilde{\beta}) - \mathcal{R}_d^*(\sqrt{n}s'))\\
     & \implies    \frac{M_u}{2\rho_2}\left\|M_{\eta}  h  + s^* \right\|^2_2 \leq \frac{M_u}{2\rho_2}\left\|M_{\eta}  h  + s' \right\|^2_2,
    \end{align*}
    where on the last line we have applied the definition of the convex conjugate (see also Proposition 11.3 of \citet{Rockefeller1997}). So, rearranging we have that 
    \[
   \frac{\|s^*\|_2}{\sqrt{n}} \leq 2M_{\eta}\frac{\|h\|_2}{\sqrt{n}} + \|\partial R_d(\tilde{\beta})\|_2 \leq 2C_{\eta} \frac{\|h\|_2}{\sqrt{n}} + \|\partial R_d(\tilde{\beta})\|_2.
    \]
    This last quantity can be bounded by the law of large numbers and part 2 of Assumption \ref{assump:reg_high_dim_assumptions}
\end{proof}

Now, without loss of generality we may assume that $C_r \geq \tilde{C}_r$ and $C_s \geq \tilde{C}_s$. So, applying Lemma \ref{lem:swap_to_envelope} and taking a continuous extension at $\rho_1 = 0$, our previous calculations gives us that 
\begin{equation}\label{eq:final_data_aux}
\begin{split}
&  \phi(\{\eta: \|\eta\|_2 \leq \sqrt{n}C_{\eta}\})\\
    & \geq  \min_{(|\beta_0| \leq C_{\beta_0}, 0 \leq M_u \leq C_u, 0 \leq \rho_1 \leq C_1)} \max_{(c_{\eta} \leq M_{\eta} \leq C_{\eta}, 0 < \rho_2 \leq C_2)}  \bigg( \frac{1}{n} \sum_{i=1}^n e_{\ell_{\tau}}\left( M_u  g_i +  \epsilon_i - \beta_0 ; \frac{\rho_1}{M_{\eta}} \right) \\ 
    & \hspace{2cm} - \frac{M_{\eta}^2 M_u}{2n\rho_2} \|h\|_2^2 + \frac{1}{n} e_{\mathcal{R}_d(\cdot/\sqrt{n})}\left( \sqrt{n}\tilde{\beta}-\frac{M_{\eta}M_u}{\rho_2} h  ; \frac{M_{u}}{\rho_2} \right) + \frac{M_{\eta}\rho_1}{2}  - \frac{M_u\rho_2}{2}    \bigg),
\end{split}
\end{equation}

Our final step is to replace all the random quantities above with their asymptotic limits. To do this, we will employ the following lemma which states that pointwise convergence can be converted to convergence of the minimum value of a convex function. This result is a minor variant of Lemma 10 of \citet{Thram2018} and we include a partial proof for completeness.

\begin{lemma}[Extension of Lemma 10 of \citet{Thram2018}]\label{lem:convexity_lemma}
    Fix $b>a$ and let $f_n : [a,b] \to \mmr$ be a sequence of random convex functions converging pointwise in probability to $f : [a,b] \to \mmr$. Then, 
    \[
    \inf_{x \in [a,b]} f_n(x) \stackrel{\mmp}{\to}  \inf_{x \in [a,b]} f(x) .
    \]
    Similarly, if $f_n : (a,b] \to \mmr$ is a sequence of random convex functions converging pointwise in probability to  $f : (a,b] \to \mmr$, then 
    \[
    \inf_{x \in (a,b]} f_n(x) \stackrel{\mmp}{\to}  \inf_{x \in (a,b]} f(x) .
    \]
\end{lemma}
\begin{proof}
    We will prove the first part of the lemma. Proof of the second part is similar and is omitted. For any $x' \in [a,b]$ we have that 
    \[
    \limsup_{n \to \infty}     \inf_{x \in [a,b]}f_n(x) \leq \limsup_{n \to \infty}    f_n(x') \stackrel{\mmp}{=} f(x').
    \]
    So, taking an infimum over $x'$ gives $\limsup_{n \to \infty}     \inf_{x \in [a,b]}f_n(x) \leq \inf_{x \in [a,b]} f(x)$. 
    
    It suffices to prove a matching lower bound. If $\inf_{x \in [a,b]} f(x) = -\infty$ there is nothing to show. So, assume that $\inf_{x \in [a,b]} f(x)  > -\infty$. By Lemma 7.75 of \citet{liese2008} we have that for any points $a < x_1 < x_2 < b$,
    \[
    \sup_{x \in [x_1,x_2]}|f_n(x) - f(x)| \stackrel{\mmp}{\to} 0,
    \]
    and thus also,
    \begin{equation}\label{eq:uniform_conv_convex}
     \inf_{x \in [x_1,x_2]}f_n(x)  \stackrel{\mmp}{\to} \inf_{x \in [x_1,x_2]}f(x).
    \end{equation}
    So, we just need to check what happens on the boundary. We will focus on the lower boundary. First, suppose that $\liminf_{x \to a} f(x) > \inf_{x \in [a,b]} f(x)$. Let $x^* \in (a,b]$ be such that $f(x^*) < \inf_{x \in [a,b]} f(x) + \frac{\liminf_{x \to a} f(x) - \inf_{x \in [a,b]} f(x)}{2}$. Fix any $a < x_1 < x_2 < x^*$ with $f(x_1), f(x_2) > f(x^*)$. For any $x \in [a,x_1]$ let $\lambda_x>0$ be such that $x_2 = \lambda_x x + (1-\lambda_x)x^*$. Then, 
    \[
    f_n(x_2) \leq \lambda_x f_n(x) +(1-\lambda_x)f_n(x^*) \implies f_n(x) \geq f_n(x^*) + \frac{1}{\lambda_x} (f_n(x_2) - f_n(x^*)). 
    \]
    Asymptotically, we have that $\lim_{n \to \infty} f_n(x_2) - f_n(x^*) \stackrel{\mmp}{=} f(x_2) - f(x^*) > 0$ and thus $\liminf_{n \to \infty} \inf_{x \in [a,x_1]} f_n(x) \stackrel{\mmp}{\geq} f(x^*) \geq \inf_{x \in [a,b]} f(x)$.
    
    On the other hand, suppose that $\liminf_{x \to a} f(x) =  \inf_{x \in [a,b]} f(x)$. Fix any $\delta > 0$. We claim that there exists $a < x_1 < x_2 < b$ such that $f(x_1), f(x_2) < \inf_{x \in [a,b]} f(x) + \delta$. To see this, let $x_{\delta} \in [a,(b+a)/2]$ be such that $f(x_{\delta}) < \inf_{x \in [a,b]} f(x) + \delta/2$. Fix any $0 < \xi < b - x_{\delta}$ and for any $x \in [x_{\delta},x_{\delta} + \xi]$ write 
    \[
    f(x) \leq \left(1- \frac{x-x_{\delta}}{b-x_{\delta}} \right) f(x_{\delta}) + \frac{x-x_{\delta}}{b-x_{\delta}} f(b) \leq f(x_{\delta}) + \frac{\xi}{b-x_{\delta}} (f(b) - \inf_{x' \in [a,b]} f(x')).
    \]
    Taking $\xi$ sufficiently small we find that $\sup_{x \in [x_{\delta},x_{\delta} + \xi]} f(x) \leq \inf_{x \in [a,b]} f(x) + \delta$ and so setting $x_1 < x_2$ to be any points in $(x_{\delta}, x_{\delta} + \xi)$ gives the desired claim. 
    
    Now, for any $x \in [a,x_1]$ we have
    \begin{align*}
     f_n(x_1) & = f_n\left(\frac{x_2 - x_1}{x_2 - x} x + \left(1-\frac{x_2 - x_1}{x_2 - x}\right) x_2\right ) \leq \frac{x_2 - x_1}{x_2 - x} f_n(x) + \left( 1- \frac{x_2 - x_1}{x_2 - x} \right) f_n(x_2)\\
     \implies  f_n(x) & \geq \frac{x_2 - x}{x_2 - x_1} f_n(x_1) - \frac{x_2 - x}{x_2 - x_1} \left( 1- \frac{x_2 - x_1}{x_2 - x} \right)   f_n(x_2)\\
     & \geq \min\{f_n(x_1), f_n(x_2)\} - \frac{x_2 - x}{x_2 - x_1} \left( 1- \frac{x_2 - x_1}{x_2 - x} \right)  (f_n(x_2) -  \min\{f_n(x_1), f_n(x_2)\})\\
    & \geq \min\{f_n(x_1), f_n(x_2)\} -  \frac{x_1 - a}{x_2 - x_1}(f_n(x_2) -  \min\{f_n(x_1), f_n(x_2)\})\\
    & \stackrel{\mmp}{\geq} \inf_{x' \in [a,b]} f(x') - \frac{x_1 - a}{x_2 - x_1} \delta,
    \end{align*}
    where the probability in the last inequality holds uniformly over $x$. Thus,
    \[
    \liminf_{n \to \infty} \inf_{x \in [a,x_1]} f_n(x) \stackrel{\mmp}{\geq}  \inf_{x \in [a,b]} f(x). 
    \]

    So, in total, we find that in all cases we may find $x_1 \in (a,b]$ such that 
    \[
     \liminf_{n \to \infty} \inf_{x \in [a,x_1]} f_n(x) \stackrel{\mmp}{\geq}  \inf_{x \in [a,b]} f(x).
    \]
    By a matching argument, we may also find $x'_1 \in [a,b)$ such that 
    \[
    \liminf_{n \to \infty} \inf_{x \in [x'_1,b]} f_n(x) \stackrel{\mmp}{\geq}  \inf_{x \in [a,b]} f(x).
    \]
    Combining these two facts and using (\ref{eq:uniform_conv_convex}) to get convergence on the interior gives the desired result.
\end{proof}

Combining all of the previous results we arrive at the following.
\begin{proposition}\label{prop:aux_to_asymp}
    Suppose the assumptions of Theorem \ref{thm:main_asym_consistency} hold. Let $(C_{\beta_0}, C_u, c_{\eta}, C_{\eta}, C_1, C_2)$ be constants satisfying the conclusions of Lemmas \ref{lem:dual_norm_bounds}, \ref{lem:primal_bound}, \ref{lem:r_s_bounded}, and \ref{lem:swap_to_envelope}.
    Let $V$ denote the value of the asymptotic program defined in (\ref{eq:asymp_program}). Then,
    \[
    \liminf_{n,d \to \infty} \phi(\{\eta: \|\eta\|_2 \leq \sqrt{n}C_{\eta}\}) \stackrel{\mmp}{\geq} V.
    \]
\end{proposition}
\begin{proof}
    This lemma follows by repeated applications of Lemma \ref{lem:convexity_lemma} to (\ref{eq:final_data_aux}) where the corresponding pointwise limits follow by the law of large numbers and part 3 of Assumption \ref{assump:reg_high_dim_assumptions}.
\end{proof}

\subsection{Analysis of the asymptotic program}\label{sec:asymp_program}

In this section we prove a number of useful results regarding the asymptotic auxiliary program defined in (\ref{eq:final_data_aux}). In what follows we use 
\begin{align*}
A(\beta_0,M_u,\rho_1,M_{\eta},\rho_2) & = \mme\left[e_{\ell_{\tau}}\left(M_ug_1 + \epsilon_1 - \beta_0 ; \frac{\rho_1}{M_{\eta}} \right)\right] - \frac{M_{\eta}^2M_u\gamma}{2\rho_2}\\
& + \gamma \mme\left[ e_{\nu}\left( \frac{M_{\eta}M_u}{\rho_2} h_1 + \gamma\sqrt{d}\tilde{\beta}_1 ; \frac{M_u}{\rho_2} \right)\right] + \frac{M_{\eta}\rho_1}{2} - \frac{M_u\rho_2}{2}, 
\end{align*}
to denote the objective of this optimization. 

\begin{lemma}\label{lem:A_convex_cont}
    Under the assumptions of Theorem \ref{thm:main_asym_consistency}, $A(\beta_0,M_u,\rho_1,M_{\eta},\rho_2)$ is jointly continuous, jointly convex in $(\beta_0,M_u,\rho_1)$, and jointly concave in $(M_{\eta},\rho_2)$ on the domain $\mmr \times \mmr_{\geq 0} \times \mmr_{\geq 0} \times \mmr_{> 0} \times \mmr_{>0}$.
\end{lemma}
\begin{proof}
    The second, fourth, and fifth terms of $A$ are clearly jointly continuous. The third term of $A$ is jointly continuous by part three of Assumption \ref{assump:reg_high_dim_assumptions}. Joint continuity of the first term follows by inspecting the form of $e_{\ell_{\tau}}$ (Lemma \ref{lem:pinball_env}) and applying the dominated convergence theorem. The fact that $A$ is convex-concave follows directly from the fact that it is the pointwise limit of a sequence of convex-concave functions.
\end{proof}

\begin{lemma}\label{lem:m_rho_b_strict_convex}
    Fix any $C_{\eta} > c_{\eta} >0$ and $C_2 >0$. Under the conditions of Theorem \ref{thm:main_asym_consistency}, the function
    \[
    (\beta_0,M_u,\rho_1) \mapsto \max_{c_{\eta} \leq M_{\eta} \leq C_{\eta}, 0 < \rho_2 \leq C_2} A(\beta_0,M_u,\rho_1,M_{\eta},\rho_2),
    \]
    is jointly strictly convex on $\mmr \times \mmr_{\geq 0} \times \mmr_{>0}$. Moreover, for $\rho_1 = 0$ this function is jointly strictly convex in $(\beta_0,M_u)$.
\end{lemma}
\begin{proof}
    We first consider the case where $\rho_1 > 0$. Fix any $M_{\eta} \in [c_{\eta}, C_{\eta}]$ and pair of distinct points $ (\beta_0,M_u,\rho_1),  (\beta'_0,M'_u,\rho'_1) \in \mmr \times \mmr_{\geq 0} \times \mmr_{>0}$. For $\theta \in [0,1]$ define the function
    \[
    w(\theta) = \mme\left[e_{\ell_{\tau}}\left( ((1-\theta)M_u + \theta M_u')g_1 + \epsilon_1 - (1-\theta)\beta_0 - \theta \beta_0' ; \frac{(1-\theta)\rho_1 + \theta \rho_1'}{M_{\eta}} \right)\right].
    \]
    For ease of notation, let 
    \begin{align*}
    & (\xi_1,\xi_2,\xi_3) = (M_u' - M_u, \beta'_0 - \beta_0,\rho'_1 - \rho_1),\ \rho_{\theta} = ((1-\theta)\rho_1 + \theta \rho_1'),\\
     \text{and } & Z_{\theta} = ((1-\theta)M_u + \theta M_u')g_1 + \epsilon_1 - (1-\theta)\beta_0 - \theta \beta_0'.
    \end{align*}
    By the dominated convergence theorem and a direct calculation using the form of $e_{\ell_{\tau}}$ (see Lemma \ref{lem:pinball_env}), we have 
    \begin{align*}
    w'(\theta) & = \mme\bigg[(g_1\xi_1-\xi_2)\tau \bone\left\{Z_{\theta} > \tau \frac{\rho_{\theta}}{M_{\eta}}\right\} + (g_1\xi_1-\xi_2)\frac{Z_{\theta}M_{\eta}}{\rho_{\theta}}  \bone\left\{-(1-\tau)\frac{\rho_{\theta}}{M_{\eta}} \leq Z_{\theta} \leq \tau \frac{\rho_{\theta}}{M_{\eta}}\right\}\\
    & \hspace{0.25cm} - (g_1\xi_1-\xi_2)(1-\tau) \bone\left\{Z_{\theta} < -(1-\tau)\frac{\rho_{\theta}}{M_{\eta}}\right\}     -\frac{\tau^2\xi_3}{2M_{\eta}} \bone\left\{Z_{\theta} > \tau \frac{\rho_{\theta}}{M_{\eta}}\right\}\\
    & \hspace{0.25cm} - \frac{Z_{\theta}^2M_{\eta}\xi_3}{2\rho_{\theta}^2} \bone\left\{-(1-\tau)\frac{\rho_{\theta}}{M_{\eta}} \leq Z_{\theta} \leq \tau \frac{\rho_{\theta}}{M_{\eta}}\right\} - \frac{(1-\tau)^2\xi_3}{2M_{\eta}} \bone\left\{Z_{\theta} < -(1-\tau)\frac{\rho_{\theta}}{M_{\eta}}\right\} \bigg],
    \end{align*}
    and 
    \begin{align*}
       & w''(\theta)\\
       & = \mme\left[ \left( (g\xi_1-\xi_2)^2\frac{M_{\eta}}{\rho_{\theta}} - 2(g_1\xi_1-\xi_2)\xi_3\frac{Z_{\theta}M_{\eta}}{\rho_{\theta}^2} + \xi_3^2\frac{Z_{\theta}^2M_{\eta}}{\rho_{\theta}^3} \right)\bone\left\{-(1-\tau)\frac{\rho_{\theta}}{M_{\eta}} \leq Z_{\theta} \leq \tau \frac{\rho_{\theta}}{M_{\eta}}\right\}   \right]\\
       & = \frac{M_{\eta}}{\rho_{\theta}} \mme\left[ \left( g_1\xi_1-\xi_2 - \xi_3\frac{Z_{\theta}}{\rho_{\theta}} \right)^2\bone\left\{-(1-\tau)\frac{\rho_{\theta}}{M_{\eta}} \leq Z_{\theta} \leq \tau \frac{\rho_{\theta}}{M_{\eta}}\right\}   \right].
    \end{align*}
    Recall that $\epsilon_1$ has positive support on $\mmr$. Thus, $Z_{\theta}$ has positive support on $\mmr$ and $g_1\xi_1-\xi_2 - \xi_3\frac{Z_{\theta}}{\rho_{\theta}}$ has positive support on $\mmr$ if $\xi_3 \neq 0$. Moreover, if $\xi_3 = 0$, then $(\xi_1,\xi_2) \neq (0,0)$ and we clearly have that  $\mmp(g_1\xi_1-\xi_2 = 0) = 0$ . In either case, we conclude that $w''(\theta) > 0$ and thus that
    \[
    (\beta_0, M_u, \rho_1) \mapsto \mme\left[e_{\ell_{\tau}}\left( M_u g_1 + \epsilon_1 - \beta_0; \frac{\rho_1}{M_{\eta}} \right)\right],
    \]
    is strictly convex. Since this term does not involve $\rho_2$ and the remainder of the objective is convex (it is the pointwise limit of a convex function), we conclude that the function 
    \[
 (\beta_0,M_u,\rho_1) \mapsto \max_{0 < \rho_2 \leq C_2} A(\beta_0,M_u,\rho_1,M_{\eta},\rho_2),
   \]
   is strictly convex. Finally, since $A$ is convex-concave we have that for any $(\beta_0,M_u,\rho_1) \in \mmr \times \mmr_{\geq 0} \times \mmr_{>0}$, $M_{\eta} \mapsto \max_{0 < \rho_2 \leq C_2} A(\beta_0,M_u,\rho_1,M_{\eta},\rho_2)$ is concave on $\mmr_{>0}$ and thus continuous on $[c_{\eta}, C_{\eta}]$. The desired result then follows by Lemma \ref{lem:convex_under_opt}.

    Now, consider the case $\rho_1 = 0$. Once again, fix $M_{\eta} \in [c_{\eta}, C_{\eta}]$ and a pair of distinct points $(M_u,\beta_0), (M'_u,\beta'_0) \in \mmr_{\geq 0} \times \mmr$. For $\theta \in [0,1]$ consider the function
    \[
    \tilde{w}(\theta) = \mme[\ell_{\tau}(((1-\theta) M_u + \theta M_u') g_1 + \epsilon_1 - (1-\theta)\beta_0 - \theta\beta_0') ].
    \]
    Let $Z_{\theta} := ((1-\theta) M_u + \theta M_u') g_1 + \epsilon_1 - (1-\theta)\beta_0 - \theta\beta_0'$. By a direct calculation,  
    \begin{align*}
         \tilde{w}'(\theta) & = \mme[\tau((M'_u - M_u) g_1 - (\beta_0' - \beta_0)) \bone\{Z_{\theta} > 0\}\\
         & \hspace{4cm} - (1-\tau)((M'_u - M_u) g_1 - (\beta_0' - \beta_0)) \bone\{Z_{\theta} \leq 0\}  ]\\
        & = \mme\left[((M_u - M_u') g_1 - (\beta_0 - \beta_0' )) \bone\{Z_{\theta} \leq 0\} + \tau ((M'_u - M_u) g_1 - (\beta_0' - \beta_0))  \right],
    \end{align*}
    and 
   \begin{align*}
    & \tilde{w}''(\theta)\\
    & = \frac{d}{d\theta} \mme\left[ \mme\left[((M_u - M'_u) g_1 - (\beta_0 - \beta'_0)) \bone\left\{ \epsilon_1 \leq (1-\theta) (\beta_0 - M_u g_1) + \theta (\beta_0' - M_u' g_1)  \right\} \right]  \mid g_1 \right]\\
    & = \mme\left[ ((M_u - M'_u) g_1 - (\beta_0 - \beta'_0))^2 p_{\epsilon}(  (1-\theta) (\beta_0 - M_u g_1) + \theta (\beta_0' - M_u' g_1) ) \right]\\
    & > 0,
    \end{align*}
    where $p_{\epsilon}$ denotes the density of $\epsilon_1$. Since this last term is positive we find that $\tilde{w}$ is strictly convex. The desired result then follows by arguing as above. 
\end{proof}

\begin{lemma}\label{lem:unique_pos_solution}
   Suppose the assumptions of Theorem \ref{thm:main_asym_consistency} hold. Fix any constants $C_{\beta_0},C_u,C_1,C_2, C_{\eta}, c_{\eta} > 0$ with $C_{\eta} > c_{\eta}$ and $c_{\eta} < \sqrt{(1/2)\min\{\tau^2,(1-\tau)^2\}}$. Then, the asymptotic optimization program (\ref{eq:asymp_program}) admits a unique solution for $(\beta_0,M_u,\rho_1)$. Moreover, letting $(\beta^*_0,M^*_u,\rho^*_1)$ denote this solution we have that $M_u^* > 0 \implies \rho_1^* > 0$. 
\end{lemma}
\begin{proof}
    Since the optimization domain for $(\beta_0,M_u,\rho_1)$ is compact and $A$ is jointly continuous and convex-concave (Lemma \ref{lem:A_convex_cont}) the optimization program (\ref{eq:asymp_program}) must obtain its minimum in $(\beta_0,M_u,\rho_1)$ (cf. Theorem 1.9 and Proposition 1.26 of \citet{Rockefeller1997}). The fact that this minimizer is unique then follows directly from Lemma \ref{lem:m_rho_b_strict_convex}.
    
    Now, let $(\beta^*_0,M^*_u,\rho^*_1)$ denote this unique solution and suppose that $M^*_u > 0$. Recall the identity (Lemma \ref{lem:envelope_identity} below),
    \[
    e_f(x;\rho) + e_{f^*}(x/\rho; 1/\rho) = \frac{x^2}{2\rho}.
    \]
    Applying this to our optimization problem, we have that for any $\rho_2 > 0$ and $0 \leq \rho_1 \leq C_1$,
    \begin{align*}
    A(\beta^*_0,M^*_u,\rho_1,M_{\eta},\rho_2) & = \mme\left[e_{\ell_{\tau}}\left(M^*_ug_1 + \epsilon_1 - \beta^*_0 ; \frac{\rho_1}{M_{\eta}} \right)\right] - \frac{M_{\eta}^2M^*_u\gamma}{2\rho_2}\\
    &  \ \ \  -\gamma \mme\left[ e_{\nu^*}\left( M_{\eta} h_1 + \gamma\frac{\rho_2}{M^*_u} \sqrt{d}\tilde{\beta}_1 ; \frac{\rho_2}{M^*_u} \right)\right]\\
    & + \gamma \frac{\rho_2}{2M^*_u} \mme\left[\left(\frac{M_{\eta}M^*_u}{\rho_2} h_1 + \sqrt{d}\tilde{\beta}_1 \right)^2\right] \\
    & \ \ \ + \frac{M_{\eta}\rho_1}{2} - \frac{M^*_u\rho_2}{2}\\
    & = \mme\left[e_{\ell_{\tau}}\left(M^*_ug_1 + \epsilon_1 - \beta^*_0 ; \frac{\rho_1}{M_{\eta}} \right)\right]\\
    & \ \ \ -\gamma \mme\left[ e_{\nu^*}\left( M_{\eta} h_1 + \gamma\frac{\rho_2}{M^*_u} \sqrt{d}\tilde{\beta}_1 ; \frac{\rho_2}{M^*_u} \right)\right]\\
    & \ \ \ + \gamma \frac{\rho_2}{2M^*_u}\mme[(\sqrt{d}\tilde{\beta}_1)^2]  + \frac{M_{\eta}\rho_1}{2} - \frac{M^*_u\rho_2}{2}.
    \end{align*}
    So, in particular, 
    \begin{equation}\label{eq:lower_A_rho_zero_bound}
    \begin{split}
        & \max_{(0 < \rho_2 \leq C_2, c_{\eta} \leq M_{\eta} \leq C_{\eta})}  A(\beta_0^*,M_u^*,0,M_{\eta},\rho_2)   \geq  \max_{(0 < \rho_2 \leq C_2)} A(\beta_0^*,M_u^*,0,c_{\eta},\rho_2)\\
         & = \mme\left[\ell_{\tau}\left(M^*_ug_1 + \epsilon_1 - \beta^*_0 \right)\right]  - \gamma \mme\left[ e_{\nu^*}\left(c_{\eta} h_1 +  \gamma\frac{\rho_2}{M_u} \sqrt{d}\tilde{\beta}_1 ; \frac{\rho_2}{M^*_u} \right)\right]\\
         & \hspace{1cm} + \gamma \frac{\rho_2}{2M^*_u}\mme[(\sqrt{d}\tilde{\beta}_1)^2]  - \frac{M^*_u\rho_2}{2}.
    \end{split}
    \end{equation}
    We will now compare this lower bound against a matching upper bound when $\rho_1 >0$ is small and positive. Fix $\rho_1, \rho_2 > 0$. By directly examining the definition of $e_{\ell_{\tau}}$, (Lemma \ref{lem:pinball_env}) we have the pointwise inequality 
    \begin{align*}
    & e_{\ell_{\tau}}\left(M^*_ug_1 + \epsilon_1 - \beta^*_0 ; \frac{\rho_1}{M_{\eta}} \right)  \\
    & \leq \ell_{\tau}(M^*_ug_1 + \epsilon_1 - \beta^*_0 ) - \min\{\tau^2,(1-\tau)^2\} \frac{\rho_1}{2M_{\eta}} \bone\left\{M^*_ug_1 + \epsilon_1 - \beta^*_0  \notin \left[ -\frac{\rho_1}{M_{\eta}}(1-\tau), \frac{\rho_1}{M_{\eta}}\tau \right] \right\},
    \end{align*}
    and thus for any $M_{\eta} \in [c_{\eta},C_{\eta}]$,
    \begin{align*}
 & \mme\left[e_{\ell_{\tau}}\left(M^*_ug_1 + \epsilon_1 - \beta^*_0 ; \frac{\rho_1}{M_{\eta}} \right)\right]\\
 & \leq \mme\left[\ell_{\tau}\left(M^*_ug_1 + \epsilon_1 - \beta^*_0 ; \frac{\rho_1}{M_{\eta}} \right)\right] \\
& \hspace{1cm} -  \min\{\tau^2,(1-\tau)^2\} \frac{\rho_1}{2M_{\eta}} \left( 1- \mmp\left(M^*_ug_1 + \epsilon_1 - \beta^*_0  \in \left[ -\frac{\rho_1}{M_{\eta}}(1-\tau), \frac{\rho_1}{M_{\eta}}\tau \right] \right) \right)\\
&   \leq \mme\left[\ell_{\tau}\left(M^*_ug_1 + \epsilon_1 - \beta^*_0 ; \frac{\rho_1}{M_{\eta}} \right)\right] \\
& \hspace{1cm} -  \min\{\tau^2,(1-\tau)^2\} \frac{\rho_1}{2M_{\eta}} \left( 1- \mmp\left(M^*_ug_1 + \epsilon_1 - \beta^*_0  \in \left[ -\frac{\rho_1}{c_{\eta}}(1-\tau), \frac{\rho_1}{c_{\eta}}\tau \right] \right) \right).
    \end{align*}
Now, let $\rho_1$ be sufficiently small such that $\mmp\left(M^*_ug_1 + \epsilon_1 - \beta^*_0  \in \left[ -\frac{\rho_1}{c_{\eta}}(1-\tau), \frac{\rho_1}{c_{\eta}}\tau \right] \right) \leq 1/2$. Then, the above implies that
\[
 \mme\left[e_{\ell_{\tau}}\left(M^*_ug_1 + \epsilon_1 - \beta^*_0 ; \frac{\rho_1}{M_{\eta}} \right)\right]   \leq \mme\left[\ell_{\tau}\left(M^*_ug_1 + \epsilon_1 - \beta^*_0 ; \frac{\rho_1}{M_{\eta}} \right)\right] - \min\{\tau^2,(1-\tau)^2\} \frac{\rho_1}{4M_{\eta}}.
\]
On the other hand, by part five of Assumption \ref{assump:reg_high_dim_assumptions} we also have that 
\begin{align*}
\frac{d}{dM_{\eta}} \mme\left[ e_{\nu^*}\left( M_{\eta} h_1 + \gamma\frac{\rho_2}{M^*_u} \sqrt{d}\tilde{\beta}_1 ; \frac{\rho_2}{M^*_u} \right)\right] & = \mme\left[ h_1 \partial_x e_{\nu^*}\left( M_{\eta} h_1 + \gamma\frac{\rho_2}{M^*_u} \sqrt{d}\tilde{\beta}_1 ; \frac{\rho_2}{M^*_u} \right)\right]\\
& = M_{\eta} \mme\left[ \partial_x^2 e_{\nu^*}\left( M_{\eta} h_1 + \gamma\frac{\rho_2}{M^*_u} \sqrt{d}\tilde{\beta}_1 ; \frac{\rho_2}{M^*_u} \right)\right].
\end{align*}
Let $c = \min_{c_{\eta} \leq M_{\eta} \leq C_{\eta}, 0 < \rho_2 \leq C_2} M_{\eta}  \mme\left[ \partial_x^2 e_{\nu^*}\left( M_{\eta} h_1 + \gamma\frac{\rho_2}{M^*_u} \sqrt{d}\tilde{\beta}_1 ; \frac{\rho_2}{M^*_u} \right)\right]$ and note that by part 6 of Assumption \ref{assump:reg_high_dim_assumptions} we have $c>0.$ Then, the mean value theorem gives
\[
\mme\left[ e_{\nu^*}\left( M_{\eta} h_1 + \gamma\frac{\rho_2}{M^*_u} \sqrt{d}\tilde{\beta}_1 ; \frac{\rho_2}{M^*_u} \right)\right] \geq \mme\left[ e_{\nu^*}\left(c_{\eta}h_1 + \gamma\frac{\rho_2}{M^*_u} \sqrt{d}\tilde{\beta}_1 ; \frac{\rho_2}{M^*_u} \right)\right] + c (M_{\eta} - c_{\eta}).
\]
Putting this all together, we find that for $\rho_1$ sufficiently close to $0$,
\begin{align*}
      & \max_{(0 < \rho_2 \leq C_2, c_{\eta} \leq M_{\eta} \leq C_{\eta})}  A(\beta_0^*,M_u^*,\rho_1,M_{\eta},\rho_2)\\
      & \leq \max_{(0 < \rho_2 \leq C_2, c_{\eta} \leq M_{\eta} \leq C_{\eta})} \mme\left[\ell_{\tau}\left(M^*_ug_1 + \epsilon_1 - \beta^*_0 \right)\right]  - \gamma \mme\left[ e_{\nu^*}\left(c_{\eta}h_1 +   \gamma\frac{\rho_2}{M^*_u} \sqrt{d}\tilde{\beta}_1 ; \frac{\rho_2}{M^*_u} \right)\right]\\
     & \ \ \  + \gamma \frac{\rho_2}{2M_u^*} \mme[(\sqrt{d}\tilde{\beta}_1)^2] - \frac{M_u^*\rho_2}{2} + \frac{M_{\eta}\rho_1}{2} - \min\{\tau^2,(1-\tau)^2\} \frac{\rho_1}{4M_{\eta}} - c (M_{\eta}-c_{\eta}).
\end{align*}
We claim that for $\rho_1$ sufficiently small the term 
\[
w(M_{\eta}, \rho_1) = \frac{M_{\eta}\rho_1}{2} - \min\{\tau^2,(1-\tau)^2\} \frac{\rho_1}{4M_{\eta}} - c (M_{\eta}-c_{\eta}),
\]
is always negative. Indeed, by our choice of $c_{\eta}$ we have that $\frac{c_{\eta}}{2} - \min\{\tau^2,(1-\tau)^2\} \frac{1}{4c_{\eta}} < 0$. So, we may find $\delta > 0$ such that for all $c_{\eta} \leq M_{\eta} \leq c_{\eta} + \delta$, $\frac{M_{\eta}}{2} - \min\{\tau^2,(1-\tau)^2\} \frac{1}{4M_{\eta}} < 0$, and thus also $w(M_{\eta}, \rho_1) < 0$ for all $\rho_1 > 0$. On the other hand, for $c_{\eta} + \delta \leq M_{\eta} \leq C_{\eta}$ we have
\[
w(M_{\eta}, \rho_1) \leq \rho_1\left( \frac{C_{\eta}}{2} - \min\{\tau^2,(1-\tau)^2\} \frac{1}{4C_{\eta}} \right) - c\delta,
\]
which is negative for $\rho_1$ sufficiently small. This proves the desired claim and thus shows that for all $\rho_1$ sufficiently small,
\begin{align*}
      & \max_{(0 < \rho_2 \leq C_2, c_{\eta} \leq M_{\eta} \leq C_{\eta})}  A(\beta_0^*,M_u^*,\rho_1,M_{\eta},\rho_2) <   \max_{(0 < \rho_2 \leq C_2, c_{\eta} \leq M_{\eta} \leq C_{\eta})} \mme\left[\ell_{\tau}\left(M^*_ug_1 + \epsilon_1 - \beta^*_0 \right)\right]\\
      & \hspace{3cm} - \gamma \mme\left[ e_{\nu^*}\left(c_{\eta}h_1 +   \gamma\frac{\rho_2}{M^*_u} \sqrt{d}\tilde{\beta}_1 ; \frac{\rho_2}{M^*_u} \right)\right] + \gamma \frac{\rho_2}{2M_u^*} \mme[(\sqrt{d}\tilde{\beta}_1)^2] - \frac{M_u^*\rho_2}{2}.
\end{align*}
Comparing the above to our bound in (\ref{eq:lower_A_rho_zero_bound}) for the case $\rho_1 = 0$ we find that 
\[
\max_{(0 < \rho_2 \leq C_2, c_{\eta} \leq M_{\eta} \leq C_{\eta})}  A(\beta_0^*,M_u^*,\rho_1,M_{\eta},\rho_2) < \max_{(0 < \rho_2 \leq C_2, c_{\eta} \leq M_{\eta} \leq C_{\eta})}  A(\beta_0^*,M_u^*,0,M_{\eta},\rho_2).
\]
Thus, $\rho^*_1 \neq 0$, as desired. \\

\end{proof}

\begin{lemma}\label{lem:m_eta_unique}
    Suppose the assumptions of Theorem \ref{thm:main_asym_consistency} hold. Fix any constants $C_{\beta_0},C_u,C_1,C_2, C_{\eta}, c_{\eta} > 0$ with $C_{\eta} > c_{\eta}$ and $c_{\eta} < \sqrt{(1/2)\min\{\tau^2,(1-\tau)^2\}}$. Let $(\beta_0^*,M_u^*,\rho_1^*)$ denote the unique solution to the asymptotic program (\ref{eq:asymp_program}) defined in Lemma \ref{lem:unique_pos_solution} and suppose that $M_u^* >0$. Then, the asymptotic program (\ref{eq:asymp_program}) obtains a unique solution for $M_{\eta}$.
\end{lemma}
\begin{proof}
Since $A$ is jointly convex-concave (Lemma \ref{lem:A_convex_cont}) we know that the function 
\begin{equation}\label{eq:m_eta_opt}
M_{\eta} \mapsto \max_{(0 < \rho_2 \leq C_2)} \min_{(|\beta_0| \leq C_{\beta_0}, 0 \leq M_u \leq C_u, 0 \leq \rho_1 \leq C_1)} A(\beta_0,M_u,\rho_1,M_{\eta},\rho_2),
\end{equation}
is concave on $\mmr_{>0}$ and thus continuous on $[c_{\eta}, C_{\eta}]$. Thus, this function obtains its maximum. 

It remains to show that the maximizer is unique. For ease of notation, let $Z = M^*_ug_1 + \epsilon_1 - \beta^*_0$ and define the function
    \[
    w(M_{\eta}) = \mme\left[ e_{\ell_{\tau}}\left( Z ; \frac{\rho^*_1}{M_{\eta}} \right) \right].
    \]
    Recall that by Lemma \ref{lem:unique_pos_solution} we must have that $\rho_1^* > 0$. We claim that $w$ is strongly concave on $[c_{\eta}, C_{\eta}]$. To see this, note that by a direct calculation using the form of $e_{\ell_{\tau}}$ (see Lemma \ref{lem:pinball_env}), we have 
    \begin{align*}
    w'(M_{\eta}) & = \mme\bigg[ \frac{\tau^2\rho^*_1}{2M_{\eta}^2} \bone\left\{Z > \tau \frac{\rho^*_1}{M_{\eta}}\right\} + \frac{Z^2}{2\rho^*_1}  \bone\left\{-(1-\tau)\frac{\rho^*_{1}}{M_{\eta}} \leq Z \leq \tau \frac{\rho^*_1}{M_{\eta}}\right\}\\
    & \hspace{1cm} + \frac{(1-\tau)^2\rho^*_1}{2M_{\eta}^2} \bone\left\{Z < -(1-\tau)\frac{\rho^*_{1}}{M_{\eta}}\right\} \bigg],
    \end{align*}
    and
    \begin{align*}
    w''(M_{\eta}) & = \mme\bigg[ -\frac{\tau^2\rho^*_1}{M_{\eta}^3} \bone\left\{Z > \tau \frac{\rho^*_{1}}{M_{\eta}}\right\} - \frac{(1-\tau)^2\rho^*_1}{M_{\eta}^3} \bone\left\{Z < -(1-\tau)\frac{\rho^*_{1}}{M_{\eta}}\right\} \bigg].
    \end{align*}
    So, in particular,
    \[
    \sup_{c_{\eta} \leq M_{\eta} \leq C_{\eta} } w''(M_{\eta}) \leq   \mme\bigg[ -\frac{\tau^2\rho^*_1}{C_{\eta}^3} \bone\left\{Z > \tau \frac{\rho^*_{1}}{c_{\eta}}\right\} - \frac{(1-\tau)^2\rho^*_1}{C_{\eta}^3} \bone\left\{Z < -(1-\tau)\frac{\rho^*_{1}}{c_{\eta}}\right\} \bigg] < 0,
    \]
    where the get the last inequality we have applied the fact that $\epsilon_1$ has support on all of $\mmr$ (and thus that $Z$ has support on all of $\mmr$). 

    Now, assume by contradiction that there exist distinct maximizers $M_{\eta}^1$ and $M_{\eta}^2$ for (\ref{eq:m_eta_opt}) on the domain $[c_{\eta},C_{\eta}]$. Since this is a convex-concave problem, we must that that $M_{\eta}^1$ and $M_{\eta}^2$ are maximizers of the function
    \[
    M_{\eta} \mapsto  \max_{0 < \rho_2 \leq C_2} A(\beta_0^*,M_u^*,\rho_1^*,M_{\eta},\rho_2),
    \]
    on $[c_{\eta},C_{\eta}]$. Additionally, note that  $\max_{{c_{\eta} \leq M_{\eta} \leq C_{\eta} }} \max_{0 < \rho_2 \leq C_2} A(\beta_0^*,M_u^*,\rho_1^*,M_{\eta},\rho_2) < \infty$. This follows immediately from the fact that 
    \begin{align*}
    \max_{c_{\eta} \leq M_{\eta} \leq C_{\eta} } \max_{0 < \rho_2 \leq C_2} A(\beta_0^*,M_u^*,\rho_1^*,M_{\eta},\rho_2) & \leq \max_{{c_{\eta} \leq M_{\eta} \leq C_{\eta} }} A(0,0,0,M_{\eta},\rho_2)\\
    & = \mme[\ell_{\tau}(\epsilon_1)] + \gamma \mme[\nu(\gamma \sqrt{d}\tilde{\beta}_1)] < \infty.
    \end{align*}
    Fix $\delta > 0$ small and let $\rho_{2}^1, \rho_2^2 \in (0,C_2]$ be any two values such that
    \[
 \min_{k \in \{1,2\}} A(\beta_0^*,M_u^*,\rho_1^*,M^k_{\eta},\rho^k_2) \geq   \max_{{c_{\eta} \leq M_{\eta} \leq C_{\eta} }} \max_{0 < \rho_2 \leq C_2} A(\beta_0^*,M_u^*,\rho_1^*,M_{\eta},\rho_2) - \delta.
    \]
By the strong concavity of $w(\eta)$ and the joint concavity of the remainder of the terms in $A$ in $(M_{\eta}, \rho_2)$, we have 
\begin{align*}
& \max_{{c_{\eta} \leq M_{\eta} \leq C_{\eta} }} \max_{0 < \rho_2 \leq C_2} A(\beta_0^*,M_u^*,\rho_1^*,M_{\eta},\rho_2)  \geq A\left(\beta_0^*,M_u^*,\rho_1^*, \frac{1}{2}M^1_{\eta} + \frac{1}{2} M^2_{\eta},\frac{1}{2}\rho^1_2 + \frac{1}{2}\rho^2_2 \right)\\
& \geq \frac{1}{2}A(\beta_0^*,M_u^*,\rho_1^*,M^1_{\eta},\rho^1_2) + \frac{1}{2} A(\beta_0^*,M_u^*,\rho_1^*,M^2_{\eta},\rho^2_2) + \frac{ \inf_{c_{\eta} \leq M_{\eta} \leq C_{\eta} } |w''(M_{\eta})|}{8}(M^1_{\eta} - M^2_{\eta})^2\\
& = \max_{{c_{\eta} \leq M_{\eta} \leq C_{\eta} }} \max_{0 < \rho_2 \leq C_2} A(\beta_0^*,M_u^*,\rho_1^*,M_{\eta},\rho_2) - \delta + \frac{ \inf_{c_{\eta} \leq M_{\eta} \leq C_{\eta} } |w''(M_{\eta})|}{8}(M^1_{\eta} - M^2_{\eta})^2,
\end{align*}
as so rearranging,
\[
(M^1_{\eta} - M^2_{\eta})^2 \leq \frac{8\delta}{\inf_{c_{\eta} \leq M_{\eta} \leq C_{\eta} } |w''(M_{\eta})|}.
\]
Sending $\delta \to 0$ gives the desired result.
\end{proof}

Our last result of this section gives a first-order condition for $\rho_1^*$.

\begin{lemma}\label{lem:rho_characterization}
Suppose the conditions of Theorem \ref{thm:main_asym_consistency} hold. Fix any $C_{\beta_0}, C_u, C_{\eta}, c_{\eta}, C_1, C_2 >0$ with $C_{\eta} > c_{\eta}$, $c_{\eta} < \sqrt{1/2)\min\{\tau^2,(1-\tau)^2\}}$, and $C_1 > \sqrt{C_u^2 + \mme[\epsilon_1^2] + C_{\beta_0}^2 }$. Let $(M_u^*, \rho^*_1,\beta_0^*)$ denote the unique optimal solution to the asymptotic program (\ref{eq:asymp_program}) defined in Lemma \ref{lem:unique_pos_solution} and assume that $M_u^* > 0$. Let $M_{\eta}^*$ denote the  unique optimal solution in $M_{\eta}$ of (\ref{eq:asymp_program}) defined in Lemma \ref{lem:m_eta_unique}.  Then, $\rho_1^*$ satisfies the first-order condition
\[
\rho_1^* =  \sqrt{ \mme\left[ \left( M_u^* g + \epsilon_1 - \beta_0^* - \textup{prox}_{\ell_{\tau}}\left( M_u^* g + \epsilon_1 - \beta_0^*;\frac{\rho^*_1}{M^*_{\eta}} \right)\right)^2 \right]}.
\]
    
\end{lemma}
\begin{proof}
    Under the given assumptions,  $\rho_1^*$ minimizes the function
    \[
    w(\rho_1)= \mme\left[e_{\ell_{\tau}}\left( M^*_ug_1 + \epsilon_1 - \beta^*_0 ; \frac{\rho_1}{M_{\eta}^*} \right)\right] + \frac{M_{\eta}^*\rho_1}{2},
    \]
    on the interval $[0,C_1]$. For ease of notation, let $Z = M^*_ug_1 + \epsilon_1 - \beta^*_0$. By, Lemma 15(iii) of \citet{Thram2018},  $\frac{d}{d\rho} e_{\ell_{\tau}}(x;\rho) = - \frac{1}{2\rho^2} (x - \text{prox}_{\ell_{\tau}}(x;\rho))^2$ (see also the calculations in Lemma \ref{lem:pinball_env} below). Applying this fact alongside the dominated convergence theorem gives
    \[
    w'(\rho_1) = -\frac{M^*_{\eta}}{2\rho_1^2}\mme\left[ \left(Z - \text{prox}_{\ell_{\tau}}\left( Z;\frac{\rho_1}{M^*_{\eta}} \right)\right)^2 \right]  + \frac{M_{\eta}^*}{2}.
    \]
    So, 
    \[
    w'(\rho_1) > 0 \iff \rho_1 > \sqrt{ \mme\left[ \left(Z - \text{prox}_{\ell_{\tau}}\left( Z;\frac{\rho_1}{M^*_{\eta}} \right)\right)^2 \right]}.
    \]
    Finally, recall that the function $h(z) = z - \text{prox}_{\ell_{\tau}}\left(z;\rho\right)$ is $1$-Lipschitz (cf. Proposition 12.28 of \citet{Bauschke2017}). Moreover, note that $h(0) = 0$. Thus, the right-hand-side above is at most
    \[
    \sqrt{ \mme\left[ \left(Z - \text{prox}_{\ell_{\tau}}\left( Z;\frac{\rho_1}{M^*_{\eta}} \right)\right)^2 \right]} \leq \sqrt{ \mme\left[ \left(Z\right)^2 \right]} \leq \sqrt{C_u^2 + \mme[\epsilon_1^2] + C_{\beta_0}^2 }.
    \]
    In particular, we find that $w$ is increasing on the interval $(\sqrt{C_u^2 + \mme[\epsilon_1^2] + C_{\beta_0}^2 }, C_1)$. Since we also know that $w$ does not obtain it's minimum at $0$ (Lemma \ref{lem:unique_pos_solution}), we find that $w$ must obtain its minimum inside the open interval $(0,C_1)$. Thus, $w'(\rho_1^*) = 0$, or equivalently,
    \[
    \rho_1^* = \sqrt{ \mme\left[ \left(Z - \text{prox}_{\ell_{\tau}}\left( Z;\frac{\rho^*_1}{M^*_{\eta}} \right)\right)^2 \right]},
    \]
    as claimed.

\end{proof}

\subsection{Final steps}\label{sec:final_theorem_proof}

In this section we prove Theorem \ref{thm:main_asym_consistency}. We begin by stating a convergence result for the primal variables.

\begin{theorem}\label{thm:primal_convergence}
Let $C_u, C_{\beta_0}, C_{\eta}, c_{\eta}, C_1, C_2$ be constants satisfying the conclusions of Lemmas \ref{lem:dual_norm_bounds}, \ref{lem:primal_bound}, \ref{lem:r_s_bounded}, and \ref{lem:swap_to_envelope}. Let $M_u^*$ and $\beta_0^*$ denote the unique solutions for $M_u$ and $\beta_0$ in the asymptotic program (\ref{eq:asymp_program}) defined in Lemma \ref{lem:unique_pos_solution}. Then, under the assumptions of Theorem \ref{thm:main_asym_consistency}, it holds that for all $\delta > 0$,
\[
\mmp\left(\text{For all primal solutions to (\ref{eq:app_min_max_qr}), } |\|\hat{\beta} - \tilde{\beta}\|_2 - M_u^*| < \delta \text{ and } |\hat{\beta}_0 - \beta_0^*| < \delta \right) \to 1.
\]
\end{theorem}
\begin{proof}
    The proof of this result follows similar steps to the proof of Theorem \ref{thm:main_asym_consistency} and, in particular, is very similar to the proof of Proposition \ref{prop:norm_conv} below. Namely, following similar arguments to those presented in Section \ref{sec:auxiliary_reduc} for the dual variables, one can show that to prove this result it is sufficient to bound the value of the program
    \begin{align*}
\phi^{\text{primal}}(S) & := \max_{(\|s\|_2 \leq   C_s \sqrt{n},c_{\eta} \leq M_{\eta} \in \leq C_{\eta})} \min_{(\|r\|_2 \leq C_r \sqrt{n},(\beta_0, u) \in S)}  \min_{(\eta : \|\eta\|_2 = M_{\eta}\sqrt{n})} \bigg( \frac{1}{n} \|u\|_2 \eta^\top g + \frac{1}{n} \|\eta\|_2 u^\top h \\
& \hspace{1.5cm} \frac{1}{n} \eta^\top \epsilon - \frac{1}{n} \beta_0\eta^\top \pmb{1}_n - \frac{1}{n} \eta^\top r + \frac{1}{n} \sum_{i=1}^n \ell_{\tau}(r_i)  + \frac{1}{\sqrt{n}} s^\top(\tilde{\beta} + u) - \frac{1}{n}R^*_d(\sqrt{n}s) \bigg),
\end{align*}
for various choices of $S$. Arguing as above, the values of this program are completely characterized by values of the asymptotic program (\ref{eq:asymp_program}). Convergence of $\|\hat{\beta} - \tilde{\beta}\|_2$ and $\hat{\beta}_0$ then follows from similar arguments to those presented in Pr oposition \ref{prop:norm_conv} below where we show an analogous convergence result for $\|\hat{\eta}\|_2$. Since the details of this proof closely mirror our other arguments, they are omitted.
\end{proof}

We now turn to the proof of Theorem \ref{thm:main_asym_consistency}. Our first result considers the case $M_u^* = 0$.

\begin{proposition}\label{prop:M_u_is_0}
    Suppose the conditions of Theorem \ref{thm:main_asym_consistency} hold.
    Let $(M_u^*,\beta_0^*)$ be defined as in Theorem \ref{thm:primal_convergence} and suppose that $M_u^* = 0$. Let $p := \mmp(\epsilon_1 - \beta^*_0 < 0)$. Then, for all $\xi > 0$, with probability tending to one, all dual solutions $\hat{\eta}$ to \ref{eq:app_min_max_qr} satisfy
    \[
     \left|\frac{1}{n} \sum_{i=1}^n \bone\{\hat{\eta}_i  = -(1-\tau)\} - p \right|,\ \left|\frac{1}{n} \sum_{i=1}^n \bone\{\hat{\eta}_i  = \tau\} - (1-p) \right|,\ \frac{1}{n} \sum_{i=1}^n \bone\{\hat{\eta}_i  \in (-(1-\tau), \tau)\} < \xi.
    \]
    In particular, the result of Theorem \ref{thm:main_asym_consistency} goes through for $P_{\eta} = p \delta_{-(1-\tau)} + (1-p) \delta_{\tau}$.
\end{proposition}
\begin{proof}
    We will focus on the bound on $\frac{1}{n} \sum_{i=1}^n \bone\{\hat{\eta}_i  \in (-(1-\tau), \tau)\}$. The bounds on the other two terms are similar. By the first-order conditions of the optimization in $r$, we have that for any joint primal-dual solution $(\hat{\beta}_0, \hat{\beta}, \hat{\eta})$,
    \[
    \hat{\eta}_i \in \begin{cases}
         \{\tau\},& Y_i > \hat{\beta}_0 + X_i^\top \hat{\beta},\\
         [(1-\tau), \tau] ,& Y_i = \hat{\beta}_0 + X_i^\top \hat{\beta},\\
         \{-(1-\tau)\},& Y_i < \hat{\beta}_0 + X_i^\top \hat{\beta},
    \end{cases}
     = 
     \begin{cases}
        \{\tau\},& \epsilon_i > \hat{\beta}_0 + X_i^\top (\hat{\beta} - \tilde{\beta}),\\
         -[\tau, 1-\tau] ,&  \epsilon_i = \hat{\beta}_0 + X_i^\top (\hat{\beta} - \tilde{\beta}),\\
         \{-(1-\tau)\},& \epsilon_i < \hat{\beta}_0 + X_i^\top (\hat{\beta} -\tilde{\beta}).
    \end{cases}
    \]
    Now, by standard results (e.g. Theorem 3.1 of \citet{Yin1988}) we have that $\sigma_{\max}(X)/\sqrt{n}$ is converging in probability to a constant $c>0$. In particular, this implies that with probability converging to one, $\|X(\hat{\beta} - \tilde{\beta})\|_1 \leq \sqrt{n} \|X(\hat{\beta} - \tilde{\beta})\|_2 \leq n 2c \|\hat{\beta} - \tilde{\beta}\|_2$. So, for any $\rho > 0$, 
    \begin{align*}
     \frac{1}{n} \sum_{i=1}^n \bone\{\hat{\eta}_i  \in (-(1-\tau),\tau)\} & \leq \frac{1}{n} \sum_{i=1}^n \bone\{|\epsilon_i -  \hat{\beta}_0 - X_i^\top (\hat{\beta} - \tilde{\beta})| \leq \rho \}\\
    & \leq \frac{1}{n} \sum_{i=1}^n  \bone\{|\epsilon_i - \hat{\beta}_0| \leq 2\rho\} + \frac{1}{n} \sum_{i=1}^n  \bone\{|\hat{\beta}_0 + X_i^\top (\hat{\beta} - \tilde{\beta})| > \rho\} \\
    & \leq  \frac{1}{n} \sum_{i=1}^n  \bone\{|\epsilon_i -  \hat{\beta}_0 | \leq 2\rho\} + \frac{1}{n \rho} \|\hat{\beta}_0 + X_i^\top (\hat{\beta} - \tilde{\beta})\|_1\\
    & \leq \frac{1}{n} \sum_{i=1}^n  \bone\{|\epsilon_i -  \beta^*_0 | \leq 3\rho\} + \bone\{|\beta^*_0 - \hat{\beta}_0| > \rho\}   + \frac{2c}{\rho} \|\hat{\beta} - \tilde{\beta}\|_2.
    \end{align*}
    So, 
    \begin{align*}
    \sup_{\hat{\eta}} \frac{1}{n} \sum_{i=1}^n \bone\{\hat{\eta}_i  \in (-(1-\tau),\tau)\}  \leq \sup_{\hat{\beta}_0, \hat{\beta}} &  \frac{1}{n} \sum_{i=1}^n  \bone\{|\epsilon_i -  \beta^*_0 | \leq 3\rho\} + \bone\{|\beta^*_0 - \hat{\beta}_0| > \rho\} \\
    & + \frac{2c}{\rho} \|\hat{\beta} - \tilde{\beta}\|_2,
    \end{align*}
    where the suprema are over all dual solutions for $\eta$ and all primal solutions for $(\beta_0,\beta)$, respectively. Applying the law of large numbers and the results of Lemma \ref{lem:primal_bound}, we find that
    \[
    \sup_{\hat{\eta}} \frac{1}{n} \sum_{i=1}^n \bone\{\hat{\eta}_i  \in (-(1-\tau),\tau)\}  \leq \mmp(|\epsilon_1 - \beta_0^*| \leq 3\rho) + o_{\mmp}(1).
    \]
    Since $\epsilon_1$ has a continuous distribution, the desired result follows by sending $\rho \to 0$.
\end{proof}

We now turn to the main proof of Theorem \ref{thm:main_asym_consistency}, which focuses on the more difficult case in which $M_u^* >0$. To begin, we first show that $\|\hat{\eta}_2\|_2$ converges.

\begin{proposition}\label{prop:norm_conv}
Assume the conditions of Theorem \ref{thm:main_asym_consistency} hold. Let $C_u, C_{\beta_0}, C_{\eta}, c_{\eta}, C_1, C_2$ be constants satisfying the conclusions of Lemmas \ref{lem:dual_norm_bounds}, \ref{lem:primal_bound}, \ref{lem:r_s_bounded}, and \ref{lem:swap_to_envelope} as well as the assumptions of Lemma \ref{lem:m_eta_unique}. Let $M_u^*$ denote the unique solution for $M_u$ in the asymptotic program (\ref{eq:asymp_program}) defined in Lemma \ref{lem:unique_pos_solution} and assume that $M_u^* > 0$. Let $M_{\eta}^*$ denote the unique solution for $M_{\eta}$ in (\ref{eq:asymp_program}) defined in Lemma  \ref{lem:m_eta_unique}. Then, for all $\delta > 0$,
    \[
    \mmp\left( \text{For all dual solutions of (\ref{eq:app_min_max_qr}), } |\|\hat{\eta}\|_2 - \sqrt{n}M^*_{\eta}| < \delta \right) \to 1.
    \]
\end{proposition}
\begin{proof}
   Let $V$ denote the value of the asymptotic optimization program $(\ref{eq:asymp_program})$.  By Propositions \ref{prop:original_to_aux} and \ref{prop:aux_to_asymp}, it is sufficient to show that there exists $\xi > 0$ such that with probability converging to one, 
    \[
    \phi(\{\eta : \sqrt{n}c_{\eta} \leq \|{\eta}\|_2 \leq \sqrt{n}C_{\eta}, |\|{\eta}\|_2 - \sqrt{n}M^*_{\eta}| \geq \delta\}) < V - \xi.
    \]
    For ease of notation, let $S_{M,\delta} = \{M \in [c_{\eta}, C_{\eta}]: M \geq  M^*_{\eta} + \delta \text{ or } M \leq M^*_{\eta} - \delta\} $. By a direct calculation following the arguments of Section \ref{sec:auxiliary_reduc}, we have that 
    \begin{align*}
     & \phi(\{\eta : \sqrt{n}c_{\eta} \leq \|{\eta}\|_2 \leq \sqrt{n}C_{\eta}, |\|{\eta}\|_2 - \sqrt{n}M^*_{\eta}| \geq \delta\})\\
     \stackrel{\mmp}{\to} & \max_{(M_{\eta} \in S_{M,\delta}, 0 < \rho_2 \leq C_2)} \min_{(|\beta_0| \leq C_{\beta_0}, 0 \leq M_u \leq C_u, 0 \leq \rho_1 \leq C_1) } A(\beta_0, M_u, \rho_1, M_{\eta}, \rho_2).
    \end{align*}
    By Lemma \ref{lem:m_eta_unique}, this expression is strictly less than $V$, as desired.

\end{proof}

We now prove Theorem \ref{thm:main_asym_consistency}.

\begin{proof}[Proof of Theorem \ref{thm:main_asym_consistency}.]
 Let $C_u, C_{\beta_0}, C_{\eta}, c_{\eta}, C_1, C_2$ be constants satisfying the conclusions of Lemmas \ref{lem:dual_norm_bounds}, \ref{lem:primal_bound}, \ref{lem:r_s_bounded}, and \ref{lem:swap_to_envelope} as well as the assumptions of Lemmas \ref{lem:unique_pos_solution} and \ref{lem:rho_characterization}. Let $(M_u^*, \beta_0^*, \rho_1^*)$ denote the unique solutions for $(M_u,\beta_0,\rho_1)$ in the asymptotic program (\ref{eq:asymp_program}) defined in Lemma \ref{lem:unique_pos_solution}. If $M_u^* = 0$, the result follows from Proposition \ref{prop:M_u_is_0}. So, suppose $M_u^* > 0$. Let $M_{\eta}^*$ denote the unique solution for $M_{\eta}$ in (\ref{eq:asymp_program}) defined in Lemma \ref{lem:m_eta_unique}. Recall that by Lemma \ref{lem:unique_pos_solution} we must have that $\rho_1^* > 0$ and let $P_{\eta}$ denote the distribution of
 \begin{equation}\label{eq:asymptotic_dist}
 \frac{M_{\eta}^*(M_u^* g_1 + \epsilon_1 - \beta_0^* - \text{prox}_{\ell_{\tau}}(M_u^* g_1 + \epsilon_1 - \beta_0^*; \rho_1^*/M_{\eta}^*) )}{\rho_1^*}.
 \end{equation}
 Fix any bounded $L$-lipschitz function $\psi$. Let $V$ denote the value of the asymptotic optimization program $(\ref{eq:asymp_program})$ and fix any $\kappa, \delta > 0$ small. Let $S_{\kappa, \delta}$ denote the set
\[
 \left\{\eta : \max\{c_{\eta}, M_{\eta}^* - \kappa\} \leq \frac{\|{\eta}\|_2}{\sqrt{n}} \leq \min\{C_{\eta}, M_{\eta}^* + \kappa\}, \left|\frac{1}{n} \sum_{i=1}^n \psi(\eta_i) - \mme_{Z \sim P_{\eta}}[\psi(Z)]\right| \geq \delta\right\}. 
\]
By Propositions \ref{prop:original_to_aux},  \ref{prop:aux_to_asymp}, and \ref{prop:norm_conv} it is sufficient to show that there exists $\xi > 0$ such that with probability converging to one, 
    \[
    \phi(S_{\kappa, \delta}) < V - \xi.
    \]
First, note that since $\psi$ is Lipschitz we may assume that $\kappa$ is sufficiently small such that $\eta \in S_{\kappa,\delta} \implies (M_{\eta}^*/\|\eta\|_2)\eta \in S_{0,\delta/2}$, and thus, in particular,
\begin{align*}
    & \phi(S_{\kappa,\delta} )\\
    &  =  \min_{(\|r\|_2 \leq C_r \sqrt{n},|\beta_0| \leq C_{\beta_0}, 0 \leq M_u \leq C_u)}  \max_{(\|s\|_2 \leq  C_s \sqrt{n}, \eta \in S_{\kappa,\delta})}  \bigg( \frac{1}{n} M_u \eta^\top g  - M_u\left\|\frac{1}{n} \|\eta\|_2 h + \frac{1}{n} s \right\|_2\\
& \hspace{1cm} + \frac{1}{n} \eta^\top \epsilon - \frac{1}{n} \beta_0\eta^\top\pmb{1}_n   - \frac{1}{n} \eta^\top r + \frac{1}{n} \sum_{i=1}^n \ell_{\tau}(r_i)  + \frac{1}{\sqrt{n}} s^\top\tilde{\beta}  - \frac{1}{\sqrt{n}}{\mathcal{R}}^*_d(\sqrt{n}s) \bigg)\\
& \leq \phi(S_{0,\delta/2}) + \max_{(\|r\|_2 \leq C_r \sqrt{n},|\beta_0| \leq C_{\beta_0}, 0 \leq M_u \leq C_u, \eta \in S_{\kappa,\delta})} \left|\frac{1}{n}\left(1- \frac{M_{\eta}^*}{\|\eta\|_2}\right)\eta^\top(M_u g + \epsilon - \beta_0\pmb{1}_n - r)\right|\\
& \leq \phi(S_{0,\delta/2}) + \kappa \frac{C_u\|g\|_2 + \|\epsilon\|_2 + C_{\beta_0} \sqrt{n} + C_r\sqrt{n}}{\sqrt{n}}\\
& = \phi(S_{0,\delta/2}) +\kappa O_{\mmp}(1).
\end{align*}
Moreover,
\begin{align*}
\phi(S_{0,\delta/2})\\
&  =  \min_{(\|r\|_2 \leq C_r \sqrt{n},|\beta_0| \leq C_{\beta_0}, 0 \leq M_u \leq C_u)}  \max_{(\|s\|_2 \leq  C_s \sqrt{n}, \eta \in S_{0,\delta/2})}  \bigg( \frac{1}{n} M_u \eta^\top g  - M_u\left\|\frac{1}{n} \|\eta\|_2 h + \frac{1}{n} s \right\|_2\\
& \hspace{1cm} + \frac{1}{n} \eta^\top \epsilon - \frac{1}{n} \beta_0\eta^\top\pmb{1}_n   - \frac{1}{n} \eta^\top r + \frac{1}{n} \sum_{i=1}^n \ell_{\tau}(r_i)  + \frac{1}{\sqrt{n}} s^\top\tilde{\beta}  - \frac{1}{n}{\mathcal{R}}^*_d(\sqrt{n}s) \bigg)\\
& \leq \min_{(\|r\|_2 \leq C_r \sqrt{n})} \max_{(\|s\|_2 \leq  C_s \sqrt{n}, \eta \in S_{0,\delta/2})}  \bigg(\frac{1}{n} M^*_u \eta^\top g  - M^*_u\left\|\frac{1}{n} M^*_{\eta} h + \frac{1}{n} s \right\|_2 + \frac{1}{n} \eta^\top \epsilon \\
& \hspace{2cm} - \frac{1}{n} \beta^*_0 \eta^\top\pmb{1}_n   - \frac{1}{n} \eta^\top r + \frac{1}{n} \sum_{i=1}^n \ell_{\tau}(r_i)  + \frac{1}{\sqrt{n}} s^\top\tilde{\beta}  - \frac{1}{n}{\mathcal{R}}^*_d(\sqrt{n}s)\bigg)\\
& =  \min_{(\|r\|_2 \leq C_r \sqrt{n})} \max_{(\eta \in S_{0,\delta/2})}\bigg( \frac{1}{n}\eta^\top(M_u^*g + \epsilon - \beta_0^*\pmb{1}_n - r) + \frac{1}{n}\sum_{i=1}^n \ell_{\tau}(r_i)\bigg) \numberthis \label{eq:first_term_eta_opt_theorem}\\
&  \hspace{2cm} + \max_{\|s\|_2 \leq C_s\sqrt{n}}  - M^*_u\left\|\frac{1}{n} M^*_{\eta} h + \frac{1}{n} s \right\|_2  + \frac{1}{\sqrt{n}} s^\top \tilde{\beta} - \frac{1}{n} {\mathcal{R}}^*_d(\sqrt{n}s).
\end{align*}
Arguing as in Section \ref{sec:auxiliary_reduc} (and in particular applying Lemmas \ref{lem:swap_to_envelope} and \ref{lem:convexity_lemma} along with the law of large numbers), the second term converges as
\begin{equation}\label{eq:asymp_second_term}
\begin{split}
& \max_{\|s\|_2 \leq C_s\sqrt{n}}  - M^*_u\left\|\frac{1}{n} M^*_{\eta} h + \frac{1}{n} s \right\|_2  + \frac{1}{\sqrt{n}} s^\top \tilde{\beta} - \frac{1}{n} {\mathcal{R}}^*_d(\sqrt{n}s)\\
 \stackrel{\mmp}{\to} & \max_{0  < \rho_2 \leq C_2} - \frac{(M_{\eta}^*)^2M_u^*\gamma}{2} + \gamma\mme\left[e_{\nu}\left(\frac{M_{\eta}^*M_u^*}{\rho_2}h + \gamma \sqrt{d} \tilde{\beta} ; \frac{M^*_u}{\rho_2} \right) \right]- \frac{M_u^*\rho_2}{2}.
 \end{split}
\end{equation}

It remains to consider the first term. Let $r^* \in \mmr^n$ be the vector given by $r^*_i = \text{prox}_{\ell_{\tau}}(M_u^* g_i + \epsilon_i - \beta_0^*; \rho_1^*/M_{\eta}^*)$. By the law of large numbers, we have
\[
 \frac{1}{n} \sum_{i=1}^n \psi\left(\frac{M_{\eta}^*(M_u^* g_i + \epsilon_i - \beta_0^* - r^*_i )}{\rho_1^*} \right) \stackrel{\mmp}{\to} \mme_{Z \sim P_{\eta}}[Z],
\]
and that 
\[
\frac{\|M_u^* g_i + \epsilon_i - \beta_0^* - r^*_i \|_2}{\sqrt{n}} \stackrel{\mmp}{\to} \rho_1^*,
\]
 where we recall that by Lemma \ref{lem:rho_characterization}, $\rho_1^* = \sqrt{\mme_{Z \sim P_{\eta}}[Z^2]} $.   Since $\psi$ is $L$-Lipschitz, this implies that
\begin{align*}
& \liminf_{n \to \infty} \min_{\eta \in S_{0,\delta/2}} \frac{1}{\sqrt{n}} \left\|\eta - \frac{M_{\eta}^*(M_u^* g_i + \epsilon_i - \beta_0^* - r^*_i )}{\rho_1^*} \right\|_2\\
& \geq \liminf_{n \to \infty}  \min_{\eta \in S_{0,\delta/2}} \frac{1}{L} \left| \frac{1}{n} \sum_{i=1}^n \psi\left(\frac{M_{\eta}^*(M_u^* g_i + \epsilon_i - \beta_0^* - r^*_i )}{\rho_1^*} \right) - \frac{1}{n} \sum_{i=1}^n \psi(\eta_i)\right| \stackrel{\mmp}{\geq} \frac{\delta}{2L}.
\end{align*}
For ease of notation, let $Z^* = M_u^*g + \epsilon - \beta_0^*\pmb{1}_n - r^*$. Applying these calculations, we find that the optimization appearing on line (\ref{eq:first_term_eta_opt_theorem}) can be bounded as
\begin{align*}
     &  \min_{(\|r\|_2 \leq C_r \sqrt{n})} \max_{(\eta \in S_{0,\delta/2})}\bigg( \frac{1}{n}\eta^\top(M_u^*g + \epsilon - \beta_0^*\pmb{1}_n - r) + \frac{1}{n}\sum_{i=1}^n \ell_{\tau}(r_i)\bigg)\\
     & \leq  \max_{(\eta \in S_{0,\delta/2})}\bigg( \frac{1}{n}\eta^\top Z^* + \frac{1}{n}\sum_{i=1}^n \ell_{\tau}(r^*_i)\bigg)\\
     & = \max_{(\eta \in S_{0,\delta/2})}\bigg( \frac{1}{n}\frac{\eta^\top Z^*}{M_{\eta}^* \|Z^*\|_2} M_{\eta}^*\|Z^*\|_2\bigg)  + \mme\left[ \ell_{\tau}\left( \text{prox}_{\ell_{\tau}}\left(M_u^* g_1 + \epsilon_1 - \beta_0^*; \frac{\rho_1^*}{M_{\eta}^*}\right) \right)\right] + o_{\mmp}(1)\\
     & = \max_{(\eta \in S_{0,\delta/2})} \left(1 - \frac{1}{2} \left\|\frac{1}{\sqrt{n}}\frac{\eta}{M_{\eta}^* } - \frac{Z^*}{\|Z^*\|_2}  \right\|_2^2 \right)\frac{M_{\eta}^* \|Z^*\|_2}{\sqrt{n}}\\
     & \hspace{4cm} + \mme\left[ \ell_{\tau}\left( \text{prox}_{\ell_{\tau}}\left(M_u^* g_1 + \epsilon_1 - \beta_0^*; \frac{\rho_1^*}{M_{\eta}^*} \right)\right)\right] + o_{\mmp}(1) \\
     & \leq \left( 1- \frac{\delta^2}{8 L^2 (M_{\eta}^*)^2} \right)M_{\eta}^* \rho_1^* + \mme\left[ \ell_{\tau}\left( \text{prox}_{\ell_{\tau}}\left(M_u^* g_1 + \epsilon_1 - \beta_0^*; \frac{\rho_1^*}{M_{\eta}^*} \right)\right)\right] + o_{\mmp}(1)\\
     & = \frac{M_{\eta}^*\rho_1^*}{2} + \mme\left[ e_{\ell_{\tau}}\left( M_u^* g_1 + \epsilon_1 - \beta_0^*; \frac{\rho_1^*}{M_{\eta}^*} \right)\right] - \frac{\delta^2}{8 L^2 (M_{\eta}^*)^2}  + o_{\mmp}(1),
\end{align*}
where the last line applies the formula for $\rho_1^*$ given in Lemma \ref{lem:rho_characterization} alongside the definition of the Moreau envelope. Combining this with (\ref{eq:asymp_second_term}), we conclude that 
\begin{align*}
\phi(S_{\kappa,\delta}) \leq   V -  \frac{\delta^2}{8 L^2 (M_{\eta}^*)^2} + \kappa O_{\mmp}(1) + o_{\mmp}(1).
\end{align*}
Sending $\kappa \to 0$ gives the desired result.

\end{proof}

\subsection{Corollaries of Theorem \ref{thm:main_asym_consistency}}\label{sec:corollaries}

We now prove Corollaries  \ref{corr:loo_cov_consistency} and \ref{corr:quantile_consistency}.

\begin{proof}[Proof of Corollary \ref{corr:loo_cov_consistency}]
For all $i \in \{1,\dots,n\}$, let $(\hat{\beta_0}^{(-i)},\hat{\beta}^{(-i)})$ denote a leave-one-out solution to the quantile regression when the $i_{\text{th}}$ sample is omitted from the fit and suppose that this solution is chosen such that $(\hat{\beta_0}^{(-i)},\hat{\beta}^{(-i)}) \independent (X_i,Y_i)$. By Proposition \ref{prop:initial_loo_comparison}, we have that for all dual solutions $\hat{\eta}$, 
\[
\frac{1}{n} \sum_{i=1}^n \bone\{Y_i < \hat{\beta}^{(-i)}_0 +  X_i^\top \hat{\beta}^{(-i)}\} \leq \frac{1}{n} \sum_{i=1}^n \bone\{\hat{\eta}_i \leq 0\},
\]
and 
\[
\frac{1}{n} \sum_{i=1}^n \bone\{\hat{\eta}_i < 0 \} \leq  \frac{1}{n} \sum_{i=1}^n \bone\{Y_i \leq \hat{\beta}^{(-i)}_0 +  X_i^\top \hat{\beta}^{(-i)}\}.
\]
Moreover, since the distribution of $Y_i \mid X_i$ is continuous, we must have that $\mmp(Y_i = \hat{\beta}_0^{(-i)} + X_i^\top \hat{\beta}^{(-i)}) = 0$. So, combining the above, we find that with probability one all dual solutions are such that
\begin{equation}\label{eq:dual_loo_comparison_corr}
\left| \frac{1}{n} \sum_{i=1}^n \bone\{Y_i \leq \hat{\beta}^{(-i)}_0 +  X_i^\top \hat{\beta}^{(-i)}\} - \frac{1}{n} \sum_{i=1}^n \bone\{\hat{\eta}_i \leq 0\}  \right| \leq \frac{1}{n} \sum_{i=1}^{n} \bone\{\hat{\eta}_i = 0\}.
\end{equation}
By Theorem \ref{thm:main_asym_consistency} and the continuity of the distribution $P_{\eta}$ at $0$ it is straightforward to show that
    \[
    \forall \delta > 0,\ \mmp\left(\text{For all dual solutions $\hat{\eta}$, }  \frac{1}{n} \sum_{i=1}^{n} \bone\{\hat{\eta}_i = 0\} \leq \delta  \right) \to 1,
    \]
    and 
    \begin{equation}\label{eq:dual_cutoff_limit}
     \forall \delta > 0,\ \mmp\left(\text{For all dual solutions $\hat{\eta}$, }  \left| \frac{1}{n} \sum_{i=1}^{n} \bone\{\hat{\eta}_i \leq 0\} - \mmp_{Z \sim P_{\eta}}(Z \leq 0) \right| \leq \delta  \right) \to 1.
    \end{equation}
Combining these facts with (\ref{eq:dual_loo_comparison_corr}) gives, in particular, that
\[
 \frac{1}{n} \sum_{i=1}^n \bone\{Y_i \leq \hat{\beta}^{(-i)}_0 +  X_i^\top \hat{\beta}^{(-i)}\}  \stackrel{\mmp}{\to} \mmp_{Z \sim P_{\eta}}(Z \leq 0).
\]
Since this random variable is bounded, we then also have that
\[
 \mmp( Y_1 \leq \hat{\beta}^{(-1)}_0 +  X_1^\top \hat{\beta}^{(-1)} ) = \mme\left[\frac{1}{n} \sum_{i=1}^n \bone\{Y_i \leq \hat{\beta}^{(-i)}_0 +  X_i^\top \hat{\beta}^{(-i)}\} \right] \to \mmp_{Z \sim P_{\eta}}(Z \leq 0),
\]
or equivalently, that  $\mmp( Y_{n+1} \leq \hat{\beta}_0 +  X_{n+1}^\top \hat{\beta} ) \to \mmp_{Z \sim P_{\eta}}(Z \leq 0)$. Combing this fact with (\ref{eq:dual_cutoff_limit}) gives the desired result.

\end{proof}

\begin{proof}[Proof of Corrolary \ref{corr:quantile_consistency}]
    Let $(C_u, C_{\beta_0}, C_{\eta}, c_{\eta}, C_1, C_2)$ the conclusions of Lemmas \ref{lem:dual_norm_bounds}, \ref{lem:primal_bound}, \ref{lem:r_s_bounded}, and \ref{lem:swap_to_envelope} as well as the assumptions of Lemmas \ref{lem:unique_pos_solution} and \ref{lem:rho_characterization}. Let $(M_u^*, \beta_0^*, \rho_1^*)$ denote the unique solution in $(M_u,\beta_0,\rho_1)$ to the asymptotic program (\ref{eq:asymp_program}) defined in Lemma \ref{lem:unique_pos_solution}. 
    
    To begin, we will first show that the unregularized quantile regression program must have $M_u^* > 0$. Let $(\hat{\beta}_0, \hat{\beta}, \hat{r}, \hat{\eta})$ denote any primal-dual solutions to the quantile regression (\ref{eq:app_min_max_qr}). The first-order conditions of this optimization in $r$ imply that $\hat{\eta} \in [-(1-\tau),\tau]^n$. By Proposition \ref{prop:M_u_is_0}, if $M_u^* = 0$ we must have that with probability converging to one,
    \[
    \frac{1}{n} \sum_{i=1}^n \bone\{\hat{\eta}_i \in (-(1-\tau),\tau) \} < d,
    \]
    We will show that this is not possible. 
    
    Introduce the notation $X_A$ to denote the submatrix of $X$ consisting of the rows in $A \subseteq \{1,\dots,n\}$ and $X_{A,B}$ to denote the submatrix with rows in $A \subseteq \{1,\dots,n\}$ and columns in $B \subseteq \{1,\dots,d\}$. Let $\hat{\eta}_A$ denote the subvector of $\hat{\eta}$ with entries in $A$ and $I_{\text{int.}} = \{i \in \{1,\dots,n\} : -(1-\tau) < \hat{\eta}_i < \tau \}$ denote the set of entries of $\hat{\eta}$ which lie in the interior. By the first-order conditions of (\ref{eq:asymp_program}) in $\beta$, we have that 
    \begin{equation}\label{eq:no_reg_eta_cond}
    \hat{\eta}^\top X = 0 \iff \hat{\eta}_{I_{\text{int.}}^c}^\top X_{I_{\text{int.}}^c} =  \hat{\eta}_{I_{\text{int.}}}^\top X_{I_{\text{int.}}}.
    \end{equation}
    On the other hand, for any fixed set $\tilde{I} \subseteq \{1,\dots,n\}$ with $|\tilde{I}| < d-1$ and vector $v \in \{-(1-\tau),\tau\}^{n-|\tilde{I}|}$ we have that with probability one $v^T X_{\tilde{I}^c}$ is not in the row space of $X_{\tilde{I}}$. This follows immediately from the fact that for any $u \in \mmr^{|\tilde{I}|}$,
    \[
    u^\top X_{\tilde{I}} = v^\top X_{\tilde{I}^c} \implies  v^\top X_{\tilde{I}^c,\{1,\dots,|\tilde{I}|\}} (X^\top_{\tilde{I},\{1,\dots,|\tilde{I}|\}})^{-1} X_{\tilde{I},\{|\tilde{I}| + 1,\dots,d\}} =   v^\top X_{\tilde{I}^c, \{|\tilde{I}| + 1,\dots,d\}},
    \]
    which occurs with probability zero since $v^\top X_{I^c, \{|\tilde{I}| + 1,\dots,d\}}$ is a continuously distributed random vector independent of $v^\top X_{\tilde{I}^c,\{1,\dots,|\tilde{I}|\}} (X^\top_{\tilde{I},\{1,\dots,|\tilde{I}|\}})^{-1} X_{\tilde{I},\{|\tilde{I}| + 1,\dots,d\}}$. Taking a union bound over all choices of $\tilde{I}$ and $v$ and applying  (\ref{eq:no_reg_eta_cond}), we find that with probability one, $|I_{\text{int.}}| > d-1$. As discussed above, this implies that $M_u^* > 0$, as claimed.

    We are now ready to prove the main result of Corollary (\ref{corr:quantile_consistency}). Fix any $\delta > 0$. Let $q^*$ denote the $\tau$ quantile of the asymptotic distribution $P_{\eta}$ defined in (\ref{eq:asymptotic_dist}). We will show that with probability converging to one the empirical quantile of $\hat{\eta}$ lies below $q^* + 2\delta$. Proof of a matching lower bound is identical. If $q^* \geq \tau - 2\delta$, then the result is immediate. So, suppose that $q^* < \tau - 2\delta$. Let $\psi_{\delta}$ be the step function
    \[
    \psi_{\delta}(x) = \begin{cases}
        0,\ & x > q^* + 2\delta,\\
       \frac{q^*+2\delta - x}{\delta},\ & q^* + \delta \leq x \leq q^* + 2\delta\\
       1,\ & x < q^* + \delta.
    \end{cases} 
    \]
    Fix a small value $\xi > 0$ to be specified shortly. By Theorem \ref{thm:main_asym_consistency}, we have that with probability converging to one all dual solutions satisfy 
    \begin{align*}
    \frac{1}{n} \sum_{i=1}^n \bone\{\hat{\eta}_i \leq q^* + 2\delta\} & \geq \frac{1}{n} \sum_{i=1}^n \psi_{\delta}(\hat{\eta}_i) \geq \mme_{Z \sim P_{\eta}}[\psi_{\delta}(Z)] - \xi\\
    & \geq \mmp_{Z \sim P_{\eta}}(Z \leq q^*) + \mmp_{Z \sim P_{\eta}}(q^*  < Z \leq q^* + \delta) - \xi.
    \end{align*}
    Since $M_u^* > 0$, we must have that $\rho_1^* > 0$ and thus that $P_{\eta}$ has point masses at $-(1-\tau)$ and $\tau$ and a continuous distribution with positive density on $(-(1-\tau),\tau)$. In particular, by choosing $\xi$ sufficiently small we may guarantee that with probability converging to one, all dual solutions satisfy 
    \[
    \frac{1}{n} \sum_{i=1}^n \bone\{\hat{\eta}_i \leq q^* + 2\delta\}  \geq  \mmp_{P_{\eta}}(Z \leq q^*) + \mmp_{P_{\eta}}(q^*  < Z \leq q^* + \delta) - \xi \geq \tau,
    \]
    and thus that
    \[
    \text{Quantile}\left(\tau, \frac{1}{n} \sum_{i=1}^n \delta_{\hat{\eta}_i} \right) \leq q^* + 2\delta,
    \]
    as claimed.
\end{proof}

\section{Additional technical lemmas}

In this section, we give a number of auxiliary results that are useful in the main proofs. 

\begin{lemma}\label{lem:gaussian_abs_lower}
Let $\{X_i\}_{i=1}^n \stackrel{i.i.d.}{\sim} \mathcal{N}(0,I_d)$. Then, as $d/n \to \gamma  \in [0,\infty)$,
\[
\liminf_{d,n \to \infty} \inf_{(\|u\|_2 \leq 1, |\beta_0| \leq 1, \max\{\|u\|_2, |\beta_0|\} = 1)} \frac{1}{n} \sum_{i=1}^n |X_i^\top u + \beta_0| \stackrel{\mmp}{\geq} \sqrt{\frac{2}{\pi}} - \sqrt{\gamma}. 
\]
\end{lemma}
\begin{proof}
Let $X \in \mmr^{n \times d}$ denote the matrix with rows $X_1,\dots,X_n$. Write 
    \begin{align*}
    & \inf_{(\|u\|_2 \leq 1, |\beta_0| \leq 1, \max\{\|u\|_2, |\beta_0|\} = 1)}  \frac{1}{n} \sum_{i=1}^n |X_i^\top u + \beta_0|\\
    & =   \inf_{(\|u\|_2 \leq 1, |\beta_0| \leq 1, \max\{\|u\|_2, |\beta_0|\} = 1)} \max_{(v \in \{\pm 1\}^n)}\frac{1}{n} v^\top X u + \frac{1}{n}\beta_0 v^\top \pmb{1}_n .
    \end{align*}
    By the Gordon's inequality (Proposition \ref{prop:gordon} above), we have that for any $c \in \mmr$,
    \begin{align*}
    & \mmp\left( \inf_{(\|u\|_2 \leq 1, |\beta_0| \leq 1, \max\{\|u\|_2, |\beta_0|\} = 1)} \max_{(v \in \{\pm 1\}^n)}  \frac{1}{n} v^\top X u + \frac{1}{n} \beta_0 v^\top \pmb{1}_n  < c \right) \numberthis \label{eq:gaussian_compare_sign_vec}\\
    & \leq  2\mmp\left(  \inf_{(\|u\|_2 \leq 1, |\beta_0| \leq 1, \max\{\|u\|_2, |\beta_0|\} = 1)} \max_{(v \in \{\pm 1\}^n)} \frac{1}{n} \|v\|_2 u^\top h +\frac{1}{n} \|u\|_2 v^\top g + \frac{1}{n} \beta_0 v^\top \pmb{1}_n \leq c  \right),
    \end{align*}
    where $h \sim \mathcal{N}(0,I_d)$ and $g \sim \mathcal{N}(0,I_n)$ are independent. Now, 
    \begin{align*}
         & \inf_{(\|u\|_2 \leq 1, |\beta_0| \leq 1, \max\{\|u\|_2, |\beta_0|\} = 1)} \max_{(v \in \{\pm 1\}^n)} \frac{1}{n} \|v\|_2 u^\top h + \frac{1}{n}\|u\|_2 v^\top g + \frac{1}{n} \beta_0 v^\top \pmb{1}_n \\
         &  =  \inf_{(\|u\|_2 \leq 1, |\beta_0| \leq 1, \max\{\|u\|_2, |\beta_0|\} = 1)}  \frac{1}{\sqrt{n}} u^\top h +  \frac{1}{n} \sum_{i=1}^n | \|u\|_2 g_i + \beta_0| \\
         & = \inf_{(0 \leq c_u \leq 1, \beta_0 \leq 1, \max\{c_u, \beta_0\} = 1)} \frac{1}{n} \sum_{i=1}^n | c_u g_i + \beta_0|  - c_u \frac{\|h\|_2}{\sqrt{n}}\\
         & \stackrel{\mmp}{\to} \inf_{(0 \leq c_u \leq 1, \beta_0 \leq 1, \max\{c_u, \beta_0\} = 1)} \mme[|c_u g_1 + \beta_0|] - c_u \sqrt{\gamma},
    \end{align*}
    where the limit follows from standard uniform concentration arguments (e.g. Lemma 7.75 of \citet{liese2008}). 
    
    Finally, note that for any $c_u$, $\beta_0 \mapsto \mme[|c_u g_1 + \beta_0|]$ is a convex, even function and thus obtains its minimum at $0$. So, 
    \[
     \inf_{(c_u = 1, |\beta_0| \leq 1)} \mme[|c_u g_1 + \beta_0|] - c_u \sqrt{\gamma} = \mme[|g_1|] - \sqrt{\gamma} = \sqrt{\frac{2}{\pi}} - \sqrt{\gamma}.
    \]
    On the other hand, by Jensen's inequality,
    \[
    \inf_{(0 \leq c_u \leq 1, |\beta_0| = 1)} \mme[|c_u g_1 + \beta_0|] - c_u \sqrt{\gamma}  \geq  \inf_{(0 \leq c_u \leq 1)} |c_u \mme[g_1] + 1]| - c_u \sqrt{\gamma} = 1 - \sqrt{\gamma}.
    \]
    Combining the above, we conclude that
    \[
    \inf_{(\|u\|_2 \leq 1, |\beta_0| \leq 1, \max\{\|u\|_2, |\beta_0|\} = 1)} \max_{(v \in \{\pm 1\}^n)} \frac{1}{n} \|v\|_2 u^\top h + \frac{1}{n}\|u\|_2 v^\top g + \frac{1}{n} \beta_0 v^\top \pmb{1}_n \stackrel{\mmp}{\to} \sqrt{\frac{2}{\pi}} - \sqrt{\gamma},
    \]
    and applying (\ref{eq:gaussian_compare_sign_vec}) gives the desired result.
\end{proof}

Our next lemma gives sufficient conditions under which partial optimization preserves strict convexity.

\begin{lemma}[Lemma 19 of \citet{Thram2018}]\label{lem:convex_under_opt}
    Let $\mathcal{A}$ and $\mathcal{B}$ be convex sets and $\Psi: \mathcal{A} \times \mathcal{B} \to \mmr$ be strictly convex in its first argument. Assume that $\Psi(a,\cdot)$ obtains its maximum for all $a \in \mathcal{A}$. Then, $a \mapsto \max_{b \in \mathcal{B}} \Psi(a,b)$ is strictly convex.
\end{lemma}

Our next result computes the Moreau envelope of the pinball loss. 

\begin{lemma}\label{lem:pinball_env}
    For any $x \in \mmr$ and $\rho \geq 0$ the Moreau envelope of the pinball loss is given by
    \[
    e_{\ell_{\tau}}(x;\rho) = \begin{cases}
         \frac{\tau^2\rho}{2} + \tau (x-\rho\tau) ,\ & x-\rho\tau > 0,\\
         \frac{x^2}{2\rho},\ & x \in [- \rho(1-\tau),\rho\tau],\\
         \frac{(1-\tau)^2\rho}{2} - (1-\tau)(x+\rho(1-\tau))  ,\ & x+\rho(1-\tau) < 0.
    \end{cases}
    \]
\end{lemma}
\begin{proof}
    The case $\rho = 0$ is given by the continuous extension stated in Lemma \ref{lem:envelope_extension}. So, consider the case $\rho > 0$. We begin by computing the proximal function. Let $f(v) = \frac{1}{2\rho}(v- x)^2 + \ell_{\tau}(v)$ denote the objective appearing in the definition of the Moreau envelope and the proximal function. We have that 
    \[
    \partial f(v) = \begin{cases}
         \{\frac{v-x}{\rho} + \tau\},\ & v > 0,\\
         [\frac{v-x}{\rho} -(1-\tau), \frac{v-x}{\rho} + \tau],\ & v = 0,\\
         \{\frac{v-x}{\rho} - (1-\tau)\},\ & v < 0.
    \end{cases}
    \]
    Setting this to zero, we find that 
    \[
    \text{prox}_{\ell_{\tau}}(x;\rho) = \begin{cases}
         x-\rho\tau ,\ & x-\rho\tau > 0,\\
         0,\ & x \in [- \rho(1-\tau),\rho\tau],\\
         x+\rho(1-\tau) ,\ & x + \rho(1-\tau) < 0.
    \end{cases}
    \]
    Plugging this into the definition of the Moreau envelope gives the desired result.
\end{proof}

The next two lemmas useful facts from convex analysis that were applied in the proofs above.

\begin{lemma}[Part i of Theorem 14.3 in \citet{Bauschke2017}]\label{lem:envelope_identity}
     Let $f : \mmr^k \to \mmr \cup \{+\infty\}$ be a proper, lower semicontinuous convex function. Then, for any $x \in \mmr^k$ and $\rho > 0$ we have the identity
    \[
    e_f(x;\rho) + e_{f^*}(x/\rho;1/\rho) = \frac{\|x\|^2}{2\rho}.
    \]
    
\end{lemma}

\begin{lemma}[Corollary of Theorem 1.25 in \citet{Rockefeller1997}]\label{lem:envelope_extension}
    Let $f :\mmr^k \to \mmr$ be convex. Then for all $x \in \mmr^k$,
    \[
    \lim_{\rho \downarrow 0} e_f(x;\rho) = f(x).
    \]
\end{lemma}

\end{document}